\documentclass[transmag]{IEEEtran}

\usepackage{latexsym}
\usepackage{hyperref}
\usepackage[utf8]{inputenc}
\usepackage[english]{babel}
\usepackage{cite}
\usepackage{amsmath,amssymb,amsfonts}
\usepackage{subfig}    
\usepackage[demo]{graphicx}
\usepackage{textcomp}
\usepackage{xcolor}
\usepackage{mathrsfs}
\usepackage[utf8]{inputenc}
\usepackage{comment}
\usepackage[ruled,vlined]{algorithm2e}
\usepackage{mathtools}
\usepackage[font=small,labelfont=bf]{caption}
\usepackage{blkarray, bigstrut}
\usepackage{subfig}
\usepackage{cases}

\usepackage{enumerate}

\usepackage{lipsum}

\SetCommentSty{mycommfont}
\SetKwInput{KwInput}{Input}                
\SetKwInput{KwOutput}{Output}              

\usepackage{enumitem}
\usepackage{float}

\usepackage{pgf, tikz}
\usetikzlibrary{arrows, automata}
\usepackage{tikz}
\usetikzlibrary{automata,positioning}

\newcommand*\circled[2][1.6]{\tikz[baseline=(char.base)]{
    \node[shape=circle, draw, inner sep=1pt, 
        minimum height={\f@size*#1},] (char) {\vphantom{WAH1g}#2};}}

\newlist{legal}{enumerate}{10}
\setlist[legal]{label*=\arabic*.}
\usepackage{ifpdf}

\usepackage{amsthm}

\usepackage{booktabs}

\newtheoremstyle{mystyle1}
  {.5pt}
  {.5pt}
  {}
  {}
  {\bfseries}
  {.}
  { }
  {\thmname{#1}\thmnumber{ #2}\thmnote{ (#3)}}

\theoremstyle{mystyle1}

\newtheorem{cor}{Corollary}
\newtheorem{defn}{Definition}

\newtheorem{rem}{Remark}
\newtheorem{exmp}{Example}

\newtheoremstyle{mystyle}
  {.5pt}
  {.5pt}
  {\itshape}
  {}
  {\bfseries}
  {.}
  { }
  {\thmname{#1}\thmnumber{ #2}\thmnote{ (#3)}}

 \theoremstyle{mystyle}

\newtheorem{thm}{Theorem}
\newtheorem{lem}{Lemma}
\newtheorem{prop}{Proposition}

\usepackage[utf8]{inputenc}
\usepackage[english]{babel}
\usepackage{cite}
\usepackage{amsmath,amssymb,amsfonts}
\usepackage{algorithmic}
\usepackage{graphicx}
\usepackage{textcomp}
\usepackage{xcolor}
\usepackage{mathrsfs} 
\usepackage{enumitem}
\usepackage{amsthm}
\usepackage{cite}
\usepackage{array}
\usepackage{amsmath}

\setlist[legal]{label*=\arabic*.}

\theoremstyle{mystyle1}

\usepackage{enumitem}
\usepackage{geometry}
\usepackage{hyperref}
\hypersetup{
    colorlinks=true,
    linkcolor=blue,
    filecolor=blue,      
    urlcolor=blue,
    citecolor=blue
}

\allowdisplaybreaks

\begin{document}
  \title{On the Optimality of Linear Index Coding over the Fields with Characteristic Three}

\author{Arman Sharififar, Parastoo Sadeghi, Neda Aboutorab
\\
\textit{University of New South Wales, Australia} \\
Email:\{a.sharififar, p.sadeghi, n.aboutorab\}@unsw.edu.au}
\maketitle

\begin{abstract}
    It has been known that the insufficiency of linear coding in achieving the optimal rate of the general index coding problem is rooted in its rate's dependency on the field size. However, this dependency has been described only through the two well-known matroid instances, namely the Fano and non-Fano matroids, which, in turn, limits its scope only to the fields with characteristic two. In this paper, we extend this scope to demonstrate the reliance of linear index coding rate on fields with characteristic three. By constructing two index coding instances of size 29, we prove that for the first instance, linear coding is optimal only over the fields with characteristic three, and for the second instance, linear coding over any field with characteristic three can never be optimal. Then, a variation of the second instance is designed as the third index coding instance of size 58. For this instance, it is proved that while linear coding over any field with characteristic three cannot be optimal, there exists a nonlinear code over the fields with characteristic three, which achieves its optimal rate. Connecting the first and third index coding instances in two specific ways, called no-way and two-way connections, will lead to two new index coding instances of size 87 and 91, for which linear coding is outperformed by nonlinear codes.
    Another main contribution of this paper is the reduction of the key constraints on the space of the linear coding for the first and second index coding instances, each of size 29, into a matroid instance with the ground set of size 9, whose linear representability is dependent on the fields with characteristic three. The proofs and discussions provided in this paper through using these two relatively small matroid instances will shed light on the underlying reason causing the linear coding to become insufficient for the general index coding problem.
    \end{abstract}
\begin{IEEEkeywords}
    Index coding, insufficiency of linear coding, nonlinear code, matroid theory, broadcast with side information.
    \end{IEEEkeywords}

\section{Introduction}
\IEEEPARstart{I}{ndex} coding problem was first introduced by Birk and Kol \cite{Birk1998} in the context of satellite communication where through a noiseless shared channel, a single server is assigned the task of communicating $m$ messages to multiple users. While each user requests one distinct message from the server, it may have prior knowledge about a subset of the messages requested by other users, which is referred to as its side information. While sending uncoded messages leads to the total $m$ transmissions, by taking advantage of users' side information, the server might be able to satisfy all the users with a smaller number of transmissions.
The canonical model of index coding problem can be useful in studying other research areas, including network coding \cite{Rouayheb2010,Effros2015}, distributed storage \cite{Li2018}, coded caching \cite{Maddah-ali2014, Wan2020}, and topological interference management \cite{Jafar2014, Maleki2014}.

Different settings have been defined for an index coding instance. An index coding instance is said to be a unicast instance if each of its messages is requested by a single user \cite{Nujoom2019}. However, when at least one of its messages is requested by multiple users, it is referred to as a groupcast index coding instance \cite{Tehrani2012},\cite{Shanmugam2014},\cite{Sharififar-groupcast}. An index coding instance is referred to as a symmetric-rate instance if the rates of its messages are all equal. Otherwise,  it is said to be an asymmetric-rate index coding instance \cite{Maleki2014}.

Index coding schemes are broadly categorized into linear and nonlinear codes. Although
linear index coding has been the center of attention due to their straightforward encoding
and decoding processes \cite{Birk1998 ,SharififarISIT2021, Arbabjolfaeicomposite, Maleki2014,liu2018three, Thapa2017, Thomas2018, Asadi2014}, for the general index coding problem, they can be outperformed by
nonlinear codes.
The insufficiency of linear coding was proved in the context of network coding \cite{Dougherty2005}, where two network coding instances were provided to illustrate the reliance of linear coding rate on the fields with characteristic two. In fact, it was shown that for the first network coding instance, linear coding is optimal only over the fields with characteristic two, while for the second instance, linear coding over any field with characteristic two cannot be optimal. This implies that the insufficiency of linear coding is due to the dependency of its rate on the characteristic of the field on which it is operating. 
In \cite{Dougherty2007}, the authors illustrated how the constraints on the linear space of the aforementioned network coding instances can be equivalently modeled as the well-known matroid instances, namely the Fano and non-Fano matroids.
While the Fano matroid is linearly representable only over the fields with characteristic two, the non-Fano matroid has no linear representation over the fields with characteristic two. 

The connection of network coding and matroid theory with index coding was established in \cite{Rouayheb2010} and \cite{Effros2015} by presenting a reduction method to convert any network coding or matroid instance into a groupcast index coding instance. In fact, it was shown that the Fano and non-Fano matroids can be equivalently mapped into two index coding instances.
In \cite{Maleki2014}, a systematic technique of turning any groupcast index coding instance into an asymmetric-rate unicast index coding instance was proposed, implying the insufficiency of linear coding for the unicast index coding. This will convert the Fano and non-Fano matroids into two asymmetric-rate unicast index coding instances. In \cite{sharififar2021broadcast}, two symmetric-rate unicast index coding instances were directly built for which linear coding can be optimal only over the fields with characteristic two for one and odd characteristic for the other.
In terms of the scalar linear coding problem, the authors in \cite{Lubetzky2009} provided an explicit way of constructing index coding instances to show that the gap between the linear coding rate over different field sizes can be significant, highlighting the strong dependence of scalar linear coding rate on the field's characteristic. 
However, for the vector index coding problem, the scope of linear coding rate's dependency on the field size has been limited to only the fields with characteristic two.

In this paper, this scope is extended to demonstrate the reliance of linear coding rate on the fields with characteristic three. 
\\
First, by directly constructing two symmetric-rate unicast index coding instances of size 29, we prove that for the first instance, linear coding is optimal only over the fields with characteristic three, while for the second instance, linear coding over any field with characteristic three cannot be optimal. It is shown that for each index coding instance, the main constraints on the column space of its encoding matrix can be captured by a matroid instance with the ground set of size 9. Presenting the proofs using these two relatively small matroids is useful to point out the key constraints causing the linear coding rate to become dependent on the field size. In addition, applying the mapping methods in \cite{Rouayheb2010} and \cite{Maleki2014} to these matroids will lead to asymmetric-rate unicast index coding instances, each consisting of more than 1000 users, while the corresponding symmetric-rate unicast index coding instances constructed in this paper are significantly simpler as each instance is of size 29. 
\\
Second, we design the third symmetric-rate unicast index coding instance of size 58, which is a variation of the second index coding instance. It is proved that while linear coding over the fields with characteristic three cannot achieve its optimal rate, there exists an optimal nonlinear code over the fields with characteristic three. It is shown that the main constraints on the linear space of its encoding matrix can be captured by a matroid instance with the ground set of size 18, which is linearly representable over fields with any characteristic other than characteristic three.
\\
Finally, connecting the first and third index coding instances in two specific ways, namely no-way and two-way connections, will result in two new index coding instances of size 87, 91 for which linear coding is outperformed by nonlinear codes.

\subsection{Summary of Contributions}
\begin{enumerate} 
     \item In Section \ref{sec:Insufficiency of Linear Index Coding-main-results}, we design two symmetric-rate unicast index coding instances of size 87, 91 for which linear coding is outperformed by the nonlinear codes. Each instance consists of two distinct subinstances (the first and third index coding instances), which are connected in two different specific ways, called no-way and two-way connections.
    \item We design the first index coding instance of size 29 and prove that linear coding achieves its optimal rate if and only if the chosen field has characteristic three. In more details,
    \begin{enumerate} 
        \item first, we define a matroid instance with the ground set of size 9, which is linearly representable only over the fields with characteristic three (Subsection \ref{sub:matroid-instance-1-1}).
        \item Then, the first index coding instance of size 29 is characterized such that the main constraints on the column space of its encoding matrix can be captured by this matroid instance (Subsection \ref{sub:index-coding-1-1}).
    \end{enumerate}
    \item We build the second index coding instance of size 29 and prove that linear coding can never achieve its optimal rate over any field with characteristic three. In more details,
    \begin{enumerate}
        \item first, we define a matroid instance with the ground set of size 9, which does not have linear representation over any field with characteristic three (Subsection \ref{sub:matroid-instance-1-2}).
        \item Then, the second index coding instance of size 29 is characterized such that the main constraints on the column space of its encoding matrix can be captured by this matroid instance (Subsection \ref{sub:index-coding-1-2}).
    \end{enumerate}
    \item In Section \ref{sub:third-instance-3}, we design the third unicast index coding instance of size 58 (which is, in fact, a variation of the second index coding instance) and prove that while linear coding over the fields with characteristic three cannot achieve its optimal rate, there exists a scalar nonlinear code over the fields with characteristic three, which is optimal. In more details,
    \begin{enumerate}
        \item first, we define a matroid instance with the ground set of size 18, which is not linearly representable over any field with characteristic three (Subsection \ref{sub:matroid-instance-1-3}).
        \item Then, it is shown that for the third index coding instance, the main constraints on the column space of its encoding matrix can be captured by this matroid instance (Subsection \ref{sub:index-coding-1-3}).
        \item Finally, we provide a scalar nonlinear code over the fields with characteristic three, which achieves the optimal rate for this index coding instance (Subsection \ref{sub:nonlinear-index-coding-1-3}).
    \end{enumerate}
\end{enumerate}

The organization of this paper is presented in Table \ref{Table1}.

\begin{table*}
\centering
\captionsetup{format=hang}
\caption{\textsc{Organization of the Paper}}
\label{Table1}
\begin{tabular}{| >{\centering\arraybackslash}m{8mm} |  >{\centering\arraybackslash}m{12mm}  |  m{14cm}  |}
\hline
Section           
&  
Subsection
&
Content  
\\ [5pt]
\hline 
\ref{sec:system-model}   
&
\ref{sub:background-system-model}
&
The system model for index coding problem is established.
\\[5pt] \cline{2-3}
&    
\ref{sub:background-general-index-code}
&
For any index coding instance, the definitions related to its index code and broadcast rate are provided.
\\[5pt] \cline{2-3}
&    
\ref{sub:background-linear-index-code}
&
For any index coding instance, the definitions related to its linear index code, encoding matrix, linear broadcast rate,
vector and scalar linear index code are presented.
\\ \cline{2-3}
&    
\ref{sub:background-graph}
&
For any index coding instance, the definitions related to its independent sets, minimal cyclic sets, acyclic sets, are provided based on its interfering message sets.
\\ \cline{2-3}
&    
\ref{sub:background-matroid}
&
A brief overview of matroid theory and the definitions related to a matroid's basic and circuit sets, vector and scalar linear representation, and also the basic and circuit sets of its linear representation matrix are provided.
\\
\hline

\ref{sec:Insufficiency of Linear Index Coding-main-results}             
&
&
\begin{itemize}[leftmargin=*]
    \item First, for any two index coding instances, based on their interfering message sets, we characterize two specific connections, namely no-way and two-way connections.
    
    \item  Then, in Theorem \ref{thm:main-4-nonlinear-coding-instance}, it is proved that the no-way and two-way connections between the first and third index coding instances in this paper will lead to two new index coding instances for which linear coding is outperformed by the nonlinear codes.
\end{itemize}
\\
\hline

\ref{sec:Dependency-on-Char-3}                           
& 
\ref{sub:matroid-instance-1--2}
&
\begin{itemize}[leftmargin=*]
    \item Definitions \ref{def:matroid-N1} and \ref{def:matroid-N2}, respectively, characterize the first and second matroid instances $\mathcal{N}_{1}$ and $\mathcal{N}_{2}$, each with a ground set of size 9.
    
    \item In Proposition \ref{prop-matroid-N_1}, it is proved that matroid instance $\mathcal{N}_{1}$ is linearly representable only over fields with characteristic three. 
    
    \item Proposition \ref{prop-matroid-N_2} proves that matroid instance $\mathcal{N}_{2}$ is linearly representable over the fields with any characteristic other than characteristic three. 
\end{itemize}
\\ \cline{2-3}
& 
\ref{sub:index-coding-reduction-matroid}
&
Lemmas \ref{lem:MAIS1}-\ref{lem:col(9)} establish reduction techniques to map specific constraints on the encoding matrix of an index coding instance to the constraints on the matrix which linearly represents a matroid instance (their proof are provided in Appendix \ref{app:proof- Lemma1-5}).         
\\ \cline{2-3}
& 
\ref{sub:index-coding-1-1,2}
&
\begin{itemize}[leftmargin=*]
    \item Definition \ref{def:index-I1} characterizes the first index coding instance $\mathcal{I}_{1}$, comprising 29 users.
    
    \item  Theorem \ref{thm:Fano-Index-Instance-characteristic-3} states that the necessary and sufficient condition for a linear index code to be optimal for $\mathcal{I}_{1}$ is that the chosen field does have characteristic three.
    The sufficient and necessary conditions are separately proved in Propositions \ref{prop-thm1-1} and \ref{prop-thm1-2}, respectively. 
    \begin{itemize}[leftmargin=*]
        \item In Proposition \ref{prop-thm1-1}, it is shown that there exists a linear code over the fields with characteristic three, which is optimal for $\mathcal{I}_{1}$ (its proof is provided in Appendix \ref{app:proof-prop-H-I1-I2}).
        
        \item In Proposition \ref{prop-thm1-2}, using Lemmas \ref{lem:MAIS1}-\ref{lem:col(9)},  it is proved that the main constraints on the column space of the encoding matrix of index coding instance $\mathcal{I}_{1}$ are equivalent to the constraints on the column space of the matrix, which is a linear representation of matroid instance $\mathcal{N}_{1}$. This combined with Proposition \ref{prop-matroid-N_1} implies that linear coding is optimal for $\mathcal{I}_{1}$ only over the fields with characteristic three.
    \end{itemize}
\end{itemize}

\begin{itemize}[leftmargin=*]
    \item Definition \ref{def:index-I2} characterizes the second index coding instance $\mathcal{I}_{2}$, comprising 29 users.
    
    \item  Theorem \ref{thm:non-Fano-Index-Instance-characteristic-3} states that the necessary and sufficient condition for a linear index code to be optimal for $\mathcal{I}_{2}$ is that the chosen field does have any characteristic other than characteristic three.
    The sufficient and necessary conditions are separately proved in Propositions \ref{prop-thm2-1} and \ref{prop-thm2-2}, respectively. 
    \begin{itemize}[leftmargin=*]
        \item In Proposition \ref{prop-thm2-1}, it is shown that there exists a linear code over the fields with any characteristic other than characteristic three, which is optimal for $\mathcal{I}_{2}$ (its proof is provided in Appendix \ref{app:proof-prop-H-I1-I2}).
        
        \item In Proposition \ref{prop-thm2-2}, using Lemmas \ref{lem:MAIS1}-\ref{lem:col(9)}, it is proved that the main constraints on the column space of the encoding matrix of index coding instance $\mathcal{I}_{2}$ are equivalent to the constraints on the column space of the matrix, which is a linear representation of matroid instance $\mathcal{N}_{2}$. This combined with Proposition \ref{prop-matroid-N_2} implies that linear coding is optimal for $\mathcal{I}_{2}$ only over the fields with any characteristic other than characteristic three.
    \end{itemize}
\end{itemize}
\\
\hline

\ref{sub:third-instance-3}
&       
\ref{sub:matroid-instance-1-3}
&
\begin{itemize}[leftmargin=*]
    \item The concept of quasi-circuit set is defined in Definition
    \ref{def:quasi-circuit-set}.
    \item Using the concept of quasi-circuit set, Definition \ref{def:matroid-N3} characterizes matroid instance $\mathcal{N}_{3}$, with the ground set of size 18.
    
    \item In Proposition \ref{prop-matroid-N_3}, it is proved that matroid instance $\mathcal{N}_{3}$ is linearly representable over fields with any characteristic other than characteristic three.
\end{itemize}
\\ \cline{2-3}
&       
\ref{sub:index-coding-reduction-matroid-quasi}
&
Lemmas \ref{lem:quasi-independent-cycle}-\ref{lem:two-interference-dimention} establish reduction techniques to map the constraints on the encoding matrix of an index coding instance to the constraints on the representation matrix of a matroid instance (their proof are provided in Appendix \ref{app:proof- Lemma 7-9}).
\\ \cline{2-3}
&       
\ref{sub:index-coding-1-3}
&
\begin{itemize}[leftmargin=*]
    \item Definition \ref{def:index-I3} characterizes the third index coding instance $\mathcal{I}_{3}$, comprising 58 users.
    
    \item  Theorem \ref{thm:Quasi-non-Fano-Index-Instance-characteristic-3} states that the necessary and sufficient condition for a linear index code to be optimal for $\mathcal{I}_{3}$ is that the chosen field does have any characteristic other than characteristic three. However, there exists a scalar nonlinear code over the fields with characteristic three, which is optimal for $\mathcal{I}_{3}$.
    The sufficient and necessary conditions, and the existence of that nonlinear code are separately proved in Propositions \ref{prop-thm3-1}, \ref{prop-thm3-2}, and \ref{prop-thm3-3}, respectively. 
    \begin{itemize}[leftmargin=*]
        \item In Proposition \ref{prop-thm3-1}, it is shown that there exists a linear code over the fields with any characteristic other than characteristic three, which is optimal for $\mathcal{I}_{3}$.
        
        \item In Proposition \ref{prop-thm3-2}, using Lemmas \ref{lem:MAIS1}-\ref{lem:MAIS2} and Lemmas \ref{lem:quasi-independent-cycle}-\ref{lem:two-interference-dimention}, it is proved that the main constraints on the column space of the encoding matrix of index coding instance $\mathcal{I}_{3}$ are equivalent to the constraints on the column space of the matrix, which is a linear representation of matroid instance $\mathcal{N}_{3}$ (its proof is provided in Appendix \ref{app:proof-Prop-thm3-2}). This combined with Proposition \ref{prop-matroid-N_3} implies that linear coding is optimal for $\mathcal{I}_{3}$ only over the fields with any characteristic other than characteristic three.
        
        \item In Proposition \ref{prop-thm3-3}, it is shown that there exists a scalar nonlinear code over the fields with characteristic three, which is optimal for index coding instance $\mathcal{I}_{3}$.
    \end{itemize}
\end{itemize}
\\
\hline
\end{tabular}
\end{table*}

\section{System Model and Background} \label{sec:system-model}

\subsection{Notation}
Small letters such as $n$  denote an integer where  $[n]\triangleq\{1,...,n\}$ and $[n:m]\triangleq\{n,n+1,\dots m\}$ for $n<m$. Capital letters such as $L$ denote a set, with $|L|$ denoting its cardinality. Symbols in bold face such as $\boldsymbol{l}$ and $\boldsymbol{L}$, respectively, denote a vector and a matrix, with $\mathrm{rank}(\boldsymbol{L})$ and $\mathrm{col}(\boldsymbol{L})$ denoting the rank and column space of matrix $\boldsymbol{L}$, respectively. A calligraphic symbol such as $\mathcal{L}$ denotes a set whose elements are sets.\\
We use $\mathbb{F}_q$ to denote a finite field of size $q$ and write $\mathbb{F}_{q}^{n\times m}$ to denote the vector space of all $n\times m$ matrices over the field $\mathbb{F}_{q}$.
$\boldsymbol{I}_{n}$ denotes the identity matrix of size $n\times n$, and $\boldsymbol{0}_{n}$ represents an $n\times n$ matrix whose elements are all zero.

\subsection{System Model} \label{sub:background-system-model}
Consider a broadcast communication system in which a server transmits a set of $mt$ messages $X=\{x_{i}^{j},\ i\in[m],\ j\in [t]\},\ x_{i}^{j}\in \mathcal{X}$, to a number of users $U=\{u_i,\ i\in[m]\}$ through a noiseless broadcast channel. Each user $u_i$ wishes to receive a message of length $t$, $X_i=\{x_{i}^{j},\  j\in[t]\}$ and may have a priori knowledge of a subset of the messages $S_i:=\{x_{l}^{j},\ l\in A_{i},\ j\in[t]\},\ A_{i}\subseteq[m]\backslash \{i\}$, which is referred to as its side information set. The main objective is to minimize the number of coded messages which is required to be broadcast so as to enable each user to decode its requested message.
An instance of index coding problem $\mathcal{I}$ can be either characterized by the side information set of its users as $\mathcal{I}=\{A_{i}, i\in[m]\}$, or by their interfering message set $B_{i}=[m]\backslash (A_i \cup \{i\})$ as $\mathcal{I}=\{B_{i}, i\in[m]\}$.

\subsection{General Index Code}\label{sub:background-general-index-code}
\begin{defn}[$\mathcal{C}_{\mathcal{I}}$: Index Code for $\mathcal{I}$]
Given an instance of index coding problem $\mathcal{I}=\{A_{i}, i\in[m]\}$, a $(t,r)$ index code is defined as $\mathcal{C}_{\mathcal{I}}=(\phi_{\mathcal{I}},\{\psi_{\mathcal{I}}^{i}\})$, where
 \begin{itemize}
     \item $\phi_{\mathcal{I}}: \mathcal{X}^{mt}\rightarrow \mathcal{X}^{r}$ is the encoding function which maps the $mt$ message symbol $x_{i}^{j}\in \mathcal{X}$ to the $r$ coded messages as $Y=\{y_{1},\dots,y_{r}\}$, where $y_k\in \mathcal{X}, \forall k\in [r]$.
     \item $\psi_{\mathcal{I}}^{i}:$ represents the decoder function, where for each user $u_i, i\in[m]$, the decoder $\psi_{\mathcal{I}}^{i}: \mathcal{X}^{r}\times \mathcal{X}^{|A_i|t}\rightarrow \mathcal{X}^{t}$ maps the received $r$ coded messages $y_k\in Y, k\in[r]$ and the $|A_i|t$ messages $x_{l}^{j}\in S_i$ in the side information to the $t$ messages $\psi_{\mathcal{I}}^{i}(Y,S_i)=\{\hat{x}_{i}^{j}, j\in [t]\}$, where $\hat{x}_{i}^{j}$ is an estimate of $x_{i}^{j}$.
 \end{itemize}
\end{defn}

\begin{defn}[$\beta(\mathcal{C}_{\mathcal{I}})$: Broadcast Rate of $\mathcal{C}_{\mathcal{I}}$]
Given an instance of the index coding problem $\mathcal{I}$, the broadcast rate of a $(t,r)$ index code $\mathcal{C}_{\mathcal{I}}$ is defined as $\beta(\mathcal{C}_{\mathcal{I}})=\frac{r}{t}$.
\end{defn}

\begin{defn}[$\beta(\mathcal{I})$: Broadcast Rate of $\mathcal{I}$]
Given an instance of the index coding problem $\mathcal{I}$, the broadcast rate $\beta(\mathcal{I})$ is defined as
\begin{equation} \label{eq:opt-rate}
    \beta(\mathcal{I})=\inf_{t} \inf_{\mathcal{C}_{\mathcal{I}}} \beta(\mathcal{C}_{\mathcal{I}}).
\end{equation}
Thus, the broadcast rate of any index code $\mathcal{C}_{\mathcal{I}}$ provides an upper bound on the broadcast rate of $\mathcal{I}$, i.e., $\beta(\mathcal{I}) \leq \beta(\mathcal{C}_{\mathcal{I}})$.
\end{defn}


\subsection{Linear Index Code} \label{sub:background-linear-index-code}
Let $\boldsymbol{x}=[\boldsymbol{x}_{1},\dots,\boldsymbol{x}_{m}]^{T}\in \mathbb{F}_{q}^{mt\times 1}$ denote the vector message.

\begin{defn}[Linear Index Code]
Given an instance of the index coding problem $\mathcal{I}=\{B_{i}, i\in[m]\}$, a $(t,r)$ linear index code is defined as $\mathcal{L}_{\mathcal{I}}=(\boldsymbol{H},\{\psi_{\mathcal{I}}^{i}\})$, where
  \begin{itemize}
      \item $\boldsymbol{H}: \mathbb{F}_{q}^{mt\times 1}\rightarrow \mathbb{F}_{q}^{r\times 1}$ is the $r\times mt$ encoding matrix which maps the message vector $\boldsymbol{x}\in \mathbb{F}_{q}^{mt\times 1}$  to a coded message vector $\boldsymbol{\Bar{y}}=[y_{1},\dots,y_{r}]^{T}\in \mathbb{F}_{q}^{{r}\times 1}$ as follows
      \begin{equation} \nonumber
      \boldsymbol{y}=\boldsymbol{H}\boldsymbol{x}=\sum_{i\in [m]} \boldsymbol{H}^{\{i\}}\boldsymbol{x}_i.
      \end{equation}
      Here $\boldsymbol{H}^{\{i\}}\in \mathbb{F}_{q}^{r\times t}$ is the local encoding matrix of the $i$-th message $\boldsymbol{x}_i$ such that
      $\boldsymbol{H}=
      \left [\begin{array}{c|c|c}
        \boldsymbol{H}^{\{1\}} & \dots & \boldsymbol{H}^{\{m\}}
      \end{array}
     \right ]\in \mathbb{F}_{q}^{r\times mt}$.
     \item $\psi_{\mathcal{I}}^{i}$ represents the linear decoder function for user $u_{i}, i\in[m]$, where $\psi_{\mathcal{I}}^{i}(\boldsymbol{y}, S_i)$ maps the received coded message $\boldsymbol{y}$ and its side information messages $S_i$ to $\hat{\boldsymbol{x}}_{i}$, which is an estimate of the requested message vector $\boldsymbol{x}_i$.
  \end{itemize}
\end{defn}

\begin{prop} [\cite{Sharififar-UMCD-J}]\label{prop-lineardecoding-condition}
The necessary and sufficient condition for linear decoder $\psi_{\mathcal{I}}^{i}, \forall i\in[m]$ to correctly decode the requested message vector $\boldsymbol{x}_i$ is
   \begin{equation} \label{eq:dec-cond}
       \mathrm{rank} (\boldsymbol{H}^{\{i\}\cup B_{i}})= \mathrm{rank} (\boldsymbol{H}^{B_{i}}) + t,
   \end{equation}
where $\boldsymbol{H}^L$ denotes the matrix $\left [\begin{array}{c|c|c}
       \boldsymbol{H}^{\{l_1\}} & \dots & \boldsymbol{H}^{\{l_{|L|}\}}
     \end{array}
    \right ]$ for the given set $L=\{l_1,\dots,l_{|L|}\}$. 
\end{prop}

\begin{defn}[$\lambda_{q}(\mathcal{L}_{\mathcal{I}})$: Linear Broadcast Rate of $\mathcal{L}_{\mathcal{I}}$ over $\mathbb{F}_q$]
Given an instance of index coding problem $\mathcal{I}$, the linear broadcast rate of a $(t,r)$ linear index code $\mathcal{L}_{\mathcal{I}}$ over field $\mathbb{F}_q$ is defined as $\lambda_{q}(\mathcal{L}_{\mathcal{I}})=\frac{r}{t}$.
\end{defn}

\begin{defn}[$\lambda_{q}(\mathcal{I})$: Linear Broadcast Rate of $\mathcal{I}$ over $\mathbb{F}_q$]
Given an instance of index coding problem $\mathcal{I}$, the linear broadcast rate $\lambda_{q}(\mathcal{I})$ over field $\mathbb{F}_q$ is defined as
    \begin{equation} \nonumber
        \lambda_{q}(\mathcal{I})=\inf_{t} \inf_{\mathcal{L}_{\mathcal{I}}} \lambda_{q}(\mathcal{L}_{\mathcal{I}}).
    \end{equation}
\end{defn}

\begin{defn}[$\lambda(\mathcal{I})$: Linear Broadcast Rate for $\mathcal{I}$]
Given an instance of index coding problem $\mathcal{I}$, the linear broadcast rate is defined as 
\begin{equation} \label{eq:opt-lin-rate}
    \lambda(\mathcal{I})=\min_{q} \lambda_{q}(\mathcal{I}).
\end{equation}
\end{defn}

\begin{defn}[Scalar and Vector Linear Index Code]
The linear index code $\mathcal{C}_{\mathcal{I}}$ is said to be scalar if $t=1$. Otherwise, it is called a vector (or fractional) code. For scalar codes, we use $x_{i}=x_{i}^{1}, i\in [m]$, for simplicity.
\end{defn}

\subsection{Graph Definitions} \label{sub:background-graph}
Given an index coding instance $\mathcal{I}$, the following concepts are defined based on its interfering message sets, which are, in fact, related to its graph representation \cite{sharififar2021broadcast}.

\begin{defn}[Independent Set of $\mathcal{I}$]
We say that set $M\subseteq[m]$ is an independent set of $\mathcal{I}$ if $B_{i}\cap M=M\backslash \{i\}$ for all $i\in M$.
\end{defn}

\begin{defn}[Minimal Cyclic Set of $\mathcal{I}$] \label{def:minimal-cyclic-set}
Let $M=\{i_{j}, j\in [|M|]\}\subseteq[m]$. Now, $M$ is referred to as a minimal cyclic set of $\mathcal{I}$ if
 \begin{equation} \label{eq:def-minimal-cyclic-set}
 B_{i_{j}}\cap M=
    \left\{
      \begin{array}{cc}
         M\backslash \{i_{j}, i_{j+1}\},\ \ \ \ \ \ \ \ j\in [|M|-1],
         \\ \\
        M\backslash \{i_{|M|}, i_{1}\},\ \ \ \ \ \ \ \ \ \ \ \  j=i_{|M|}.\ \ \
     \end{array}
     \right.
\end{equation}
\end{defn}

\begin{defn}[Acyclic Set of $\mathcal{I}$] \label{def:acyclic-set}
We say that $M\subseteq[m]$ is an acyclic set of $\mathcal{I}$, if none of its subsets $M^{\prime}\subseteq M$ forms a minimal cyclic set of $\mathcal{I}$. We note that each independent set is an acyclic set as well.
\end{defn}

\begin{prop}[\cite{Arbabjolfaei2018}] \label{prop:rate-acyclic-minimal}
Let $\mathcal{I}=\{B_{i}, i\in [m]\}$. It can be shown that
\begin{itemize}[leftmargin=*]
    \item if set $[m]$ is an acyclic set of $\mathcal{I}$, then $\lambda_{q}(\mathcal{I})=\beta(\mathcal{I})=m$.
    \item if set $[m]$ is a minimal cyclic set of $\mathcal{I}$, then $\lambda_{q}(\mathcal{I})=\beta(\mathcal{I})=m-1$.
\end{itemize}

\end{prop}

\begin{defn}[Maximum Acyclic Induced Subgraph (MAIS) of $\mathcal{I}$]
Let $\mathcal{M}$ be the set of all sets $M\subseteq[m]$ which are acyclic sets of $\mathcal{I}$. Then, set $M\in \mathcal{M}$ with the maximum size $|M|$ is referred to as the MAIS set of $\mathcal{I}$, and $\beta_{\text{MAIS}}(\mathcal{I})=|M|$ is called the MAIS bound for $\lambda_{q}(\mathcal{I})$, as we always have \cite{Bar-Yossef2011}
\begin{align}
    \lambda_{q}(\mathcal{I})\geq \beta_{\text{MAIS}}(\mathcal{I}).
    \label{eq:MAIS-linear-coding}
\end{align}
\end{defn}

\begin{rem}
Equation \eqref{eq:MAIS-linear-coding} establishes a sufficient condition for optimality of linear coding rate as follows.
Given an index coding instance $\mathcal{I}$, if $\lambda_{q}(\mathcal{I})= \beta_{\text{MAIS}}(\mathcal{I})$, then linear coding rate is optimal for $\mathcal{I}$. In this paper, the encoding matrix which achieves this optimal rate is denoted by $\boldsymbol{H}_{\ast}$.
\end{rem}

\begin{exmp}
Consider the index coding instance $\mathcal{I}=\{B_{i}, i\in [4]\}$ where
\begin{align}
    B_{1}=\{3\}, 
    B_{2}=\{1\},
    B_{3}=\{2\},
    B_{4}=\{1,2,3\}.
\end{align}
Now, it can be seen that set \{1,2,3\} is a minimal cyclic set of $\mathcal{I}$, and each set $\{1,2,4\}, \{1,3,4\}$ and $\{2,3,4\}$ is an acyclic and also a MAIS set of $\mathcal{I}$. Thus, $\beta_{\text{MAIS}}(\mathcal{I})=3$. Now, it can be easily verified that the following encoding matrix $\boldsymbol{H}_{\ast}$ achieves the MAIS bound, and so, it is optimal for $\mathcal{I}$
\begin{align}
    \boldsymbol{H}_{\ast}=
    \begin{bmatrix}
    1 & 0 & 1 & 0\\
    0 & 1 & 1 & 0\\
    0 & 0 & 0 & 1
    \end{bmatrix}.
    \nonumber
\end{align}
\end{exmp}

\subsection{Overview of Matroid Theory} \label{sub:background-matroid}

\begin{defn}[$\mathcal{N}$: Matroid Instance \cite{Rouayheb2010,Thomas}]
A matroid instance $\mathcal{N}=\{f(N), N\subseteq [n]\}$ is a set of functions $f: 2^{[n]}\rightarrow \{0,1,2,\dots \}$ that satisfy the following three conditions:
\begin{align}
    &f(N) \leq |N|,\ \ \ \ \ \ \ \ \ \ \ \ \ \ \ \ \ \ \ \ \ \ \ \ \ \ \ \ \ \ \ \ \ \ \ \ \forall N\subseteq[n],
    \nonumber
    \\
    &f(N_{1}) \leq f(N_{2}),\ \ \ \ \ \ \ \ \ \ \ \ \ \ \ \ \ \ \ \ \ \ \ \ \ \ \ \ \ \ \ \ \forall N_{1}\subseteq N_{2}\subseteq[n],
    \nonumber
    \\
    &f(N_{1}\cup N_{2})+f(N_{1}\cap N_{2}) \leq f(N_{1})+f(N_{2}),\forall N_{1}, N_{2}\subseteq[n].
    \nonumber
\end{align}
Here, set $[n]$ and function $f(\cdot)$, respectively, are called the ground set and the rank function of $\mathcal{N}$. 
The rank of matroid $\mathcal{N}$ is defined as $f(\mathcal{N})=f([n])$.
\end{defn}

\begin{defn}[Basis and Circuit Sets of $\mathcal{N}$] \label{def:Basis-Circuit-Matroid}
Consider a matroid $\mathcal{N}$ of rank $f(\mathcal{N})$. We say that $N\subseteq[n]$ is an independent set of $\mathcal{N}$ if $f(N)=|N|$. Otherwise, $N$ is said to be a dependent set. A maximal independent set $N$ is referred to as a basis set. A minimal dependent set $N$ is referred to as a circuit set. 
Let sets $\mathcal{B}$ and $\mathcal{C}$, respectively, denote the set of all basis and circuit sets of $\mathcal{N}$. It can be shown that
\begin{align}
    &f(\mathcal{N})=f(N)=|N|,\ \ \ \ \ \ \ \ \ \ \ \ \forall N\in \mathcal{B},
    \nonumber
    \\
    &f(N\backslash \{i\})=|N|-1, \ \ \ \ \ \ \ \ \ \ \ \ \  \forall i\in N,\ \forall N\in \mathcal{C}. \label{eq:cicuitset}
\end{align}
\end{defn}

\begin{defn}[$(t)$-linear Representation of $\mathcal{N}$ over $\mathbb{F}_{q}$] \label{def:linear-rep}
We say that matroid $\mathcal{N}=\{f(N), N\subseteq [n]\}$ of rank $f(\mathcal{N})$ has a $(t)$-linear representation over $\mathbb{F}_{q}$ if there exists a matrix $\boldsymbol{H}=
      \left [\begin{array}{c|c|c}
        \boldsymbol{H}^{\{1\}} & \dots & \boldsymbol{H}^{\{n\}}
      \end{array}
     \right ]\in \mathbb{F}_{q}^{f(\mathcal{N})t \times nt}$
such that
\begin{equation} \label{matroid-linear-representation}
    \mathrm{rank}(\boldsymbol{H}^{N})=f(N)t,\ \ \  \ \forall N\subseteq [n].
\end{equation}
\end{defn}

Now, based on Definitions \ref{def:Basis-Circuit-Matroid} and \ref{def:linear-rep}, the concepts of basis and circuit sets can also be defined for matrix $\boldsymbol{H}$.

\begin{defn}[Basis and Circuit Sets of $\boldsymbol{H}$] \label{def:Basis-Circuit-Matrix}
Let $N\subseteq [n]$. We say that $N$ is an independent set of $\boldsymbol{H}$, if $\mathrm{rank}(\boldsymbol{H}^{N})=|N|t$, otherwise $N$ is a dependent set of $\boldsymbol{H}$. The independent set $N$ is a basis set of $\boldsymbol{H}$ if $\mathrm{rank}(\boldsymbol{H})=\mathrm{rank}(\boldsymbol{H}^{N})=|N|t$. The dependent set $N$ is a circuit set of $\boldsymbol{H}$ if
\begin{align}
    \mathrm{rank}(\boldsymbol{H}^{N\backslash \{j\}})=\mathrm{rank}(\boldsymbol{H}^{N})=(|N|-1)t, \ \ \forall j\in N,
    \label{eq: def:circuit-set-H}
\end{align}
which requires that
\begin{align} 
\boldsymbol{H}^{\{j\}}=\sum_{i\in N\backslash \{j\}} \boldsymbol{H}^{\{i\}}\boldsymbol{M}_{j,i}    
\label{eq:circuit-set-M-invertible}
\end{align}
where each $\boldsymbol{M}_{j,i}$ is invertible.
\end{defn}

\begin{defn}[Scalar and Vector Linear Representation]
If matroid $\mathcal{N}$ has linear representation with $t=1$, it is said that $\mathcal{N}$ has a scalar linear representation. Otherwise, the linear representation is called a vector representation.
\end{defn}

\begin{exmp}
Consider matroid instance $\mathcal{N}$ with the ground set of size $n=3$ and rank $f(\mathcal{N})=2$ such that sets $\{1,2\}, \{1,3\}, \{2,3\}$ are basis sets, and set $\{1,2,3\}$ is a circuit set. Then, the following matrix $\boldsymbol{H}$ is a scalar linear representation of $\mathcal{N}$
\begin{align}
    \boldsymbol{H}=
    \begin{bmatrix}
    1 & 0 & 1 \\
    0 & 1 & 1
    \end{bmatrix}.
    \nonumber
\end{align}
\end{exmp}

\begin{rem}
Note that the condition \eqref{eq:dec-cond} requires that
\begin{align}
    \mathrm{rank}(\boldsymbol{H}^{\{i\}})=t,\ \forall i\in [m].
    \nonumber
\end{align}
Thus, for matrix $\boldsymbol{H}$, which is a linear representation of matroid $\mathcal{N}$ with the ground set $[n]$, we also assume that
\begin{align}
    \mathrm{rank}(\boldsymbol{H}^{\{i\}})=t, \ \forall i\in [n].
    \label{eq:rank-Hi-assumption}
\end{align}
\end{rem}

\section{Main Results} \label{sec:Insufficiency of Linear Index Coding-main-results}
This section presents two new index coding instances of size 87 and 91 for which linear coding is outperformed by nonlinear codes. Each instance is composed of two index coding subinstances, which are connected using two specific ways, referred to as no-way and two-way connections. In the following sections of this paper, it will be proved that for one of these subinstances, linear coding is optimal only over the fields with characteristic three, and for the other instance, while linear coding cannot be optimal over the fields with characteristic three, there exists a nonlinear code over the fields with characteristic three, which achieves its optimal rate. This implies that although linear coding over any field cannot simultaneously be optimal for these two subinstances, there exists a nonlinear code over the fields with characteristic three, which can achieve their optimal rate at the same time.

\begin{defn}[$\mathcal{I}_{1} \nleftrightarrow \mathcal{I}_{2}$: No-way Connection of $\mathcal{I}_{1}$ and $\mathcal{I}_{2}$]
Given two index coding instances $\mathcal{I}_{1}=\{B_{i}^{1}, i\in [m_{1}]\}$ and $\mathcal{I}_{2}=\{B_{i}^{2}, i\in [m_{2}]\}$, no-way connection of $\mathcal{I}_{1}$ and $\mathcal{I}_{2}$, denoted by $\mathcal{I}_{1} \nleftrightarrow \mathcal{I}_{2}$, is defined as a new index coding instance $ \mathcal{I}=\{B_{i}, i\in[m]\}$, where $m=m_{1}+m_{2}$ and
\begin{equation} \nonumber
    \left\{
      \begin{array}{cc}
        B_{i}&=B_{i}^{1} \cup ([m]\backslash[m_{1}]),\ \ \ \ \ \ \ \ \ \forall i\in [m_{1}], \\ \\
        B_{i+m_{1}}&=B_{i}^{2} \cup [m_{1}], \ \ \ \ \ \quad \quad \quad \quad \forall i\in [m_{2}],
     \end{array}
     \right.
\end{equation}
which means that the new instance $\mathcal{I}$ is a concatenation of the two subinstances $\mathcal{I}_{1}$ and $\mathcal{I}_{2}$ such that each user in $\mathcal{I}_{1}$ has all the messages requested by the users in $\mathcal{I}_{2}$ in its interfering message set and vice versa.
\end{defn}

\begin{defn}[$\mathcal{I}_{1} \leftrightarrow \mathcal{I}_{2}$: Two-way Connection of $\mathcal{I}_{1}$ and $\mathcal{I}_{2}$]
Given two index coding instances $\mathcal{I}_{1}=\{B_{i}^{1}, i\in [m_{1}]\}$ and $\mathcal{I}_{2}=\{B_{i}^{2}, i\in [m_{2}]\}$, two-way connection of $\mathcal{I}_{1}$ and $\mathcal{I}_{2}$, denoted by $\mathcal{I}_{1} \leftrightarrow \mathcal{I}_{2}$, is defined as a new index coding instance $ \mathcal{I}^{\prime}=\{B_{i}^{\prime}, i\in[m^{\prime}]\}$, where $m^{\prime}=m_{1}+m_{2}$ and
\begin{equation} \nonumber
    \left\{
      \begin{array}{cc}
        B_{i}^{\prime}&=B_{i}^{1} , \ \ \ \forall i\in [m_{1}], \\ \\
        B_{i+m_{1}}^{\prime}&=B_{i}^{2} , \ \ \ \forall i\in [m_{2}],
     \end{array}
     \right.
\end{equation}
which means that the new instance $\mathcal{I}^{\prime}$ is a concatenation of the two subinstances $\mathcal{I}_{1}$ and $\mathcal{I}_{2}$ such that each user in $\mathcal{I}_{1}$ has all the messages requested by the users in $\mathcal{I}_{2}$ in its side information set and vice versa.
\end{defn}

\begin{prop}[\mdseries Blasiak \textit{et all}. \cite{blasiak2011lexicographic}] \label{prop:no-two}
Let $\lambda_{q}(\mathcal{I}_{1})$ and $\lambda_{q}(\mathcal{I}_{2})$, respectively, denote the linear broadcast rate of $\mathcal{I}_{1}$ and $\mathcal{I}_{2}$ over $\mathbb{F}_q$. Then, for the linear broadcast rate of $\mathcal{I}=\mathcal{I}_{1} \nleftrightarrow \mathcal{I}_{2}$ and $\mathcal{I}^{\prime}=\mathcal{I}_{1} \leftrightarrow \mathcal{I}_{2}$ over $\mathbb{F}_q$, we have
\begin{equation} \nonumber
    \left\{
      \begin{array}{cc}
         \lambda_{q}(\mathcal{I})&=\lambda_{q}(\mathcal{I}_{1})+\lambda_{q}(\mathcal{I}_{2}), \ \ \ \ \ \ \ \ \\ \\
        \lambda_{q}(\mathcal{I}^{\prime})&=\max \{\lambda_{q}(\mathcal{I}_{1}), \lambda_{q}(\mathcal{I}_{2})\}.
     \end{array}
     \right.
\end{equation}
\end{prop}

\begin{thm} \label{thm:main-4-nonlinear-coding-instance}
Other than the index coding instances in \cite{Rouayheb2010} and \cite{sharififar2021broadcast}, two new index coding instances of size 87 and 91 are designed in this paper, for which linear coding is insufficient for achieving their broadcast rate.
\end{thm}
\begin{proof}
We prove that for the following two index coding instances $\mathcal{I}=\{B_{i}, i\in [m=87]\}$, $\mathcal{I}^{\prime}=\{B_{i}^{\prime}, i\in [m^{\prime}=91]\}$, linear coding is outperformed by the nonlinear codes:
\begin{equation} \nonumber
    \left\{
      \begin{array}{cccc}
        \mathcal{I}&=\mathcal{I}_{1} \nleftrightarrow \mathcal{I}_{3}, \ \ \ \ \ \ \ \ \ \ \ \
        \\
        \\
        \mathcal{I}^{\prime}&=(\mathcal{I}_{1}\nleftrightarrow \mathcal{I}_{a}) \leftrightarrow \mathcal{I}_{3},
     \end{array}
     \right.
\end{equation}
where subinstance $\mathcal{I}_{a}$ is an acyclic index coding instance of size 4, subinstances $\mathcal{I}_{1}$ and $\mathcal{I}_{3}$ are of size 29 and 58, respectively, and will be characterized, respectively, in Sections \ref{sub:index-coding-1-1,2} and \ref{sub:third-instance-3}, with the following properties:
\begin{itemize}
    \item In Theorem \ref{thm:Fano-Index-Instance-characteristic-3}, it is proved that $\lambda_{q}(\mathcal{I}_{1})=\beta(\mathcal{I}_{1})=4$ if and only if field $\mathbb{F}_{q}$ does have characteristic three.
    \item In Theorem \ref{thm:Quasi-non-Fano-Index-Instance-characteristic-3}, it is proved that $\lambda_{q}(\mathcal{I}_{3})=\beta(\mathcal{I}_{3})=8$ if and only if field $\mathbb{F}_{q}$ does have any characteristic other than characteristic three.
\end{itemize}
From Theorems \ref{thm:Fano-Index-Instance-characteristic-3} and \ref{thm:Quasi-non-Fano-Index-Instance-characteristic-3}, it is concluded that linear coding over any field cannot simultaneously be optimal for both subinstances $\mathcal{I}_{1}$ and $\mathcal{I}_{3}$. This is because if the characteristic of $\mathbb{F}_{q}$ is three, then $\lambda_{q}(\mathcal{I}_{3})>8$, and if it is not three, then $\lambda_{q}(\mathcal{I}_{1})>4$.
\\
From Proposition \ref{prop:rate-acyclic-minimal}, $\lambda_{q}(\mathcal{I}_{a})=4$ over $\mathbb{F}_{q}$ with any characteristic. 
\\
Thus, according to Proposition \ref{prop:no-two}, the linear broadcast rate of $\mathcal{I}$ and $\mathcal{I}^{\prime}$ will be
\begin{equation}
     \left\{
      \begin{array}{cc}
        \lambda(\mathcal{I})&=\min_{q}  (\lambda_{q}(\mathcal{I}_{1}) + \lambda_{q}(\mathcal{I}_{3}))>12, \ \ \ \ \ \ \ \ \
        \\ \\
        \lambda(\mathcal{I}^{\prime})&=\min_{q}  \max \{\lambda_{q}(\mathcal{I}_{1})+3, \lambda_{q}(\mathcal{I}_{3})\}>8. \
     \end{array}
     \right.
\end{equation}
Then,
\begin{itemize}
    \item In Proposition \ref{prop-thm1-1}, we show that for subinstance $\mathcal{I}_{1}$, there is an optimal scalar linear code with the encoding matrix $\boldsymbol{H}_{\ast}\in \mathbb{F}_{q}^{4\times 29}$ and four output coded messages $\{y_{1}, y_{2}, y_{3}, y_{4}\}$.
    \item In Proposition \ref{prop-thm3-3}, for subinstance $\mathcal{I}_{3}$, we design an optimal nonlinear code with the encoder $\phi_{\mathcal{I}_{3}}$ and eight output coded messages $\{z_{1}, \dots, z_{8}\}$.
    \item According to Proposition \ref{prop:rate-acyclic-minimal}, $\lambda_{q}(\mathcal{I}_{a})=4$. Assume that the coded messages $\{y_{5},y_{6}, y_{7}, y_{8}\}$ are the optimal linear code for $\mathcal{I}_{a}$.
\end{itemize}
Now, it can be easily checked that the following coded messages are the optimal code for $\mathcal{I}$ and $\mathcal{I}^{\prime}$:

\begin{equation}
     \left\{
      \begin{array}{cc}
        \mathcal{I}: \{y_{1}, \dots, y_{4}\}\cup\{z_{1}, \dots, z_{8}\},
        \\
        \\
        \mathcal{I}^{\prime}: \{y_{1}+ z_{1},\dots, y_{8}+ z_{8}\}, \ \ \ \ \ \
     \end{array}
     \right.
\end{equation}
which completes the proof.
\end{proof}

\section{The Dependency of Linear Coding Rate on the Fields with Characteristic Three} \label{sec:Dependency-on-Char-3}

This section presents two index coding instances $\mathcal{I}_{1}$ and $\mathcal{I}_{2}$. We prove that while linear coding is optimal for $\mathcal{I}_{1}$ only over the fields with characteristic three, it can never be optimal for $\mathcal{I}_{2}$ over any field with characteristic three. To prove this, we first define two matroid instances $\mathcal{N}_{1}$ and $\mathcal{N}_{2}$ and show that their linear representation is dependent on the fields with characteristic three. Then, we show that the main constraints on the column space of the encoding matrices of  $\mathcal{I}_{1}$ and $\mathcal{I}_{2}$ can be reduced to the constraints on the column space of the matrices, which are the linear representation of $\mathcal{N}_{1}$ and $\mathcal{N}_{2}$, respectively.

\subsection{Matroid Instances $\mathcal{N}_{1}$ and $\mathcal{N}_{2}$} \label{sub:matroid-instance-1--2}

\subsubsection{Matroid Instance $\mathcal{N}_{1}$} \label{sub:matroid-instance-1-1}
\begin{defn}[Matroid Instance $\mathcal{N}_{1}$]\label{def:matroid-N1}
Consider matroid instance $\mathcal{N}_{1}=\{f(N), N\subseteq [n]\}$, where $n=9$, $f(\mathcal{N}_{1})=4$, set $N_{0}=[4]$ is a basis, and the following $N_{i}$'s, $i\in [9]$, are circuit sets:
\begin{align}
    N_{1}&=\{1,2,3,5\},\ \ 
    \nonumber
    \\
    N_{2}&=\{1,2,4,6\},\ \ 
    \nonumber
    \\
    N_{3}&=\{1,3,4,7\},
    \nonumber
    \\
    N_{4}&=\{2,3,4,8\},\ \ 
    \nonumber
    \\
    N_{5}&=\{1,8,9\},\ \ \ \ \ 
    \nonumber
    \\
    N_{6}&=\{2,7,9\},
    \nonumber
    \\
    N_{7}&=\{3,6,9\},\ \ \ \ \ 
    \nonumber
    \\
    N_{8}&=\{4,5,9\},\ \ \ \ \ 
    \nonumber
    \\
    N_{9}&=\{5,6,7,8\}.
    \label{eq:-N1-matroid}
\end{align}
\end{defn}

\begin{prop} \label{prop-matroid-N_1}
Matroid instance $\mathcal{N}_{1}$ is linearly representable only over fields with characteristic three.
\end{prop}

\begin{proof}
First, since set $[4]$ is a basis set, we get
\begin{align}
    \mathrm{rank}(\boldsymbol{H}^{[4]})=4t.
    \label{eq:matroid-N_1-N_0}
\end{align}
Since each $N_{i}, i\in [9]$ forms a circuit set, we have
\begin{align}
    N_{1}&\rightarrow \boldsymbol{H}^{\{5\}}=\boldsymbol{H}^{\{1\}}\boldsymbol{M}_{5,1}+\boldsymbol{H}^{\{2\}}\boldsymbol{M}_{5,2}+\boldsymbol{H}^{\{3\}}\boldsymbol{M}_{5,3},
    \label{eq:matroid-N_1-5}
    \\
    N_{2}&\rightarrow \boldsymbol{H}^{\{6\}}=\boldsymbol{H}^{\{1\}}\boldsymbol{M}_{6,1}+\boldsymbol{H}^{\{2\}}\boldsymbol{M}_{6,2}+\boldsymbol{H}^{\{4\}}\boldsymbol{M}_{6,4},
    \label{eq:matroid-N_1-6}
    \\
    N_{3}&\rightarrow \boldsymbol{H}^{\{7\}}=\boldsymbol{H}^{\{1\}}\boldsymbol{M}_{7,1}+\boldsymbol{H}^{\{3\}}\boldsymbol{M}_{7,3}+\boldsymbol{H}^{\{4\}}\boldsymbol{M}_{7,4},
    \label{eq:matroid-N_1-7}
    \\
    N_{4}&\rightarrow \boldsymbol{H}^{\{8\}}=\boldsymbol{H}^{\{2\}}\boldsymbol{M}_{8,2}+\boldsymbol{H}^{\{3\}}\boldsymbol{M}_{8,3}+\boldsymbol{H}^{\{4\}}\boldsymbol{M}_{8,4},
    \label{eq:matroid-N_1-8}
    \\
    N_{5}&\rightarrow \boldsymbol{H}^{\{9\}}=\boldsymbol{H}^{\{1\}}\boldsymbol{M}_{9,1}+\boldsymbol{H}^{\{8\}}\boldsymbol{M}_{9,8},
    \label{eq:matroid-N_1-9-1}
    \\
    N_{6}&\rightarrow \boldsymbol{H}^{\{9\}}=\boldsymbol{H}^{\{2\}}\boldsymbol{M}_{9,2}+\boldsymbol{H}^{\{7\}}\boldsymbol{M}_{9,7},
    \label{eq:matroid-N_1-9-2}
    \\
    N_{7}&\rightarrow \boldsymbol{H}^{\{9\}}=\boldsymbol{H}^{\{3\}}\boldsymbol{M}_{9,3}+\boldsymbol{H}^{\{6\}}\boldsymbol{M}_{9,6},
    \label{eq:matroid-N_1-9-3}
    \\
    N_{8}&\rightarrow \boldsymbol{H}^{\{9\}}=\boldsymbol{H}^{\{4\}}\boldsymbol{M}_{9,4}+\boldsymbol{H}^{\{5\}}\boldsymbol{M}_{9,5},
    \label{eq:matroid-N_1-9-4}
    \\
    N_{9}&\rightarrow \boldsymbol{H}^{\{8\}}=\boldsymbol{H}^{\{5\}}\boldsymbol{M}_{8,5}+\boldsymbol{H}^{\{6\}}\boldsymbol{M}_{8,6}+\boldsymbol{H}^{\{7\}}\boldsymbol{M}_{8,7},
    \label{eq:matroid-N_1-5:8}
\end{align}
where all matrices $\boldsymbol{M}_{j,i}$ are invertible. Now, in \eqref{eq:matroid-N_1-9-1}-\eqref{eq:matroid-N_1-9-4}, we replace $\boldsymbol{H}^{\{5\}},\boldsymbol{H}^{\{6\}},\boldsymbol{H}^{\{7\}}$ and $\boldsymbol{H}^{\{8\}}$, respectively, with their equal terms in \eqref{eq:matroid-N_1-5}-\eqref{eq:matroid-N_1-8}. Thus, $\boldsymbol{H}^{\{9\}}$ will be equal to
\begin{align}
    &\boldsymbol{H}^{\{9\}}=
    \nonumber
    \\
    &\boldsymbol{H}^{\{1\}}\boldsymbol{M}_{9,1}+(\boldsymbol{H}^{\{2\}}\boldsymbol{M}_{8,2}+\boldsymbol{H}^{\{3\}}\boldsymbol{M}_{8,3}+\boldsymbol{H}^{\{4\}}\boldsymbol{M}_{8,4})\boldsymbol{M}_{9,8},
    \nonumber
    \\
    &\boldsymbol{H}^{\{9\}}=
    \nonumber
    \\
    &\boldsymbol{H}^{\{2\}}\boldsymbol{M}_{9,2}+(\boldsymbol{H}^{\{1\}}\boldsymbol{M}_{7,1}+\boldsymbol{H}^{\{3\}}\boldsymbol{M}_{7,3}+\boldsymbol{H}^{\{4\}}\boldsymbol{M}_{7,4})\boldsymbol{M}_{9,7},
    \nonumber
    \\
    &\boldsymbol{H}^{\{9\}}=
    \nonumber
    \\
    &\boldsymbol{H}^{\{3\}}\boldsymbol{M}_{9,3}+(\boldsymbol{H}^{\{1\}}\boldsymbol{M}_{6,1}+\boldsymbol{H}^{\{2\}}\boldsymbol{M}_{6,2}+\boldsymbol{H}^{\{4\}}\boldsymbol{M}_{6,4})\boldsymbol{M}_{9,6},
    \nonumber
    \\
    &\boldsymbol{H}^{\{9\}}=
    \nonumber
    \\
    &\boldsymbol{H}^{\{4\}}\boldsymbol{M}_{9,4}+(\boldsymbol{H}^{\{1\}}\boldsymbol{M}_{5,1}+\boldsymbol{H}^{\{2\}}\boldsymbol{M}_{5,2}+\boldsymbol{H}^{\{3\}}\boldsymbol{M}_{5,3})\boldsymbol{M}_{9,5}.
    \nonumber
\end{align}
Now, due to \eqref{eq:matroid-N_1-N_0}, the above four equations, representing $\boldsymbol{H}^{\{9\}}$, are all equal only if their coefficients of $\boldsymbol{H}^{\{1\}},\boldsymbol{H}^{\{2\}},\boldsymbol{H}^{\{3\}}$ and $\boldsymbol{H}^{\{4\}}$, are equal. Thus, by equating the coefficients of $\boldsymbol{H}^{\{1\}},\boldsymbol{H}^{\{2\}},\boldsymbol{H}^{\{3\}}$ and $\boldsymbol{H}^{\{4\}}$, respectively, we have
\begin{align}
    \boldsymbol{M}_{9,1}=\boldsymbol{M}_{5,1}\boldsymbol{M}_{9,5}&=\boldsymbol{M}_{6,1}\boldsymbol{M}_{9,6}=\boldsymbol{M}_{7,1}\boldsymbol{M}_{9,7},
    \label{eq:matroid-N_1-coefficient-1}
    \\
    \boldsymbol{M}_{9,2}=\boldsymbol{M}_{5,2}\boldsymbol{M}_{9,5}&=\boldsymbol{M}_{6,2}\boldsymbol{M}_{9,6}=\boldsymbol{M}_{8,2}\boldsymbol{M}_{9,8},
    \label{eq:matroid-N_1-coefficient-2}
    \\
    \boldsymbol{M}_{9,3}=\boldsymbol{M}_{5,3}\boldsymbol{M}_{9,5}&=\boldsymbol{M}_{7,3}\boldsymbol{M}_{9,7}=\boldsymbol{M}_{8,3}\boldsymbol{M}_{9,8},
    \label{eq:matroid-N_1-coefficient-3}
    \\
    \boldsymbol{M}_{9,4}=\boldsymbol{M}_{6,4}\boldsymbol{M}_{9,6}&=\boldsymbol{M}_{7,4}\boldsymbol{M}_{9,7}=\boldsymbol{M}_{8,4}\boldsymbol{M}_{9,8}.
    \label{eq:matroid-N_1-coefficient-4}
\end{align}
Now, we have
\begin{align}
    \eqref{eq:matroid-N_1-coefficient-1}\ &\rightarrow \boldsymbol{M}_{9,5}=\boldsymbol{M}_{5,1}^{-1}\boldsymbol{M}_{6,1}\boldsymbol{M}_{9,6}=\boldsymbol{M}_{5,1}^{-1}\boldsymbol{M}_{7,1}\boldsymbol{M}_{9,7},
    \label{eq:matroid-N_1-coefficient-01}
    \\
    \eqref{eq:matroid-N_1-coefficient-2}\ &\rightarrow \boldsymbol{M}_{9,5}=\boldsymbol{M}_{5,2}^{-1}\boldsymbol{M}_{6,2}\boldsymbol{M}_{9,6}=\boldsymbol{M}_{5,2}^{-1}\boldsymbol{M}_{8,2}\boldsymbol{M}_{9,8},
    \label{eq:matroid-N_1-coefficient-02}
    \\
    \eqref{eq:matroid-N_1-coefficient-3}\ &\rightarrow \boldsymbol{M}_{9,5}=\boldsymbol{M}_{5,3}^{-1}\boldsymbol{M}_{7,3}\boldsymbol{M}_{9,7}=\boldsymbol{M}_{5,3}^{-1}\boldsymbol{M}_{8,3}\boldsymbol{M}_{9,8}.
    \label{eq:matroid-N_1-coefficient-03}
\end{align}
Thus,
\begin{align}
    \eqref{eq:matroid-N_1-coefficient-01},\eqref{eq:matroid-N_1-coefficient-02}\ &\rightarrow \boldsymbol{M}_{5,1}^{-1}\boldsymbol{M}_{6,1}=\boldsymbol{M}_{5,2}^{-1}\boldsymbol{M}_{6,2},
    \label{eq:matroid-N_1-coefficient-result-1}
    \\
    \eqref{eq:matroid-N_1-coefficient-02},\eqref{eq:matroid-N_1-coefficient-03}\ &\rightarrow \boldsymbol{M}_{5,2}^{-1}\boldsymbol{M}_{8,2}=\boldsymbol{M}_{5,3}^{-1}\boldsymbol{M}_{8,3},
    \label{eq:matroid-N_1-coefficient-result-2}
    \\
    \eqref{eq:matroid-N_1-coefficient-01},\eqref{eq:matroid-N_1-coefficient-03}\ &\rightarrow \boldsymbol{M}_{5,1}^{-1}\boldsymbol{M}_{7,1}=\boldsymbol{M}_{5,3}^{-1}\boldsymbol{M}_{7,3}.
    \label{eq:matroid-N_1-coefficient-result-3}
\end{align}
On the other hand, in \eqref{eq:matroid-N_1-5:8}, we replace $\boldsymbol{H}^{\{5\}},\boldsymbol{H}^{\{6\}},\boldsymbol{H}^{\{7\}}$ and $\boldsymbol{H}^{\{8\}}$, with their equal terms in \eqref{eq:matroid-N_1-5}-\eqref{eq:matroid-N_1-8}. By equating the coefficients of $\boldsymbol{H}^{\{1\}},\boldsymbol{H}^{\{2\}},\boldsymbol{H}^{\{3\}}$ and $\boldsymbol{H}^{\{4\}}$, we get
\begin{align}
    \boldsymbol{0}_{t}&=\boldsymbol{M}_{5,1}\boldsymbol{M}_{8,5}+\boldsymbol{M}_{6,1}\boldsymbol{M}_{8,6}+\boldsymbol{M}_{7,1}\boldsymbol{M}_{8,7},
    \label{eq:matroid-N_1-coefficient-5}
    \\
    \boldsymbol{M}_{8,2}&=\boldsymbol{M}_{5,2}\boldsymbol{M}_{8,5}+\boldsymbol{M}_{6,2}\boldsymbol{M}_{8,6},
    \label{eq:matroid-N_1-coefficient-6}
    \\
    \boldsymbol{M}_{8,3}&=\boldsymbol{M}_{5,3}\boldsymbol{M}_{8,5}+\boldsymbol{M}_{7,3}\boldsymbol{M}_{8,7},
    \label{eq:matroid-N_1-coefficient-7}
    \\
    \boldsymbol{M}_{8,4}&=\boldsymbol{M}_{6,4}\boldsymbol{M}_{8,6}+\boldsymbol{M}_{7,4}\boldsymbol{M}_{8,7}.
    \label{eq:matroid-N_1-coefficient-8}
\end{align}
Now, if \eqref{eq:matroid-N_1-coefficient-6} and \eqref{eq:matroid-N_1-coefficient-7} are multiplied by $\boldsymbol{M}_{5,2}^{-1}$ and $\boldsymbol{M}_{5,3}^{-1}$, respectively, we have 
\begin{align}
    \boldsymbol{M}_{5,2}^{-1}\boldsymbol{M}_{8,2}&=\boldsymbol{M}_{8,5}+\boldsymbol{M}_{5,2}^{-1}\boldsymbol{M}_{6,2}\boldsymbol{M}_{8,6},
    \label{eq:matroid-N_1-coefficient-9}
    \\
    \boldsymbol{M}_{5,3}^{-1}\boldsymbol{M}_{8,3}&=\boldsymbol{M}_{8,5}+\boldsymbol{M}_{5,3}^{-1}\boldsymbol{M}_{7,3}\boldsymbol{M}_{8,7}.
    \label{eq:matroid-N_1-coefficient-10}
\end{align}
Now, combining \eqref{eq:matroid-N_1-coefficient-result-2}, \eqref{eq:matroid-N_1-coefficient-9} and \eqref{eq:matroid-N_1-coefficient-10} results in
\begin{align}
    &\boldsymbol{M}_{5,2}^{-1}\boldsymbol{M}_{6,2}\boldsymbol{M}_{8,6}=\boldsymbol{M}_{5,3}^{-1}\boldsymbol{M}_{7,3}\boldsymbol{M}_{8,7}
    \\
    \rightarrow & \boldsymbol{M}_{5,1}^{-1}\boldsymbol{M}_{6,1}\boldsymbol{M}_{8,6}=\boldsymbol{M}_{5,1}^{-1}\boldsymbol{M}_{7,1}\boldsymbol{M}_{8,7}
    \label{eq:matroid-N_1-coefficient-11}
    \\
    \rightarrow & \boldsymbol{M}_{6,1}\boldsymbol{M}_{8,6}=\boldsymbol{M}_{7,1}\boldsymbol{M}_{8,7},
    \label{eq:matroid-N_1-coefficient-11-b}
\end{align}
where \eqref{eq:matroid-N_1-coefficient-11} is due to \eqref{eq:matroid-N_1-coefficient-result-1} and \eqref{eq:matroid-N_1-coefficient-result-3}. 
\\
Now, we prove $\boldsymbol{M}_{5,1}\boldsymbol{M}_{8,5}=\boldsymbol{M}_{7,1}\boldsymbol{M}_{8,7}$. From \eqref{eq:matroid-N_1-coefficient-2}-\eqref{eq:matroid-N_1-coefficient-4}, we have
\begin{align}
    \eqref{eq:matroid-N_1-coefficient-1}\ &\rightarrow \boldsymbol{M}_{9,6}=\boldsymbol{M}_{6,1}^{-1}\boldsymbol{M}_{5,1}\boldsymbol{M}_{9,5}=\boldsymbol{M}_{6,1}^{-1}\boldsymbol{M}_{7,1}\boldsymbol{M}_{9,7},
    \label{eq:matroid-N_1-coefficient-01-repeat}
    \\
    \eqref{eq:matroid-N_1-coefficient-2}\ &\rightarrow \boldsymbol{M}_{9,6}=\boldsymbol{M}_{6,2}^{-1}\boldsymbol{M}_{5,2}\boldsymbol{M}_{9,5}=\boldsymbol{M}_{6,2}^{-1}\boldsymbol{M}_{8,2}\boldsymbol{M}_{9,8},
    \label{eq:matroid-N_1-coefficient-02-repeat}
    \\
    \eqref{eq:matroid-N_1-coefficient-4}\ &\rightarrow \boldsymbol{M}_{9,6}=\boldsymbol{M}_{6,4}^{-1}\boldsymbol{M}_{7,4}\boldsymbol{M}_{9,7}=\boldsymbol{M}_{6,4}^{-1}\boldsymbol{M}_{8,4}\boldsymbol{M}_{9,8}.
    \label{eq:matroid-N_1-coefficient-03-repeat}
\end{align}
Thus,
\begin{align}
    \eqref{eq:matroid-N_1-coefficient-01-repeat},\eqref{eq:matroid-N_1-coefficient-02-repeat}\ &\rightarrow \boldsymbol{M}_{6,1}^{-1}\boldsymbol{M}_{5,1}=\boldsymbol{M}_{6,2}^{-1}\boldsymbol{M}_{5,2},
    \label{eq:matroid-N_1-coefficient-result-1-repeat}
    \\
    \eqref{eq:matroid-N_1-coefficient-02-repeat},\eqref{eq:matroid-N_1-coefficient-03-repeat}\ &\rightarrow \boldsymbol{M}_{6,2}^{-1}\boldsymbol{M}_{8,2}=\boldsymbol{M}_{6,4}^{-1}\boldsymbol{M}_{8,4},
    \label{eq:matroid-N_1-coefficient-result-2-repeat}
    \\
    \eqref{eq:matroid-N_1-coefficient-01-repeat},\eqref{eq:matroid-N_1-coefficient-03-repeat}\ &\rightarrow \boldsymbol{M}_{6,1}^{-1}\boldsymbol{M}_{7,1}=\boldsymbol{M}_{6,4}^{-1}\boldsymbol{M}_{7,4}.
    \label{eq:matroid-N_1-coefficient-result-3-repeat}
\end{align}
If \eqref{eq:matroid-N_1-coefficient-6} and \eqref{eq:matroid-N_1-coefficient-8} are multiplied by $\boldsymbol{M}_{6,2}^{-1}$ and $\boldsymbol{M}_{6,4}^{-1}$, respectively, we have 
\begin{align}
    \boldsymbol{M}_{6,2}^{-1}\boldsymbol{M}_{8,2}&=\boldsymbol{M}_{6,2}^{-1}\boldsymbol{M}_{5,2}\boldsymbol{M}_{8,5}+\boldsymbol{M}_{8,6},
    \label{eq:matroid-N_1-coefficient-9-repeat}
    \\
    \boldsymbol{M}_{6,4}^{-1}\boldsymbol{M}_{8,4}&=\boldsymbol{M}_{6,4}^{-1}\boldsymbol{M}_{7,4}\boldsymbol{M}_{8,7}+\boldsymbol{M}_{8,6}.
    \label{eq:matroid-N_1-coefficient-10-repeat}
\end{align}
Since, based on \eqref{eq:matroid-N_1-coefficient-result-2-repeat}, $\boldsymbol{M}_{6,2}^{-1}\boldsymbol{M}_{8,2}=\boldsymbol{M}_{6,4}^{-1}\boldsymbol{M}_{8,4}$, \eqref{eq:matroid-N_1-coefficient-9-repeat} and \eqref{eq:matroid-N_1-coefficient-10-repeat} will lead to
\begin{align}
    &\boldsymbol{M}_{6,2}^{-1}\boldsymbol{M}_{5,2}\boldsymbol{M}_{8,5}=\boldsymbol{M}_{6,4}^{-1}\boldsymbol{M}_{7,4}\boldsymbol{M}_{8,7}
    \\
    \rightarrow & \boldsymbol{M}_{6,1}^{-1}\boldsymbol{M}_{5,1}\boldsymbol{M}_{8,5}=\boldsymbol{M}_{6,1}^{-1}\boldsymbol{M}_{7,1}\boldsymbol{M}_{8,7},
    \label{eq:matroid-N_1-coefficient-11-1-repeat}
    \\
    \rightarrow & \boldsymbol{M}_{5,1}\boldsymbol{M}_{8,5}=\boldsymbol{M}_{7,1}\boldsymbol{M}_{8,7},
    \label{eq:matroid-N_1-coefficient-11-repeat}
\end{align}
where \eqref{eq:matroid-N_1-coefficient-11-1-repeat} is due to \eqref{eq:matroid-N_1-coefficient-result-1-repeat} and \eqref{eq:matroid-N_1-coefficient-result-3-repeat}. 
\\
Now, \eqref{eq:matroid-N_1-coefficient-5}, \eqref{eq:matroid-N_1-coefficient-11-b} and \eqref{eq:matroid-N_1-coefficient-11-repeat} will lead to
\begin{align}
    \boldsymbol{0}_{t}&=\boldsymbol{M}_{7,1}\boldsymbol{M}_{8,7}+\boldsymbol{M}_{7,1}\boldsymbol{M}_{8,7}+\boldsymbol{M}_{7,1}\boldsymbol{M}_{8,7}
    \nonumber
    \\
    &=(\boldsymbol{I}_{t}+\boldsymbol{I}_{t}+\boldsymbol{I}_{t}) \boldsymbol{M}_{7,1}\boldsymbol{M}_{8,7},
\end{align}
which is possible only over the fields with characteristic three as both $\boldsymbol{M}_{7,1}$ and $\boldsymbol{M}_{8,7}$ are invertible. This completes the proof.
\end{proof}
\subsubsection{Matroid Instance $\mathcal{N}_{2}$}\label{sub:matroid-instance-1-2}
\begin{defn}[Matroid Instance $\mathcal{N}_{2}$]\label{def:matroid-N2}
Consider matroid instance $\mathcal{N}_{2}=\{f(N), N\subseteq [n]\}$, where $n=9$, $f(\mathcal{N}_{2})=4$, each set $N_{0}=[4]$ and $N_{9}=\{5,6,7,8\}$ forms a basis, and the following $N_{i}$'s, $i\in [8]$, are circuit sets:
\begin{align}
    N_{1}&=\{1,2,3,5\},\ \ 
    \nonumber
    \\
    N_{2}&=\{1,2,4,6\},\ \ 
    \nonumber
    \\
    N_{3}&=\{1,3,4,7\},
    \nonumber
    \\
    N_{4}&=\{2,3,4,8\},\ \ 
    \nonumber
    \\
    N_{5}&=\{1,8,9\},\ \ \ \ \ 
    \nonumber
    \\
    N_{6}&=\{2,7,9\},
    \nonumber
    \\
    N_{7}&=\{3,6,9\},\ \ \ \ \ 
    \nonumber
    \\
    N_{8}&=\{4,5,9\}.
    \nonumber
\end{align}
\end{defn}

\begin{prop} \label{prop-matroid-N_2}
Matroid instance $\mathcal{N}_{2}$ is not linearly representable over any field with characteristic three.
\end{prop}

\begin{proof}
Since sets $N_{0}, \hdots, N_{8}$ in matroid $\mathcal{N}_{2}$ are exactly the same as the sets in matroid $\mathcal{N}_{1}$, equations \eqref{eq:matroid-N_1-coefficient-1}-\eqref{eq:matroid-N_1-coefficient-03} can also be derived for matroid $\mathcal{N}_{2}$. Now, since matrices $\boldsymbol{M}_{5,j},\boldsymbol{M}_{6,j},\boldsymbol{M}_{7,j}$, $\boldsymbol{M}_{8,j}$ are invertible for all $j\in [3]$, then according to equations \eqref{eq:matroid-N_1-coefficient-01}, \eqref{eq:matroid-N_1-coefficient-02} and \eqref{eq:matroid-N_1-coefficient-03}, we have 
\begin{align}
    \mathrm{col}(\boldsymbol{M}_{9,5})=\mathrm{col}(\boldsymbol{M}_{9,6})=\mathrm{col}(\boldsymbol{M}_{9,7})=\mathrm{col}(\boldsymbol{M}_{9,8}).
    \label{eq:matroid-N_2-1}
\end{align}
Now, equations \eqref{eq:matroid-N_2-1} and \eqref{eq:matroid-N_1-coefficient-1}-\eqref{eq:matroid-N_1-coefficient-4} will lead to
\begin{align}
    \mathrm{col}(\boldsymbol{M}_{9,1})=\mathrm{col}(\boldsymbol{M}_{9,2})=\mathrm{col}(\boldsymbol{M}_{9,3})=\mathrm{col}(\boldsymbol{M}_{9,4}). 
\end{align}
Thus, each $\boldsymbol{M}_{9,j}, j\in [4]$ must be invertible, since otherwise, it leads to $\mathrm{rank}(\boldsymbol{H}^{\{9\}})< t$, which contradicts \eqref{eq:rank-Hi-assumption} for $i=9$.
\\
Now, assuming that the field has characteristic three, \eqref{eq:matroid-N_1-coefficient-1}-\eqref{eq:matroid-N_1-coefficient-4}, respectively, will result in
\begin{align}
    &\boldsymbol{0}_{t}=\boldsymbol{M}_{5,1}\boldsymbol{M}_{9,5}+\boldsymbol{M}_{6,1}\boldsymbol{M}_{9,6}+\boldsymbol{M}_{7,1}\boldsymbol{M}_{9,7},
    \label{eq:matroid-N_2-coefficient-1}
    \\
    &2\boldsymbol{M}_{8,2}\boldsymbol{M}_{9,8}=\boldsymbol{M}_{5,2}\boldsymbol{M}_{9,5}+\boldsymbol{M}_{6,2}\boldsymbol{M}_{9,6},
    \label{eq:matroid-N_2-coefficient-2}
    \\
    &2\boldsymbol{M}_{8,3}\boldsymbol{M}_{9,8}=\boldsymbol{M}_{5,3}\boldsymbol{M}_{9,5}+\boldsymbol{M}_{7,3}\boldsymbol{M}_{9,7},
    \label{eq:matroid-N_2-coefficient-3}
    \\
    &2\boldsymbol{M}_{8,4}\boldsymbol{M}_{9,8}=\boldsymbol{M}_{6,4}\boldsymbol{M}_{9,6}+\boldsymbol{M}_{7,4}\boldsymbol{M}_{9,7},
    \label{eq:matroid-N_2-coefficient-4}
\end{align}
which can be rewritten as 
\begin{align}
     &2
    \begin{bmatrix}
    \boldsymbol{0}_{t}\\
    \boldsymbol{M}_{8,2}\\
    \boldsymbol{M}_{8,3}\\
    \boldsymbol{M}_{8,4}
    \end{bmatrix}
    \boldsymbol{M}_{9,8}
    =
    \nonumber
    \\
    &
    \begin{bmatrix}
    \boldsymbol{M}_{5,1}\\
    \boldsymbol{M}_{5,2}\\
    \boldsymbol{M}_{5,3}\\
    \boldsymbol{0}_{t}
    \end{bmatrix}
    \boldsymbol{M}_{9,5}
    +
    \begin{bmatrix}
    \boldsymbol{M}_{6,1}\\
    \boldsymbol{M}_{6,2}\\
    \boldsymbol{0}_{t}\\
    \boldsymbol{M}_{6,4}
    \end{bmatrix}
    \boldsymbol{M}_{9,6}
    +
    \begin{bmatrix}
    \boldsymbol{M}_{7,1}\\
    \boldsymbol{0}_{t}\\
    \boldsymbol{M}_{7,3}\\
    \boldsymbol{M}_{7,4}
    \end{bmatrix}
    \boldsymbol{M}_{9,7},
    \nonumber
\end{align}
which means that
\begin{align}
    2\boldsymbol{H}^{\{8\}}\boldsymbol{M}_{9,8}=\boldsymbol{H}^{\{5\}}\boldsymbol{M}_{9,5}+\boldsymbol{H}^{\{6\}}\boldsymbol{M}_{9,6}+\boldsymbol{H}^{\{7\}}\boldsymbol{M}_{9,7}.
    \label{N_2-final}
\end{align}
Now, since each $\boldsymbol{M}_{9,5},\boldsymbol{M}_{9,6},\boldsymbol{M}_{9,7}$ and $\boldsymbol{M}_{9,8}$ is invertible, from \eqref{N_2-final}, it is concluded that set $\{5,6,7,8\}$ forms a circuit set, which contradicts the assumption that set $N_{9}=\{5,6,7,8\}$ is a basis set of matroid $\mathcal{N}_{2}$. This completes the proof. 
\end{proof}

\subsection{On the Reduction Process from Index Coding to Matroid} \label{sub:index-coding-reduction-matroid}
In this subsection, through Lemmas \ref{lem:MAIS1}-\ref{lem:col(9)}, we establish some reduction techniques to map specific constraints on the column space of the encoder matrix of an index coding instance to the constraints on the column space of the matrix, which is a linear representation of a matroid instance.
Proofs of Lemmas \ref{lem:MAIS1}-\ref{lem:col(9)} are provided in Appendix \ref{app:proof- Lemma1-5}.

\begin{rem}
Note that the reduction technique from matroid to index coding, proposed in \cite{Rouayheb2010}, requires all the basis sets $\mathcal{B}$ and circuit sets $\mathcal{C}$ of a matroid to map the constraints on its linear representation matrix to the constraints on the encoding matrix of an index coding instance. This results in a groupcast index coding instance, with significantly high number of users. For example, applying this method to matroid instances $\mathcal{N}_{1}$ and $\mathcal{N}_{2}$ results in two groupcast index coding instances, each with more than 300 users. Moreover, applying the reduction method in \cite{Maleki2014} (from groupcast to unicast index coding instance) will lead to two asymmetric-rate unicast index coding instances, each comprising more than 1000 users. However, in the reduction techniques in this paper (Lemmas \ref{lem:MAIS1}-\ref{lem:col(9)}), we efficiently use some specific constraints to build the two symmetric-rate unicast index coding instances $\mathcal{I}_{1}$ and $\mathcal{I}_{2}$, containing only 29 users. 
\end{rem}

In this subsection, we assume that $M\subseteq [m]$, $i,l\in M$, and $j\in [m]\backslash M$.

\begin{lem}  \label{lem:MAIS1}
Assume $M$ is an acyclic set of $\mathcal{I}$. Then, the condition in \eqref{eq:dec-cond} for all $i\in M$ requires $\mathrm{rank} (\boldsymbol{H}^M)=|M|t$, implying that $M$ must be an independent set of $\boldsymbol{H}$.
\end{lem}

\begin{lem} \label{lem:min-cyc}
Let $M$ be a minimal cyclic set of $\mathcal{I}$. To have $\mathrm{rank}(\boldsymbol{H}^{M})=(|M|-1)t$, $M$ must be a circuit set of $\boldsymbol{H}$.
\end{lem}

\begin{exmp}
Consider the index coding instance $\mathcal{I}=\{B_{i}, i\in [4]\}$, where
\begin{align}
    B_{1}=\{2,3\}, \ \ B_{2}=\{3,4\},
    \nonumber
    \\
    B_{3}=\{1,4\}, \ \ B_{4}=\{1,2\}.
\end{align}
First, since set $[3]$ is an acyclic set of $\mathcal{I}$, according to Lemma \ref{lem:MAIS1}, we must have $\mathrm{rank}(\boldsymbol{H}^{[3]})=3t$. Besides, set $[4]$ is a minimal cyclic set of $\mathcal{I}$. To have $\mathrm{rank}(\boldsymbol{H}^{[4]})=3t$, according to Lemma \ref{lem:min-cyc}, set $[4]$ must be a circuit set of $\boldsymbol{H}$. It can be easily seen that the users can be all satisfied by the following encoder matrix
\begin{align}
\boldsymbol{H}=
    \begin{bmatrix}
    1 & 0 & 0 & 1 \\
    0 & 1 & 0 & 1 \\
    0 & 0 & 1 & 1
    \end{bmatrix}.
\end{align}
\end{exmp}

\begin{lem}  \label{lem:MAIS2}
Assume $M$ is an independent set of $\mathcal{I}$, and $j\in B_i,\forall i\in M\backslash \{l\}$ for some $l\in M$. Then, if $\mathrm{col}(\boldsymbol{H}^{\{j\}})\subseteq \mathrm{col}(\boldsymbol{H}^{M})$, we must have $\mathrm{col}(\boldsymbol{H}^{\{j\}})=\mathrm{col}(\boldsymbol{H}^{\{l\}})$.
\end{lem}

\begin{exmp}
Consider the index coding instance $\mathcal{I}=\{B_{i}, i\in [4]\}$, where
\begin{align}
    B_{1}=\{2,3\},\
    B_{2}=\{1,3,4\},\ 
    B_{3}=\{1,2,4\},\ 
    B_{4}=\emptyset.
    \nonumber
\end{align}
Since set $[3]$ is an independent set of $\mathcal{I}$, Lemma \ref{lem:MAIS1} requires that $\mathrm{rank}(\boldsymbol{H}^{[3]})=3t$. Now, if we desire $\mathrm{rank}(\boldsymbol{H}^{[4]})=3t$, then we must have $\mathrm{col}(\boldsymbol{H}^{\{4\}})\subseteq \mathrm{col}(\boldsymbol{H}^{[3]})$. Since $4\in B_{i}, i\in [3]\backslash \{1\}$, according to Lemma \ref{lem:MAIS2}, we must have $\mathrm{col}(\boldsymbol{H}^{\{4\}})= \mathrm{col}(\boldsymbol{H}^{\{1\}})$. It can be easily checked that the following encoder matrix can satisfy all the four users
\begin{align}
    \begin{bmatrix}
    1 & 0 & 0 & 1 \\
    0 & 1 & 0 & 0 \\
    0 & 0 & 1 & 0
    \end{bmatrix}.
\end{align}
\end{exmp}

\begin{lem}  \label{lem:independent-cycle}
Let $M\subseteq[m]$ and $j\in [m]\backslash M$. Assume that
\begin{enumerate}[label=(\roman*)]
  \item $M$ is an independent set of $\boldsymbol{H}$,
  \item $\mathrm{col}(\boldsymbol{H}^{\{j\}})\subseteq \mathrm{col}(\boldsymbol{H}^{M})$,
  \item $M$ forms a minimal cyclic set of $\mathcal{I}$,
  \item $j\in B_{i}, \forall i\in M$.
\end{enumerate}
Now, the condition in \eqref{eq:dec-cond} for all $i\in [m]$ requires set $\{j\}\cup M$ to be a circuit set of $\boldsymbol{H}$.
\end{lem}

\begin{exmp}
Consider the index coding instance $\mathcal{I}=\{B_{i}, i\in [4]\}$, where
\begin{align}
    B_{1}=\{2,4\}, B_{2}=\{3,4\}, B_{3}=\{1,4\}, B_{4}=\emptyset.
\end{align}
Assume that for the encoder matrix $\boldsymbol{H}$, we have $\mathrm{rank}(\boldsymbol{H}^{[3]})=3t$, as follows
\begin{align}
    \begin{bmatrix}
    1 & 0 & 0 \\
    0 & 1 & 0 \\
    0 & 0 & 1
    \end{bmatrix}.
    \nonumber
\end{align}
It can be seen that set $[3]$ is a minimal cyclic set of $\mathcal{I}$. Now, if we desire $\mathrm{col}(\boldsymbol{H}^{\{4\}})\subseteq \mathrm{col}(\boldsymbol{H}^{[3]})$, due to $4\in B_{i}, i\in [3]$, set $[4]$ must be a circuit set of $\boldsymbol{H}$, as follows
\begin{align}
    \begin{bmatrix}
    1 & 0 & 0 & 1\\
    0 & 1 & 0 & 1 \\
    0 & 0 & 1 & 1
    \end{bmatrix}.
    \nonumber
\end{align}
\end{exmp}

\begin{lem}  \label{lem:col(9)}
Assume for matrix $\boldsymbol{H}\in \mathbb{F}_{q}^{4t\times 9t}$, 
\begin{enumerate}[label=(\roman*)]
  \item set $[4]$ is a basis set,
  \item each set $\{1,2,3,5\}, \{1,2,4,6\}, \{1,3,4,7\}$, $\{2,3,4,8\}$ is a circuit set,
  \item \begin{align}
         \mathrm{col}(\boldsymbol{H}^{\{9\}})\subseteq \mathrm{col}(\boldsymbol{H}^{\{4,5\}}),
         \nonumber
         \\
         \mathrm{col}(\boldsymbol{H}^{\{9\}})\subseteq \mathrm{col}(\boldsymbol{H}^{\{3,6\}}),
         \nonumber
         \\
         \mathrm{col}(\boldsymbol{H}^{\{9\}})\subseteq \mathrm{col}(\boldsymbol{H}^{\{2,7\}}),
         \nonumber
         \\
         \mathrm{col}(\boldsymbol{H}^{\{9\}})\subseteq \mathrm{col}(\boldsymbol{H}^{\{1,8\}}).
         \nonumber
     \end{align}
\end{enumerate}
Then, each set $\{1,8,9\}, \{2,7,9\}, \{3,6,9\}$, $\{4,5,9\}$ is also a circuit set.
\end{lem}

\subsection{Index Coding Instances $\mathcal{I}_{1}$ and $\mathcal{I}_{2}$} \label{sub:index-coding-1-1,2}

This subsection characterizes the index coding instances $\mathcal{I}_{1}$ and $\mathcal{I}_{2}$, each of size 29, and each with the broadcast rate $\beta(\mathcal{I}_{1})=\beta(\mathcal{I}_{2})=4$. The interfering message set of all the users in $\mathcal{I}_{1}$ and $\mathcal{I}_{2}$ are exactly the same, except for users $u_{i}, i\in [5:9]$. Theorems \ref{thm:Fano-Index-Instance-characteristic-3} and \ref{thm:non-Fano-Index-Instance-characteristic-3} establish the sufficient and necessary conditions for linear coding to be optimal for $\mathcal{I}_{1}$ and $\mathcal{I}_{2}$, respectively.
\begin{itemize}
    \item Sufficient condition: It is shown that scalar linear coding with the encoding matrix $\boldsymbol{H}_{\ast}\in \mathbb{F}_{q}^{4\times 29}$, shown in Figure \ref{fig:1}, achieves the optimal broadcast rate of $\mathcal{I}_{1}$ if its field $\mathbb{F}_{q}$ does have characteristic three (such as $GF(3)$), and it is optimal for $\mathcal{I}_{2}$ if its field $\mathbb{F}_{q}$ does have any characteristic other than characteristic three (such as $GF(2)$).
    
    \item Necessary condition:
    Using Lemmas \ref{lem:MAIS1}-\ref{lem:col(9)}, it is proved that the constraints on the column space of the local encoding matrix of the first 9 users $\boldsymbol{H}^{[9]}$ in $\mathcal{I}_{1}$ and $\mathcal{I}_{2}$, respectively, are equivalent to the constraints on the column space of the matrices, which linearly represent matroid instances $\mathcal{N}_{1}$ and $\mathcal{N}_{2}$. This implies that an encoding matrix $\boldsymbol{H}$ is optimal for $\mathcal{I}_{1}$ only if field does have characteristic three, and it is optimal for $\mathcal{I}_{2}$ only if field does have any characteristic other than characteristic three.
\end{itemize}

\subsubsection{Index Coding Instance $\mathcal{I}_{1}$} \label{sub:index-coding-1-1}
\begin{defn}[Index Coding Instance $\mathcal{I}_{1}$] \label{def:index-I1}
The index coding instance $\mathcal{I}_{1}=\{B_{i}, i\in [29]\}$ is characterized as follows
\begin{align}
    B_{1}\ &=([4]\backslash\{1\})\cup \{8\}\cup ([10:25]\backslash \{10,14,18,22\}),
    \nonumber
    \\
    B_{2}\ &=([4]\backslash\{2\})\cup \{7\}\cup ([10:25]\backslash \{11,15,19,23\}),
    \nonumber
    \\
    B_{3}\ &=([4]\backslash\{3\})\cup \{6\}\cup ([10:25]\backslash \{12,16,20,24\}),
    \nonumber
    \\
    B_{4}\ &=([4]\backslash\{4\})\cup \{5\}\cup ([10:25]\backslash \{13,17,21,25\}),
    \nonumber
    \\
    B_{5}\ &=\{7,8\},
    \nonumber
    \\
    B_{6}\ &=\{5,8\}, 
    \nonumber
    \\
    B_{7}\ &=\{5,6\},
    \nonumber
    \\
    B_{8}\ &=\{6,7\},
    \nonumber
    \\
    B_{9}\ &=\{5,6,7,8\},
    \nonumber
    \\
    B_{10}&=\{5,11\},
    \nonumber
    \\
    B_{11}&=\{5,12\},
    \nonumber
    \\
    B_{12}&=\{5,10\},
    \nonumber
    \\
    B_{13}&=\{1,8,9\},
    \nonumber
    \\
    B_{14}&=\{6,15\},
    \nonumber
    \\
    B_{15}&=\{6,17\},
    \nonumber
    \\
    B_{16}&=\{4,5,9\},
    \nonumber
    \\
    B_{17}&=\{6,14\},
    \nonumber
    \\
    B_{18}&=\{7,20\},
    \nonumber
    \\
    B_{19}&=\{3,6,9\},
    \nonumber
    \\
    B_{20}&=\{7,21\},
    \nonumber
    \\
    B_{21}&=\{7,18\},
    \nonumber
    \\
    B_{22}&=\{2,7,9\},
    \nonumber
    \\
    B_{23}&=\{8,24\},
    \nonumber
    \\
    B_{24}&=\{8,25\},
    \nonumber
    \\
    B_{25}&=\{8,23\},
    \nonumber
    \\
    B_{26}&=\{4,5,9,16\},
    \nonumber
    \\
    B_{27}&=\{3,6,9,19\},
    \nonumber
    \\
    B_{28}&=\{2,7,9,22\},
    \nonumber
    \\
    B_{29}&=\{1,8,9,13\}.
    \label{I_1-B_i}
\end{align}
\end{defn}

\begin{thm} \label{thm:Fano-Index-Instance-characteristic-3}
$\lambda_{q}(\mathcal{I}_{1})=\beta_{\text{MAIS}}(\mathcal{I}_{1})=4$ if and only if $\mathbb{F}_{q}$ does have characteristic three. In other words, linear coding is optimal for $\mathcal{I}_{1}$ only over the fields with characteristic three.
\end{thm}
The proof can be concluded from Propositions \ref{prop-thm1-1} and \ref{prop-thm1-2}.

\begin{prop} \label{prop-thm1-1}
There exists a scalar linear code ($t=1$) over a field with characteristic three, which is optimal for $\mathcal{I}_{1}$.
\end{prop}

\begin{proof}
In Appendix \ref{app:proof-prop-H-I1-I2}, it is shown that the encoding matrix $\boldsymbol{H}_{\ast}\in \mathbb{F}_{q}^{4\times 29}$, shown in Figure \ref{fig:1}, will satisfy all users in $\mathcal{I}_{1}$, where the field $\mathbb{F}_{q}$ has characteristic three. The key part of $\boldsymbol{H}_{\ast}$ is its submatrix $\boldsymbol{H}_{\ast}^{\{5,6,7,8\}}$, which is as follows
\begin{align}
    \boldsymbol{H}_{\ast}^{\{5,6,7,8\}}=
    \begin{bmatrix}
    1 & 1 & 1 & 0 \\
    1 & 1 & 0 & 1 \\
    1 & 0 & 1 & 1 \\
    0 & 1 & 1 & 1 
    \end{bmatrix}.
    \nonumber
\end{align}
It can be seen that for this submatrix, $\mathrm{rank} (\boldsymbol{H}^{\{5,6,7,8\}})=3$ is achievable only over the fields with characteristic three, as
\begin{align}
    \boldsymbol{H}_{\ast}^{\{8\}}=\boldsymbol{H}_{\ast}^{\{5\}}+\boldsymbol{H}_{\ast}^{\{6\}}+\boldsymbol{H}_{\ast}^{\{7\}}.
\end{align}
This satisfies condition \eqref{eq:dec-cond} for user $u_{9}$ with $B_{9}=\{5,6,7,8\}$.
\end{proof}

\begin{figure*}
\centering
\begin{blockarray}{cccccccccccccccccccccccccccccc}
              & 
              \textcolor{blue}{{\tiny 1}}\hspace{0ex} & 
              \textcolor{blue}{{\tiny 2}}\hspace{0ex} & 
              \textcolor{blue}{{\tiny 3}}\hspace{0ex} & 
              \textcolor{blue}{{\tiny 4}}\hspace{0ex} & 
              \textcolor{blue}{{\tiny 5}}\hspace{0ex} & 
              \textcolor{blue}{{\tiny 6}}\hspace{0ex} &
              \textcolor{blue}{{\tiny 7}}\hspace{0ex} &
              \textcolor{blue}{{\tiny 8}}\hspace{0ex} & 
              \textcolor{blue}{{\tiny 9}}\hspace{0ex} & 
              \textcolor{blue}{{\tiny 10}}\hspace{0ex} &
              \textcolor{blue}{{\tiny 11}}\hspace{0ex} &
              \textcolor{blue}{{\tiny 12}}\hspace{0ex} &
              \textcolor{blue}{{\tiny 13}}\hspace{0ex} & 
              \textcolor{blue}{{\tiny 14}}\hspace{0ex} &
              \textcolor{blue}{{\tiny 15}}\hspace{0ex} &
              \textcolor{blue}{{\tiny 16}}\hspace{0ex} &
              \textcolor{blue}{{\tiny 17}}\hspace{0ex} & 
              \textcolor{blue}{{\tiny 18}}\hspace{0ex} &
              \textcolor{blue}{{\tiny 19}}\hspace{0ex} &
              \textcolor{blue}{{\tiny 20}}\hspace{0ex} &
              \textcolor{blue}{{\tiny 21}}\hspace{0ex} & 
              \textcolor{blue}{{\tiny 22}}\hspace{0ex} &
              \textcolor{blue}{{\tiny 23}}\hspace{0ex} &
              \textcolor{blue}{{\tiny 24}}\hspace{0ex} &
              \textcolor{blue}{{\tiny 25}}\hspace{0ex} & 
              \textcolor{blue}{{\tiny 26}}\hspace{0ex} &
              \textcolor{blue}{{\tiny 27}}\hspace{0ex} &
              \textcolor{blue}{{\tiny 28}}\hspace{0ex} &
              \textcolor{blue}{{\tiny 29}}\hspace{0ex}   
            \\
\begin{block}{c[cccc|cccc|c|cccc|cccc|cccc|cccc|cccc]}
               \textcolor{blue}{{\tiny 1}} & 
               1\hspace{0ex} & 0\hspace{0ex} & 0\hspace{0ex} & 0 & 
               1\hspace{0ex} & 1\hspace{0ex} & 1\hspace{0ex} & 0 & 
               1 & 
               1\hspace{0ex} & 0\hspace{0ex} & 0\hspace{0ex} & 0 & 
               1\hspace{0ex} & 0\hspace{0ex} & 0\hspace{0ex} & 0 &
               1\hspace{0ex} & 0\hspace{0ex} & 0\hspace{0ex} & 0 &
               1\hspace{0ex} & 0\hspace{0ex} & 0\hspace{0ex} & 0 &
               0\hspace{0ex} & 1\hspace{0ex} & 0\hspace{0ex} & 0 
               \\
               \textcolor{blue}{{\tiny 2}} & 
               0\hspace{0ex} & 1\hspace{0ex} & 0\hspace{0ex} & 0 & 
               1\hspace{0ex} & 1\hspace{0ex} & 0\hspace{0ex} & 1 &
               1 &
               0\hspace{0ex} & 1\hspace{0ex} & 0\hspace{0ex} & 0 &
               0\hspace{0ex} & 1\hspace{0ex} & 0\hspace{0ex} & 0 &
               0\hspace{0ex} & 1\hspace{0ex} & 0\hspace{0ex} & 0 &
               0\hspace{0ex} & 1\hspace{0ex} & 0\hspace{0ex} & 0 &
               1\hspace{0ex} & 0\hspace{0ex} & 0\hspace{0ex} & 0 
               \\
               \textcolor{blue}{{\tiny 3}} & 
               0\hspace{0ex} & 0\hspace{0ex} & 1\hspace{0ex} & 0 & 
               1\hspace{0ex} & 0\hspace{0ex} & 1\hspace{0ex} & 1 &
               1 &
               0\hspace{0ex} & 0\hspace{0ex} & 1\hspace{0ex} & 0  &
               0\hspace{0ex} & 0\hspace{0ex} & 1\hspace{0ex} & 0  &
               0\hspace{0ex} & 0\hspace{0ex} & 1\hspace{0ex} & 0 &
               0\hspace{0ex} & 0\hspace{0ex} & 1\hspace{0ex} & 0 &
               0\hspace{0ex} & 0\hspace{0ex} & 0\hspace{0ex} & 1 
               \\
               \textcolor{blue}{{\tiny 4}} & 
               0\hspace{0ex} & 0\hspace{0ex} & 0\hspace{0ex} & 1 & 
               0\hspace{0ex} & 1\hspace{0ex} & 1\hspace{0ex} & 1 & 
               1 &
               0\hspace{0ex} & 0\hspace{0ex} & 0\hspace{0ex} & 1 &
               0\hspace{0ex} & 0\hspace{0ex} & 0\hspace{0ex} & 1 &
               0\hspace{0ex} & 0\hspace{0ex} & 0\hspace{0ex} & 1 &
               0\hspace{0ex} & 0\hspace{0ex} & 0\hspace{0ex} & 1 &
               0\hspace{0ex} & 0\hspace{0ex} & 1\hspace{0ex} & 0 
               \\
             \end{block}
\end{blockarray}
\caption{$\boldsymbol{H}_{\ast}\in \mathbb{F}_{q}^{4\times 29}$: If $\mathbb{F}_{q}$ does have characteristic three (such as $GF(3)$), then $\boldsymbol{H}_{\ast}$ is an encoding matrix for the index coding instance $\mathcal{I}_{1}$, and if $\mathbb{F}_{q}$ does have any characteristic other than characteristic three (such as $GF(2)$), then $\boldsymbol{H}_{\ast}$ is an encoding matrix for the index coding instance $\mathcal{I}_{2}$.}
\label{fig:1}
\end{figure*}

\begin{prop} \label{prop-thm1-2}
Matrix $\boldsymbol{H}\in \mathbb{F}_{q}^{4t\times 29t}$ is an encoding matrix for index coding instance $\mathcal{I}_{1}$ only if its submatrix $\boldsymbol{H}^{[9]}$ is a linear representation of matroid instance $\mathcal{N}_{1}$.
\end{prop}

\begin{proof}
We prove that set $N_{0}=[4]$ is a basis set of $\boldsymbol{H}$, and each set $N_{i}, i\in [9]$ in \eqref{eq:-N1-matroid} is a circuit set of $\boldsymbol{H}$. The proof is described as follows.
\begin{itemize}
    \item First, since $\beta_{\text{MAIS}}(\mathcal{I}_{1})=4$, we must have $\mathrm{rank}(\boldsymbol{H})=4t$. Now, from $B_{i}, i\in [4]$ in \eqref{I_1-B_i}, it can be seen that set $[4]$ is an independent set of $\mathcal{I}_{1}$, so based on Lemma \ref{lem:MAIS1}, set $[4]$ is an independent set of $\boldsymbol{H}$. Since $\mathrm{rank}(\boldsymbol{H})=4t$, set $N_{0}=[4]$ will be a basis set of $\boldsymbol{H}$. Now, in order to have $\mathrm{rank} (\boldsymbol{H})=4t$, for all $j\in [29]\backslash [4]$, we must have $\mathrm{col}(\boldsymbol{H}^{\{j\}})\subseteq \mathrm{col}(\boldsymbol{H}^{[4]})$.
  
    \item According to Lemma \ref{lem:MAIS2}, from $B_{i}, i\in [4]$, it can be seen that:
    \begin{itemize}
        \item for each $j\in \{10,14,18,22\}$,
         \begin{align}
         j\in B_{i}, i\in [4]\backslash \{1\} \rightarrow \mathrm{col}(\boldsymbol{H}^{\{j\}})=\mathrm{col}(\boldsymbol{H}^{\{1\}}),
         \label{eq:prop4-third-1}
        \end{align}
       \item for each $j\in \{11,15,19,23\}$,
        \begin{align}
        j\in B_{i}, i\in [4]\backslash \{2\}\rightarrow \mathrm{col}(\boldsymbol{H}^{\{j\}})=\mathrm{col}(\boldsymbol{H}^{\{2\}}),
        \label{eq:prop4-third-2}
       \end{align}
      \item for each $j\in \{12,16,20,24\}$,
       \begin{align}
       j\in B_{i}, i\in [4]\backslash \{3\}\rightarrow \mathrm{col}(\boldsymbol{H}^{\{j\}})=\mathrm{col}(\boldsymbol{H}^{\{3\}}),
       \label{eq:prop4-third-3}
       \end{align}
      \item for each $j\in \{13,17,21,25\}$,
       \begin{align}
       j\in B_{i}, i\in [4]\backslash \{4\}\rightarrow \mathrm{col}(\boldsymbol{H}^{\{j\}})=\mathrm{col}(\boldsymbol{H}^{\{4\}}).
       \label{eq:prop4-third-4}
       \end{align}
    \end{itemize}
    Let $M_{1}=\{10,11,12\}, M_{2}=\{14,15,17\}, M_{3}=\{18,20,21\}$ and $M_{4}=\{23,24,25\}$. 
    Now, \eqref{eq:prop4-third-1}-\eqref{eq:prop4-third-4} lead to
    \begin{align}
        \mathrm{col}(\boldsymbol{H}^{M_{1}})= \mathrm{col}(\boldsymbol{H}^{[4]\backslash \{4\}}), 
        \label{eq:prop4-third-1-new}
        \\
        \mathrm{col}(\boldsymbol{H}^{M_{2}})= \mathrm{col}(\boldsymbol{H}^{[4]\backslash \{3\}}), 
        \label{eq:prop4-third-2-new}
        \\
        \mathrm{col}(\boldsymbol{H}^{M_{3}})= \mathrm{col}(\boldsymbol{H}^{[4]\backslash \{2\}}), 
        \label{eq:prop4-third-3-new}
        \\
        \mathrm{col}(\boldsymbol{H}^{M_{4}})= \mathrm{col}(\boldsymbol{H}^{[4]\backslash \{1\}}).
        \label{eq:prop4-third-4-new}
    \end{align}
    Thus, each $M_{1}, M_{2}, M_{3}$ and $M_{4}$ is an independent set of $\boldsymbol{H}$.
    \item To have $\mathrm{rank}(\boldsymbol{H})=4t$, one must have $\mathrm{rank}(\boldsymbol{H}^{B_{i}})=3t, i\in [29]$. Since $[4]$ is a basis set, from $B_{i}, i\in [4]$, we must have
    \begin{align}
    B_{4}&\rightarrow \mathrm{col}(\boldsymbol{H}^{\{5\}})\subseteq \mathrm{col}(\boldsymbol{H}^{[4]\backslash \{4\}})\stackrel{\eqref{eq:prop4-third-1-new}}{=}\mathrm{col}(\boldsymbol{H}^{M_{1}}),
    \label{eq:prop4-second-1}
    \\
    B_{3}&\rightarrow \mathrm{col}(\boldsymbol{H}^{\{6\}})\subseteq \mathrm{col}(\boldsymbol{H}^{[4]\backslash \{3\}})\stackrel{\eqref{eq:prop4-third-2-new}}{=}\mathrm{col}(\boldsymbol{H}^{M_{2}}),
    \label{eq:prop4-second-2}
    \\
    B_{2}&\rightarrow \mathrm{col}(\boldsymbol{H}^{\{7\}})\subseteq \mathrm{col}(\boldsymbol{H}^{[4]\backslash \{2\}})\stackrel{\eqref{eq:prop4-third-3-new}}{=}\mathrm{col}(\boldsymbol{H}^{M_{3}}),
    \label{eq:prop4-second-3}
    \\
    B_{1}&\rightarrow \mathrm{col}(\boldsymbol{H}^{\{8\}})\subseteq \mathrm{col}(\boldsymbol{H}^{[4]\backslash \{1\}})\stackrel{\eqref{eq:prop4-third-4-new}}{=}\mathrm{col}(\boldsymbol{H}^{M_{4}}).
    \label{eq:prop4-second-4}
  \end{align}
    \item From $B_{i}, i\in M_{1}, M_{2}, M_{3}$ and $M_{4}$, it can be verified that
    \begin{align}
        M_{1}\ \text{is a minimal cyclic set of $\mathcal{I}_{1}$}\ \& \ 5\in B_{i}, i\in M_{1},
        \label{eq:prop4-forth-1}
        \\
        M_{2}\ \text{is a minimal cyclic set of $\mathcal{I}_{1}$}\ \& \ 6\in B_{i}, i\in M_{2},
        \label{eq:prop4-forth-2}
        \\
        M_{3}\ \text{is a minimal cyclic set of $\mathcal{I}_{1}$}\ \& \ 7\in B_{i}, i\in M_{3},
        \label{eq:prop4-forth-3}
        \\
        M_{4}\ \text{is a minimal cyclic set of $\mathcal{I}_{1}$}\ \& \ 8\in B_{i}, i\in M_{4}.
        \label{eq:prop4-forth-4}
    \end{align}
     \item 
     Now, all the four conditions in Lemma \ref{lem:independent-cycle} are satisfied for set $M_{1}$ with $j=5$, set $M_{2}$ with $j=6$, set $M_{3}$ with $j=7$, and set $M_{4}$ with $j=8$. So, based on Lemma \ref{lem:independent-cycle}, each set $\{5\}\cup M_{1}, \{6\}\cup M_{2}, \{7\}\cup M_{3}$ and $\{8\}\cup M_{4}$ is a circuit set of $\boldsymbol{H}$. Now, based on \eqref{eq:prop4-third-1-new}-\eqref{eq:prop4-third-4-new}, each set $N_{1}=\{1,2,3,5\}, N_{2}=\{1,2,4,6\}, N_{3}=\{1,3,4,7\}$ and $N_{4}=\{2,3,4,8\}$ is also a circuit set. 
     \item Due to $\mathrm{rank} ( \boldsymbol{H}^{B_{i}})=3t, i\in \{26,27,28,29\}$, we must have 
     \begin{align}
         \mathrm{rank} ( \boldsymbol{H}^{\{4,5,9,16\}})=3t,
         \label{eq:prop4-new-1-1}
         \\
         \mathrm{rank} ( \boldsymbol{H}^{\{3,6,9,19\}})=3t,
         \label{eq:prop4-new-1-2}
         \\
         \mathrm{rank} ( \boldsymbol{H}^{\{2,7,9,22\}})=3t,
         \label{eq:prop4-new-1-3}
         \\
         \mathrm{rank} ( \boldsymbol{H}^{\{1,8,9,13\}})=3t.
         \label{eq:prop4-new-1-4}
     \end{align}
     Now, since $B_{16}=\{4,5,9\}, B_{19}=\{3,6,9\}, B_{22}=\{2,7,9\}$ and $B_{13}=\{1,8,9\}$, we must have
     \begin{align}
         \eqref{eq:prop4-new-1-1}\rightarrow \mathrm{rank} (\boldsymbol{H}^{\{4,5,9\}})=2t,
          \label{eq:prop4-new-1-5}
         \\
         \eqref{eq:prop4-new-1-2}\rightarrow \mathrm{rank} ( \boldsymbol{H}^{\{3,6,9\}})=2t,
          \label{eq:prop4-new-1-6}
         \\
         \eqref{eq:prop4-new-1-3}\rightarrow \mathrm{rank} ( \boldsymbol{H}^{\{2,7,9\}})=2t,
          \label{eq:prop4-new-1-7}
         \\
         \eqref{eq:prop4-new-1-4}\rightarrow \mathrm{rank} ( \boldsymbol{H}^{\{1,8,9\}})=2t.
          \label{eq:prop4-new-1-8}
     \end{align}
     Thus,
     \begin{align}
         \eqref{eq:prop4-new-1-5}\rightarrow \mathrm{col}(\boldsymbol{H}^{\{9\}})\subseteq \mathrm{col}(\boldsymbol{H}^{\{4,5\}}),
         \nonumber
         \\
         \eqref{eq:prop4-new-1-6}\rightarrow \mathrm{col}(\boldsymbol{H}^{\{9\}})\subseteq \mathrm{col}(\boldsymbol{H}^{\{3,6\}}),
         \nonumber
         \\
         \eqref{eq:prop4-new-1-7}\rightarrow \mathrm{col}(\boldsymbol{H}^{\{9\}})\subseteq \mathrm{col}(\boldsymbol{H}^{\{2,7\}}),
         \nonumber
         \\
         \eqref{eq:prop4-new-1-8}\rightarrow \mathrm{col}(\boldsymbol{H}^{\{9\}})\subseteq \mathrm{col}(\boldsymbol{H}^{\{1,8\}}).
         \nonumber
     \end{align}
     Hence, based on Lemma \ref{lem:col(9)}, each set $N_{5}=\{1,8,9\}, N_{6}=\{2,7,9\}, N_{7}=\{3,6,9\}, N_{8}=\{4,5,9\}$ is a circuit set.
     \item Finally, from $B_{5}, B_{6}, B_{7}$ and $B_{8}$, it can be seen that set $\{5,6,7,8\}$ is a minimal cyclic set of $\mathcal{I}_{1}$. Moreover,
     from $B_{9}$, we must have $\mathrm{rank}( \boldsymbol{H}^{B_{9}})=\mathrm{rank}( \boldsymbol{H}^{\{5,6,7,8\}})=3t$. Thus, based on Lemma \ref{lem:min-cyc}, set $N_{9}=\{5,6,7,8\}$ must be a circuit set of $\boldsymbol{H}$. This completes the proof.
\end{itemize}
\end{proof}

\subsubsection{Index Coding Instance $\mathcal{I}_{2}$} \label{sub:index-coding-1-2}

\begin{defn}[Index Coding Instance $\mathcal{I}_{2}$] \label{def:index-I2}
For the index coding instance $\mathcal{I}_{2}=\{B_{i}, i\in [29]\}$, the interfering message sets are all the same as the ones in \eqref{I_1-B_i}, except sets $B_{i}, i\in \{5,6,7,8,9\}$, which are as follows
\begin{align}
    B_{i}&=\{5,6,7,8\}\backslash \{i\}, \ \ \ i\in \{5,6,7,8\},
    \nonumber
    \\
    B_{9}&=\emptyset.
    \label{B-[5:9]}
\end{align}
\end{defn}

\begin{thm}\label{thm:non-Fano-Index-Instance-characteristic-3}
$\lambda_{q}(\mathcal{I}_{2})=\beta_{\text{MAIS}}(\mathcal{I}_{2})=4$ if and only if $\mathbb{F}_{q}$ does have any characteristic other than characteristic three. In other words,
linear coding is optimal for $\mathcal{I}_{1}$ only over the fields with any characteristic other than characteristic three.
\end{thm}
The proof can be concluded from Propositions \ref{prop-thm2-1} and \ref{prop-thm2-2}.

\begin{prop} \label{prop-thm2-1}
There exists a scalar linear code ($t=1$) over a field of any characteristic other than characteristic three, which is optimal for $\mathcal{I}_{2}$.
\end{prop}

\begin{proof}
In Appendix \ref{app:proof-prop-H-I1-I2}, it is shown that the encoding matrix $\boldsymbol{H}_{\ast}\in \mathbb{F}_{q}^{4\times 29}$, shown in Figure \ref{fig:1}, will satisfy all users in $\mathcal{I}_{2}$, where the field $\mathbb{F}_{q}$ has any characteristic other than characteristic three. The key part of $\boldsymbol{H}_{\ast}$ is its submatrix $\boldsymbol{H}_{\ast}^{\{5,6,7,8\}}$, which is as follows
\begin{align}
\boldsymbol{H}_{\ast}^{\{5,6,7,8\}}=
    \begin{bmatrix}
    1 & 1 & 1 & 0 \\
    1 & 1 & 0 & 1 \\
    1 & 0 & 1 & 1 \\
    0 & 1 & 1 & 1 
    \end{bmatrix}.
    \nonumber
\end{align}
It can be seen that $\mathrm{rank}(\boldsymbol{H}_{\ast}^{\{5,6,7,8\}})=4$ is achievable over fields with any characteristic other than three. This satisfies the condition in \eqref{eq:dec-cond} for users $u_{i}, i\in \{5,6,7,8\}$.
\end{proof}

\begin{prop} \label{prop-thm2-2}
Matrix $\boldsymbol{H}\in \mathbb{F}_{q}^{4t\times 29t}$ is an encoding matrix for index coding instance $\mathcal{I}_{2}$ only if its submatrix $\boldsymbol{H}^{[9]}$ is a linear representation of matroid instance $\mathcal{N}_{2}$.
\end{prop}

\begin{proof}
Since the interfering message sets $B_{i}, \in [29]\backslash \{5,6,7,8,9\}$ of $\mathcal{I}_{2}$ are the same as the sets in \eqref{I_1-B_i}, we can borrow the results from Proposition \ref{prop-thm1-2}, where set $[4]$ is a basis set, and sets $\{1,2,3,5\}, \{1,2,4,6\}, \{1,3,4,7\}$, $\{2,3,4,8\}, \{1,8,9\}$, $\{2,7,9\}$, $\{3,6,9\}$ and $\{4,5,9\}$ are circuit sets. Now, due to \eqref{B-[5:9]}, the set $\{5,6,7,8\}$ must also be a basis set, which completes the proof. 
\end{proof}

\section{Index Coding Instance $\mathcal{I}_{3}$} \label{sub:third-instance-3}
This subsection provides index coding instance $\mathcal{I}_{3}$ with the MAIS bound $\beta_{\text{MAIS}}(\mathcal{I}_{3})=8$. First, we prove that linear coding cannot achieve the optimal rate over any field with characteristic three. To prove this, we first define a matroid instances $\mathcal{N}_{3}$ and show that it is not linearly representable over the fields with characteristic three. Then, we show that the main constraints on the column space of the encoding matrix of  $\mathcal{I}_{3}$ can be reduced to the constraints on the column space of the matrix, which is the linear representation of $\mathcal{N}_{3}$. Finally, we provide a scalar nonlinear code over the fields with characteristic three, which is optimal for $\mathcal{I}_{3}$.

In this subsection, we assume that $M\subseteq[m], N\subseteq[n]$, and the value of each $m,n, |M|$ and $|N|$ is an even integer.

\subsection{Matroid Instance $\mathcal{N}_{3}$} \label{sub:matroid-instance-1-3}
In this subsection, we first define the concept of quasi-circuit set of a matrix which is similar to the concept of circuit set (where this similarity can be seen by comparing Equations \eqref{eq: def:quasi-circuit-set-H} and \eqref{eq:quasi-circuit-set-N-invertible}, respectively, with Equations \eqref{eq: def:circuit-set-H} and \eqref{eq:circuit-set-M-invertible}).

\begin{defn}[Qausi-circuit Set of Matrix $\boldsymbol{H}$] \label{def:quasi-circuit-set}
Let $L\subseteq[\frac{n}{2}]$. We say that set $N=\{2j-1,2j, j\in L\}\subseteq [n]$ is a quasi-circuit set of $\boldsymbol{H}$, if for all $j\in L$, we have
\begin{align}
    &\mathrm{rank}(\boldsymbol{H}^{\{2j-1, 2j\}})\ \  =2t,
    \nonumber
    \\
    &\mathrm{rank}(\boldsymbol{H}^{N\backslash \{2j-1, 2j\}})=\mathrm{rank}(\boldsymbol{H}^{N})=(|N|-2)t.
    \label{eq: def:quasi-circuit-set-H}
\end{align}
\end{defn}

\begin{lem} \label{lem:quasi-circuit-invertible}
Let $L\subseteq[\frac{n}{2}]$. Assume $N=\{2j-1,2j, j\in L\}$ is a quasi-circuit set of $\boldsymbol{H}$ and
\begin{equation}
    \boldsymbol{N}_{j,i}\triangleq
    \begin{bmatrix}
    \boldsymbol{M}_{2j-1,2i-1} & \boldsymbol{M}_{2j,2i-1}
    \\
    \boldsymbol{M}_{2j-1,2i}   & \boldsymbol{M}_{2j,2i}
    \end{bmatrix}.
\end{equation}
Now, for any $j\in L$, we have
\begin{align}
    \boldsymbol{H}^{\{2j-1,2j\}}=\sum_{i\in L\backslash \{j\}} \boldsymbol{H}^{\{2i-1,2i\}}\boldsymbol{N}_{j,i},
    \label{eq:quasi-circuit-set-N-invertible}
\end{align}
such that each $\boldsymbol{N}_{j,i}$ is invertible.
\end{lem}

\begin{proof}
Equation \eqref{eq: def:quasi-circuit-set-H} requires that
\begin{align}
    \mathrm{col}(\boldsymbol{H}^{\{2j-1, 2j\}})\subseteq \mathrm{col}(\boldsymbol{H}^{N\backslash \{2j-1, 2j\}})).
    \nonumber
\end{align}
Thus, we must have $\boldsymbol{H}^{\{2j-1, 2j\}}=\sum_{i\in L\backslash\{j\}} \boldsymbol{H}^{\{2i-1, 2i\}} \boldsymbol{N}_{j,i}$. Now,
if one of the $\boldsymbol{N}_{j,i}, i=l\in L\backslash \{j\}$ is not invertible, then $\mathrm{rank}( \boldsymbol{H}^{N\backslash\{2l-1, 2l\}})<(|N|-2)t$, which contradicts \eqref{eq: def:quasi-circuit-set-H}. Thus, all $\boldsymbol{N}_{j,i}, i\in L\backslash \{j\}$ must be invertible.
\end{proof}

\begin{exmp}
It can be seen that for the following matrix
\begin{align}
\boldsymbol{H}=
    \begin{bmatrix}
    1 & 0 & 0 & 0 & 1 & 1\\
    0 & 1 & 0 & 0 & 0 & 1\\
    0 & 0 & 1 & 0 & 1 & 0\\
    0 & 0 & 0 & 1 & 1 & 1
    \end{bmatrix}, 
    \label{eq:exmp-quasi-circuit-set}
\end{align}
set $[6]$ is a quasi-circuit set as we have
\begin{align}
    &\mathrm{rank}(\boldsymbol{H}^{\{1,2\}})=\mathrm{rank}(\boldsymbol{H}^{\{3,4\}})=\mathrm{rank}(\boldsymbol{H}^{\{5,6\}})=2,
    \nonumber
    \\
    &\mathrm{rank}(\boldsymbol{H}^{\{1,2,3,4\}})=\mathrm{rank}(\boldsymbol{H}^{\{1,2,5,6\}})=\mathrm{rank}(\boldsymbol{H}^{\{3,4,5,6\}})=4,
    \nonumber
    \\
    &\mathrm{rank}(\boldsymbol{H}^{[6]})=4.
\end{align}
\end{exmp}

\begin{defn}[Matroid Instance $\mathcal{N}_{3}$]\label{def:matroid-N3}
Consider the matroid instance $\mathcal{N}_{3}=\{f(N), N\subseteq [n]\}$, where $n=18$, $f(\mathcal{N}_{3})=8$, set $N_{0}=[8]$ is a basis set, the sets $N_{i}$'s, $i\in [8]$ are quasi-circuit sets, which are as follows
\begin{align}
    N_{1}&=\{1,2,3,4,5,6,9,10\},\ \
    \nonumber
    \\
    N_{2}&=\{1,2,3,4,7,8,11,12\},
    \nonumber
    \\
    N_{3}&=\{1,2,5,6,7,8,13,14\},\
    \nonumber
    \\
    N_{4}&=\{3,4,5,6,7,8,15,16\},
    \nonumber
    \\
    N_{5}&=\{1,2,15,16,17,18\},\ \ \ \
    \nonumber
    \\
    N_{6}&=\{3,4,13,14,17,18\},
    \nonumber
    \\
    N_{7}&=\{5,6,11,12,17,18\},\ \ \ \ 
    \nonumber
    \\
    N_{8}&=\{7,8,9,10,17,18\},
    \label{eq:-N3-matroid}
\end{align}
and 
\begin{align}
    f(N_{9}=\{9:16\})\geq 7.
    \label{eq:-N3-matroid-N9}
\end{align}
\end{defn}

\begin{prop} \label{prop-matroid-N_3}
Matroid instance $\mathcal{N}_{3}$ is not linearly representable over any field with characteristic three.
\end{prop}

\begin{proof}
Since set $[8]$ is a basis set, and each $N_{i}, i\in [8]$ in \eqref{eq:-N3-matroid} is a quasi-circuit set, \eqref{eq:quasi-circuit-set-N-invertible} results in
\begin{align}
     \boldsymbol{H}^{\{9,10\}}\ &=\boldsymbol{H}^{\{1,2\}}\boldsymbol{N}_{5,1}+\boldsymbol{H}^{\{3,4\}}\boldsymbol{N}_{5,2}+\boldsymbol{H}^{\{5,6\}}\boldsymbol{N}_{5,3},
    \label{eq:matroid-N_3-5}
    \\
     \boldsymbol{H}^{\{11,12\}}&=\boldsymbol{H}^{\{1,2\}}\boldsymbol{N}_{6,1}+\boldsymbol{H}^{\{3,4\}}\boldsymbol{N}_{6,2}+\boldsymbol{H}^{\{7,8\}}\boldsymbol{N}_{6,4},
    \label{eq:matroid-N_3-6}
    \\
     \boldsymbol{H}^{\{13,14\}}&=\boldsymbol{H}^{\{1,2\}}\boldsymbol{N}_{7,1}+\boldsymbol{H}^{\{5,6\}}\boldsymbol{N}_{7,3}+\boldsymbol{H}^{\{7,8\}}\boldsymbol{N}_{7,4},
    \label{eq:matroid-N_3-7}
    \\
     \boldsymbol{H}^{\{15,16\}}&=\boldsymbol{H}^{\{3,4\}}\boldsymbol{N}_{8,2}+\boldsymbol{H}^{\{5,6\}}\boldsymbol{N}_{8,3}+\boldsymbol{H}^{\{7,8\}}\boldsymbol{N}_{8,4},
    \label{eq:matroid-N_3-8}
    \\
     \boldsymbol{H}^{\{17,18\}}&=\boldsymbol{H}^{\{1,2\}}\boldsymbol{N}_{9,1}+\boldsymbol{H}^{\{15,16\}}\boldsymbol{N}_{9,8},
    \label{eq:matroid-N_3-9-1}
    \\
     \boldsymbol{H}^{\{17,18\}}&=\boldsymbol{H}^{\{3,4\}}\boldsymbol{N}_{9,2}+\boldsymbol{H}^{\{13,14\}}\boldsymbol{N}_{9,7},
    \label{eq:matroid-N_3-9-2}
    \\
     \boldsymbol{H}^{\{17,18\}}&=\boldsymbol{H}^{\{5,6\}}\boldsymbol{N}_{9,3}+\boldsymbol{H}^{\{11,12\}}\boldsymbol{N}_{9,6},
    \label{eq:matroid-N_3-9-3}
    \\
     \boldsymbol{H}^{\{17,18\}}&=\boldsymbol{H}^{\{7,8\}}\boldsymbol{N}_{9,4}+\boldsymbol{H}^{\{9,10\}}\boldsymbol{N}_{9,5}.
    \label{eq:matroid-N_3-9-4}
\end{align}
Now, since each $\boldsymbol{N}_{j,i}$ is invertible (according to Lemma \ref{lem:quasi-circuit-invertible}), equations \eqref{eq:matroid-N_3-5}-\eqref{eq:matroid-N_3-9-4} are similar to equations \eqref{eq:matroid-N_1-5}-\eqref{eq:matroid-N_1-9-4}. Therefore, over the fields with characteristic three, we can achieve the similar result in \eqref{N_2-final}, as follows
\begin{align}
    &2\boldsymbol{H}^{\{15,16\}}\boldsymbol{N}_{9,8}=
    \nonumber
    \\
    &\boldsymbol{H}^{\{9,10\}}\boldsymbol{N}_{9,5}+\boldsymbol{H}^{\{11,12\}}\boldsymbol{N}_{9,6}+\boldsymbol{H}^{\{13,14\}}\boldsymbol{N}_{9,7},
    \label{N_3-final}
\end{align}
where each $\boldsymbol{N}_{9,5},\boldsymbol{N}_{9,6}$, $\boldsymbol{N}_{9,7}$, end $\boldsymbol{N}_{9,8}$ is invertible. Thus,
\begin{align}
    \mathrm{rank}(\boldsymbol{H}^{\{9:16\}})=\mathrm{rank}(\boldsymbol{H}^{\{9:14\}})\leq 6,
\end{align}
which contradicts \eqref{eq:-N3-matroid-N9}. This completes the proof.
\end{proof}

\subsection{On the Reduction Process from Index Coding to Matroid} \label{sub:index-coding-reduction-matroid-quasi}

In this subsection, first we define the concept of quasi-minimal cyclic set of an index coding instance, which is similar to the concept of minimal cyclic set (where this similarity can be seen by comparing Equation \eqref{eq:def-minimal-cyclic-set} with Equations \eqref{eq:def-quasi-minimal-cyclic-set} and \eqref{eq:def-quasi-minimal-cyclic-set-1}). Then, through Lemmas \ref{lem:quasi-independent-cycle}-\ref{lem:two-interference-dimention}, we establish some reduction techniques to map specific constraints on the column space of the encoder matrix of an index coding instance to the constraints on the column space of the matrix, linearly representing a matroid instance. Lemmas \ref{lem:quasi-independent-cycle} and \ref{lem:quasi-col(9)}, respectively, are variations of Lemmas \ref{lem:independent-cycle} and \ref{lem:col(9)}, where the concept of quasi-minimal cyclic set is used instead of the concept of minimal cyclic set.
Proof of Lemmas \ref{lem:quasi-independent-cycle}-\ref{lem:two-interference-dimention} are provided in Appendix \ref{app:proof- Lemma 7-9}.

Here, we assume that $m$ is an even integer,  $L^{\prime}\subseteq L=[\frac{m}{2}]$,  $M^{\prime}=\{2i-1, 2i, i\in L^{\prime}\}\subseteq M=\{2i-1,2i, i\in L\}\subseteq[m]$. 

\begin{defn}[Qausi-minimal Cyclic Set of $\mathcal{I}$] \label{def:quasi-minimal-cyclic-set}
Let $M=\{i_{2j-1}, i_{2j}, j\in [\frac{|M|}{2}]\}$. Now, $M$ is referred to as a quasi-minimal cyclic set of $\mathcal{I}$ if
 \begin{equation} \label{eq:def-quasi-minimal-cyclic-set}
 B_{i_{2j-1}}\cap M=
    \left\{
      \begin{array}{cc}
         M\backslash \{i_{2j-1}, i_{2j+1},i_{2j+2}\}, j\in [\frac{|M|}{2}-1],
         \\ \\
        M\backslash \{i_{2j-1}, i_{1},i_{2}\},\ \ \ \ \ \ \ \ \ \ \ \ \ \ j=\frac{|M|}{2}.
     \end{array}
     \right.
     \nonumber
\end{equation}

 \begin{equation} \label{eq:def-quasi-minimal-cyclic-set-1}
 B_{i_{2j}}\cap M=
    \left\{
      \begin{array}{cc}
         M\backslash \{i_{2j}, i_{2j+1},i_{2j+2}\}, j\in [\frac{|M|}{2}-1],
         \\ \\
        M\backslash \{i_{2j}, i_{1},i_{2}\},\ \ \ \ \ \ \ \ \ \ \ \ \ \ j=\frac{|M|}{2}.
     \end{array}
     \right.
\end{equation}

\end{defn}

\begin{exmp}
Consider the index coding instance $\mathcal{I}=\{B_{i}, i\in [6]\}$, where
\begin{align}
    B_{1}=\{2,3,4\},\ B_{2}=\{1,3,4\},
    \nonumber
    \\
    B_{3}=\{4,5,6\},\ B_{4}=\{3,5,6\},
    \nonumber
    \\
    B_{5}=\{1,2,6\},\ B_{6}=\{1,2,5\}.
    \nonumber
\end{align}
It can be seen that set $[6]$ is a quasi-minimal cyclic set of $\mathcal{I}$. It can be also checked that the matrix in \eqref{eq:exmp-quasi-circuit-set} is an encoding matrix for $\mathcal{I}$ and can satisfy all the users $u_{i}, \in [6]$.
\end{exmp}

\begin{lem} \label{lem:quasi-independent-cycle}
Assume 
\begin{enumerate}[label=(\roman*)]
    \item $M$ is an independent set of $\boldsymbol{H}$,
    \item $\mathrm{col}(\boldsymbol{H}^{\{2j-1,2j\}})\subseteq \mathrm{col}(\boldsymbol{H}^{M})$, 
    \item $M$ forms a quasi-minimal cyclic set of $\mathcal{I}$,
    \item $\{2j-1,2j\}\subseteq B_{i}, \forall i\in M$.
\end{enumerate}
Now, the condition in \eqref{eq:dec-cond} for all $i\in [m]$ requires set $\{2j-1,2j\}\cup M$ to be a quasi-circuit set of $\boldsymbol{H}$.
\end{lem}

\begin{lem} \label{lem:quasi-col(9)}
Suppose for matrix $\boldsymbol{H}\in \mathbb{F}_{q}^{8t\times 18}$, 
\begin{enumerate}[label=(\roman*)]
    \item set $[8]$ is a basis set,
    \item each set $\{1,2,3,4,5,6,9,10\}$, $\{1,2,3,4,7,8,11,12\}$, $\{1,2,5,6,7,8,13,14\}$ and $\{3,4,5,6,7,8,15,16\}$ is a quasi-circuit set,
    \item 
         \begin{align}
          \mathrm{col}(\boldsymbol{H}^{\{17,18\}})&\subseteq \mathrm{col}(\boldsymbol{H}^{\{7,8,9,10\}}).
         \nonumber
         \\
          \mathrm{col}(\boldsymbol{H}^{\{17,18\}})&\subseteq \mathrm{col}(\boldsymbol{H}^{\{5,6,11,12\}}),
         \nonumber
         \\
          \mathrm{col}(\boldsymbol{H}^{\{17,18\}})&\subseteq \mathrm{col}(\boldsymbol{H}^{\{3,4,13,14\}}),
         \nonumber
         \\
          \mathrm{col}(\boldsymbol{H}^{\{17,18\}})&\subseteq \mathrm{col}(\boldsymbol{H}^{\{1,2,15,16\}}),
         \nonumber
     \end{align}
\end{enumerate}
Then, each set $\{1,2,15,16,17,18\}$, $\{3,4,13,14,17,18\}$, $\{5,6,11,12,17,18\}$ and $\{7,8,9,10,17,18\}$ will also be a quasi-circuit set.
\end{lem}

\begin{lem} \label{lem:two-interference-dimention}
Let matrix $\boldsymbol{H}$ be an encoding matrix for index coding instance $\mathcal{I}=\{B_{i}, i\in [m]\}$.
Assume $M^{\prime}\subseteq M\subseteq[m]$. Now, if $M\backslash \{i\}\subseteq B_{i}$ for all $i\in M^{\prime}$, then we must have
\begin{align}
    \mathrm{rank}(\boldsymbol{H}^{M})=\mathrm{rank}(\boldsymbol{H}^{M\backslash M^{\prime}})+|M^{\prime}|t.
    \nonumber
\end{align}
\end{lem}

\subsection{Index Coding Instance $\mathcal{I}_{3}$} \label{sub:index-coding-1-3}
\begin{defn}[Index Coding Instance $\mathcal{I}_{3}$] \label{def:index-I3}
The index coding instance $\mathcal{I}_{3}=\{B_{i}, i\in [58]\}$ is characterized as follows
\begin{align}
    B_{1}\  &=([8]\backslash\{1\})\cup \{15,16\}\cup ([19:50]\backslash \{19,27,35,43\}),
    \nonumber
    \\
    B_{2}\  &=([8]\backslash\{2\})\cup \{15,16\}\cup ([19:50]\backslash \{20,28,36,44\}),
    \nonumber
    \\
    B_{3}\  &=([8]\backslash\{3\})\cup \{13,14\}\cup ([19:50]\backslash \{21,29,37,45\}),
    \nonumber
    \\
    B_{4}\  &=([8]\backslash\{4\})\cup \{13,14\}\cup ([19:50]\backslash \{22,30,38,46\}),
    \nonumber
    \\
    B_{5}\  &=([8]\backslash\{5\})\cup \{11,12\}\cup ([19:50]\backslash \{23,31,39,47\}),
    \nonumber
    \\
    B_{6}\  &=([8]\backslash\{6\})\cup \{11,12\}\cup ([19:50]\backslash \{24,32,40,48\}),
    \nonumber
    \\
    B_{7}\  &=([8]\backslash\{7\})\cup \{9,10\}\  \cup ([19:50]\backslash \{25,33,41,49\}),
    \nonumber
    \\
    B_{8}\  &=([8]\backslash\{8\})\cup \{9,10\}\ \cup ([19:50]\backslash \{26,34,42,50\}),
    \nonumber
    \\
    B_{9}\  &=[9:16]\backslash \{9\},
    \nonumber
    \\
    B_{10}&=\{9,12,14\},
    \nonumber
    \\
    B_{11}&=[9:16]\backslash \{11\},
    \nonumber
    \\
    B_{12}&=\{11,14,16\},
    \nonumber
    \\
    B_{13}&=[9:16]\backslash \{13\},
    \nonumber
    \\
    B_{14}&=\{13,10,16\},
    \nonumber
    \\
    B_{15}&=[9:16]\backslash \{15\}, 
    \nonumber
    \\
    B_{16}&=\{15,10,12\},
    \nonumber
    \\
    B_{17}&=\{18\},
    \nonumber
    \\
    B_{18}&=\{17\},
    \nonumber
    \\
    B_{19}&=\{9,10,20,21,22\},
    \nonumber
    \\
    B_{20}&=\{9,10,19,21,22\},
    \nonumber
    \\
    B_{21}&=\{9,10,22,23,24\},
    \nonumber
    \\
    B_{22}&=\{9,10,21,23,24\},
    \nonumber
    \\
    B_{23}&=\{9,10,19,20,24\},
    \nonumber
    \\
    B_{24}&=\{9,10,19,20,23\},
    \nonumber
    \\
    B_{25}&=\{1,2,15,16,17,18,26\},
    \nonumber
    \\
    B_{26}&=\{1,2,15,16,17,18,25\},
    \nonumber
    \\
    B_{27}&=\{11,12,28,29,30\},
    \nonumber
    \\
    B_{28}&=\{11,12,27,29,30\},
    \nonumber
    \\
    B_{29}&=\{11,12,30,33,34\},
    \nonumber
    \\
    B_{30}&=\{11,12,29,33,34\},
    \nonumber
    \\
    B_{31}&=\{7,8,9,10,17,18,32\},
    \nonumber
    \\
    B_{32}&=\{7,8,9,10,17,18,31\},
    \nonumber
    \\
    B_{33}&=\{11,12,27,28,34\},
    \nonumber
    \\
    B_{34}&=\{11,12,27,28,33\},
    \nonumber
    \\
    B_{35}&=\{13,14,36,39,40\},
    \nonumber
    \\
    B_{36}&=\{13,14,35,39,40\},
    \nonumber
    \\
    B_{37}&=\{5,6,11,12,17,18,38\},
    \nonumber
    \\
    B_{38}&=\{5,6,11,12,17,18,37\},
    \nonumber
    \\
    B_{39}&=\{13,14,40,41,42\},
    \nonumber
    \\
    B_{40}&=\{13,14,39,41,42\},
    \nonumber
    \\
    B_{41}&=\{13,14,35,36,42\},
    \nonumber
    \\
    B_{42}&=\{13,14,35,36,41\},
    \nonumber
    \\
    B_{43}&=\{3,4,13,14,17,18,44\},
    \nonumber
    \\
    B_{44}&=\{3,4,13,14,17,18,43\},
    \nonumber
    \\
    B_{45}&=\{15,16,46,47,48\},
    \nonumber
    \\
    B_{46}&=\{15,16,45,47,48\},
    \nonumber
    \\
    B_{47}&=\{15,16,48,49,50\},
    \nonumber
    \\
    B_{48}&=\{15,16,47,49,50\},
    \nonumber
    \\
    B_{49}&=\{15,16,45,46,50\},
    \nonumber
    \\
    B_{50}&=\{15,16,45,46,49\},
    \nonumber
    \\
    B_{51}&=\{7,8,9,10,17,18,31,32,52\},
    \nonumber
    \\
    B_{52}&=\{7,8,9,10,17,18,31,32,51\},
    \nonumber
    \\
    B_{53}&=\{5,6,11,12,17,18,37,38,54\},
    \nonumber
    \\
    B_{54}&=\{5,6,11,12,17,18,37,38,53\},
    \nonumber
    \\
    B_{55}&=\{3,4,13,14,17,18,43,44,56\},
    \nonumber
    \\
    B_{56}&=\{3,4,13,14,17,18,43,44,55\},
    \nonumber
    \\
    B_{57}&=\{1,2,15,16,17,18,25,26,58\},
    \nonumber
    \\
    B_{58}&=\{1,2,15,16,17,18,25,26,57\}.
    \label{I_3-B_i}
\end{align}

\end{defn}

\begin{thm} \label{thm:Quasi-non-Fano-Index-Instance-characteristic-3}
$\lambda_{q}(\mathcal{I}_{3})=\beta_{\text{MAIS}}(\mathcal{I}_{3})=8$ if and only if $\mathbb{F}_{q}$ does have any characteristic other than characteristic three. In other words,
linear coding is optimal for $\mathcal{I}_{3}$ only over the fields with any characteristic other than characteristic three. However, there exists a scalar nonlinear code over the fields with characteristic three, which is optimal for $\mathcal{I}_{3}$.
\end{thm}
\begin{proof}
The proof can be concluded from Propositions \ref{prop-thm3-1}, \ref{prop-thm3-2}, and \ref{prop-thm3-3}.
\end{proof}

\begin{prop} \label{prop-thm3-1}
There exists a scalar linear coding over a field of any characteristic other than characteristic three, which is optimal for $\mathcal{I}_{3}$.
\end{prop}

\begin{proof}
It can be verified that the encoding matrix $\boldsymbol{H}_{\ast}\in \mathbb{F}_{q}^{8\times 58}$, shown in Figure \ref{fig:2}, will satisfy all users in $\mathcal{I}_{3}$, where the field $\mathbb{F}_{q}$ has any characteristic other than characteristic three. The key part of $\boldsymbol{H}_{\ast}$ is its submatrix $\boldsymbol{H}_{\ast}^{\{9:16\}}$, which is as follows
\begin{align}
    \boldsymbol{H}_{\ast}^{\{9:16\}}=
    \begin{bmatrix}
    \boldsymbol{I}_{2} & \boldsymbol{I}_{2} & \boldsymbol{I}_{2} & \boldsymbol{0}_{2} \\
    \boldsymbol{I}_{2} & \boldsymbol{I}_{2} & \boldsymbol{0}_{2} & \boldsymbol{I}_{2} \\
    \boldsymbol{I}_{2} & \boldsymbol{0}_{2} & \boldsymbol{I}_{2} & \boldsymbol{I}_{2} \\
    \boldsymbol{0}_{2} & \boldsymbol{I}_{2} & \boldsymbol{I}_{2} & \boldsymbol{I}_{2} 
    \end{bmatrix},
    \nonumber
\end{align}
It can be seen that $\mathrm{rank}(\boldsymbol{H}_{\ast}^{\{9:16\}})=8>7$ is achievable over fields with any characteristic other than characteristic three. This satisfies the condition in \eqref{eq:dec-cond} for users $u_{i}, i\in \{9,11,13,15\}$.
\end{proof}

\begin{figure*}
\centering
\begin{blockarray}{cccccccccccccccccccccccccccccc}
              & 
              \textcolor{blue}{{\tiny 1,2}}\hspace{-1.7ex}   & 
              \textcolor{blue}{{\tiny 3,4}}\hspace{-1.7ex}   & 
              \textcolor{blue}{{\tiny 5,6}}\hspace{-1.7ex}   & 
              \textcolor{blue}{{\tiny 7,8}}\hspace{   0ex}   & 
              \textcolor{blue}{{\tiny 9,10}}\hspace{-1.7ex}  & 
              \textcolor{blue}{{\tiny 11,12}}\hspace{-1.7ex} &
              \textcolor{blue}{{\tiny 13,14}}\hspace{-1.7ex} &
              \textcolor{blue}{{\tiny 15,16}}\hspace{   0ex} & 
              \textcolor{blue}{{\tiny 17,18}}\hspace{   0ex} & 
              \textcolor{blue}{{\tiny 19,20}}\hspace{-1.7ex} &
              \textcolor{blue}{{\tiny 21,22}}\hspace{-1.7ex} &
              \textcolor{blue}{{\tiny 23,24}}\hspace{-1.7ex} &
              \textcolor{blue}{{\tiny 25,26}}\hspace{   0ex} & 
              \textcolor{blue}{{\tiny 27,28}}\hspace{-1.7ex} &
              \textcolor{blue}{{\tiny 29,30}}\hspace{-1.7ex} &
              \textcolor{blue}{{\tiny 31,32}}\hspace{-1.7ex} &
              \textcolor{blue}{{\tiny 33,34}}\hspace{   0ex} & 
              \textcolor{blue}{{\tiny 35,36}}\hspace{-1.7ex} &
              \textcolor{blue}{{\tiny 37,38}}\hspace{-1.7ex} &
              \textcolor{blue}{{\tiny 39,40}}\hspace{-1.7ex} &
              \textcolor{blue}{{\tiny 41,42}}\hspace{   0ex} & 
              \textcolor{blue}{{\tiny 43,44}}\hspace{-1.7ex} &
              \textcolor{blue}{{\tiny 45,46}}\hspace{-1.7ex} &
              \textcolor{blue}{{\tiny 47,48}}\hspace{-1.7ex} &
              \textcolor{blue}{{\tiny 49,50}}\hspace{   0ex} & 
              \textcolor{blue}{{\tiny 51,52}}\hspace{-1.7ex} &
              \textcolor{blue}{{\tiny 53,54}}\hspace{-1.7ex} &
              \textcolor{blue}{{\tiny 55,56}}\hspace{-1.7ex} &
              \textcolor{blue}{{\tiny 57,58}}\hspace{   0ex}   
            \\
\begin{block}{c[cccc|cccc|c|cccc|cccc|cccc|cccc|cccc]}
               \textcolor{blue}{{\tiny 1,2}} & 
               $\boldsymbol{I}_{2}$\hspace{-1.7ex} & $\boldsymbol{0}_{2}$\hspace{-1.7ex} & $\boldsymbol{0}_{2}$\hspace{-1.7ex} & $\boldsymbol{0}_{2}$ & 
               $\boldsymbol{I}_{2}$\hspace{-1.7ex} & $\boldsymbol{I}_{2}$\hspace{-1.7ex} & $\boldsymbol{I}_{2}$\hspace{-1.7ex} & $\boldsymbol{0}_{2}$ & 
               $\boldsymbol{I}_{2}$ & 
               $\boldsymbol{I}_{2}$\hspace{-1.7ex} & $\boldsymbol{0}_{2}$\hspace{-1.7ex} & $\boldsymbol{0}_{2}$\hspace{-1.7ex} & $\boldsymbol{0}_{2}$ & 
               $\boldsymbol{I}_{2}$\hspace{-1.7ex} & $\boldsymbol{0}_{2}$\hspace{-1.7ex} & $\boldsymbol{0}_{2}$\hspace{-1.7ex} & $\boldsymbol{0}_{2}$ &
               $\boldsymbol{I}_{2}$\hspace{-1.7ex} & $\boldsymbol{0}_{2}$\hspace{-1.7ex} & $\boldsymbol{0}_{2}$\hspace{-1.7ex} & $\boldsymbol{0}_{2}$ &
               $\boldsymbol{I}_{2}$\hspace{-1.7ex} & $\boldsymbol{0}_{2}$\hspace{-1.7ex} & $\boldsymbol{0}_{2}$\hspace{-1.7ex} & $\boldsymbol{0}_{2}$ &
               $\boldsymbol{0}_{2}$\hspace{-1.7ex} & $\boldsymbol{I}_{2}$\hspace{-1.7ex} & $\boldsymbol{0}_{2}$\hspace{-1.7ex} & $\boldsymbol{0}_{2}$ 
               \\
               \textcolor{blue}{{\tiny 3,4}} & 
               $\boldsymbol{0}_{2}$\hspace{-1.7ex} & $\boldsymbol{I}_{2}$\hspace{-1.7ex} & $\boldsymbol{0}_{2}$\hspace{-1.7ex} & $\boldsymbol{0}_{2}$ & 
               $\boldsymbol{I}_{2}$\hspace{-1.7ex} & $\boldsymbol{I}_{2}$\hspace{-1.7ex} & $\boldsymbol{0}_{2}$\hspace{-1.7ex} & $\boldsymbol{I}_{2}$ &
               $\boldsymbol{I}_{2}$ &
               $\boldsymbol{0}_{2}$\hspace{-1.7ex} & $\boldsymbol{I}_{2}$\hspace{-1.7ex} & $\boldsymbol{0}_{2}$\hspace{-1.7ex} & $\boldsymbol{0}_{2}$ &
               $\boldsymbol{0}_{2}$\hspace{-1.7ex} & $\boldsymbol{I}_{2}$\hspace{-1.7ex} & $\boldsymbol{0}_{2}$\hspace{-1.7ex} & $\boldsymbol{0}_{2}$ &
               $\boldsymbol{0}_{2}$\hspace{-1.7ex} & $\boldsymbol{I}_{2}$\hspace{-1.7ex} & $\boldsymbol{0}_{2}$\hspace{-1.7ex} & $\boldsymbol{0}_{2}$ &
               $\boldsymbol{0}_{2}$\hspace{-1.7ex} & $\boldsymbol{I}_{2}$\hspace{-1.7ex} & $\boldsymbol{0}_{2}$\hspace{-1.7ex} & $\boldsymbol{0}_{2}$ &
               $\boldsymbol{I}_{2}$\hspace{-1.7ex} & $\boldsymbol{0}_{2}$\hspace{-1.7ex} & $\boldsymbol{0}_{2}$\hspace{-1.7ex} & $\boldsymbol{0}_{2}$ 
               \\
               \textcolor{blue}{{\tiny 5,6}} & 
               $\boldsymbol{0}_{2}$\hspace{-1.7ex} & $\boldsymbol{0}_{2}$\hspace{-1.7ex} & $\boldsymbol{I}_{2}$\hspace{-1.7ex} & $\boldsymbol{0}_{2}$ & 
               $\boldsymbol{I}_{2}$\hspace{-1.7ex} & $\boldsymbol{0}_{2}$\hspace{-1.7ex} & $\boldsymbol{I}_{2}$\hspace{-1.7ex} & $\boldsymbol{I}_{2}$ &
               $\boldsymbol{I}_{2}$ &
               $\boldsymbol{0}_{2}$\hspace{-1.7ex} & $\boldsymbol{0}_{2}$\hspace{-1.7ex} & $\boldsymbol{I}_{2}$\hspace{-1.7ex} & $\boldsymbol{0}_{2}$  &
               $\boldsymbol{0}_{2}$\hspace{-1.7ex} & $\boldsymbol{0}_{2}$\hspace{-1.7ex} & $\boldsymbol{I}_{2}$\hspace{-1.7ex} & $\boldsymbol{0}_{2}$  &
               $\boldsymbol{0}_{2}$\hspace{-1.7ex} & $\boldsymbol{0}_{2}$\hspace{-1.7ex} & $\boldsymbol{I}_{2}$\hspace{-1.7ex} & $\boldsymbol{0}_{2}$ &
               $\boldsymbol{0}_{2}$\hspace{-1.7ex} & $\boldsymbol{0}_{2}$\hspace{-1.7ex} & $\boldsymbol{I}_{2}$\hspace{-1.7ex} & $\boldsymbol{0}_{2}$ &
               $\boldsymbol{0}_{2}$\hspace{-1.7ex} & $\boldsymbol{0}_{2}$\hspace{-1.7ex} & $\boldsymbol{0}_{2}$\hspace{-1.7ex} & $\boldsymbol{I}_{2}$ 
               \\
               \textcolor{blue}{{\tiny 7,8}} & 
               $\boldsymbol{0}_{2}$\hspace{-1.7ex} & $\boldsymbol{0}_{2}$\hspace{-1.7ex} & $\boldsymbol{0}_{2}$\hspace{-1.7ex} & $\boldsymbol{I}_{2}$ & 
               $\boldsymbol{0}_{2}$\hspace{-1.7ex} & $\boldsymbol{I}_{2}$\hspace{-1.7ex} & $\boldsymbol{I}_{2}$\hspace{-1.7ex} & $\boldsymbol{I}_{2}$ & 
               $\boldsymbol{I}_{2}$ &
               $\boldsymbol{0}_{2}$\hspace{-1.7ex} & $\boldsymbol{0}_{2}$\hspace{-1.7ex} & $\boldsymbol{0}_{2}$\hspace{-1.7ex} & $\boldsymbol{I}_{2}$ &
               $\boldsymbol{0}_{2}$\hspace{-1.7ex} & $\boldsymbol{0}_{2}$\hspace{-1.7ex} & $\boldsymbol{0}_{2}$\hspace{-1.7ex} & $\boldsymbol{I}_{2}$ &
               $\boldsymbol{0}_{2}$\hspace{-1.7ex} & $\boldsymbol{0}_{2}$\hspace{-1.7ex} & $\boldsymbol{0}_{2}$\hspace{-1.7ex} & $\boldsymbol{I}_{2}$ &
               $\boldsymbol{0}_{2}$\hspace{-1.7ex} & $\boldsymbol{0}_{2}$\hspace{-1.7ex} & $\boldsymbol{0}_{2}$\hspace{-1.7ex} & $\boldsymbol{I}_{2}$ &
               $\boldsymbol{0}_{2}$\hspace{-1.7ex} & $\boldsymbol{0}_{2}$\hspace{-1.7ex} & $\boldsymbol{I}_{2}$\hspace{-1.7ex} & $\boldsymbol{0}_{2}$ 
               \\
             \end{block}
\end{blockarray}
\caption{$\boldsymbol{H}_{\ast}\in \mathbb{F}_{q}^{8\times 58}$: If $\mathbb{F}_{q}$ does have any characteristic other than characteristic three (such as $GF(2))$, then $\boldsymbol{H}_{\ast}$ is an encoding matrix for index coding instance $\mathcal{I}_{3}$.}
\label{fig:2}
\end{figure*}

\begin{prop} \label{prop-thm3-2}
Matrix $\boldsymbol{H}\in \mathbb{F}_{q}^{8t\times 58t}$ is an encoding matrix for index coding instance $\mathcal{I}_{3}$ only if submatrix $\boldsymbol{H}^{[18]}$ is a linear representation of matroid instance $\mathcal{N}_{3}$.
\end{prop}
\begin{proof}
Refer to Appendix \ref{app:proof-Prop-thm3-2}
\end{proof}

\subsection{An Optimal Nonlinear Code for $\mathcal{I}_{3}$ over Fields with Characteristic Three} \label{sub:nonlinear-index-coding-1-3}
\begin{defn}[Nonlinear Function $g(\cdot)$]
Let $x_{i},x_{j}, x_{l}, x_{v}\in \mathbb{F}_{q}=GF(3)$. Now, the nonlinear function $g(\cdot): \mathbb{F}_{q}^{4}\rightarrow \mathbb{F}_{q}$ is defined as follows:
\begin{align}
    g(x_{i},x_{j},x_{l},x_{v})=\ & 2x_{i}x_{i}(x_{j}+ x_{l} + x_{v})\ + 
    \nonumber
    \\
    & 2x_{j}x_{j}(x_{i}+ x_{l}+ x_{v})\ + 
    \nonumber
    \\
    & 2x_{l}x_{l}(x_{i}+ x_{j}+ x_{v})\ +
    \nonumber
    \\
    & 2x_{v}x_{v}(x_{i}+ x_{j}+ x_{l})\ +
    \nonumber
    \\
    & 2(x_{i}x_{j} + x_{i}x_{l} + x_{i}x_{v} + 
    \nonumber
    \\
    & x_{j}x_{l} + x_{j}x_{v} + x_{l}x_{v}) +
    \nonumber
    \\
    & x_{i}x_{j}x_{l} + x_{i}x_{j}x_{v} +x_{i}x_{l}x_{v} + x_{j}x_{l}x_{v}.
    \label{g-fucntion}
\end{align}
\end{defn}



\begin{lem} \label{lem:g-function}
Let $x_{i},x_{j}, x_{l}, x_{v}, x_{w}\in GF(3)$. Then, using the value of $x_{w}$ and the following five combinations:
\begin{align}
    &g(x_{i},x_{j},x_{l},x_{w}) + g(x_{i},x_{j},x_{v},x_{w}) +
    \nonumber
    \\
    &g(x_{i},x_{l},x_{v},x_{w}) + g(x_{j},x_{l},x_{v},x_{w}),
    \nonumber
    \\
    &x_{i}+ x_{j}+ x_{l},
    \nonumber
    \\
    &x_{i}+ x_{j}+ x_{v},
    \nonumber
    \\
    &x_{i}+ x_{l}+ x_{v},
    \nonumber
    \\
    &x_{j}+ x_{l}+ x_{v},
    \nonumber
\end{align}
we can find the value of each $x_{i}, x_{j}, x_{l}$, and $x_{v}$.
\end{lem}

\begin{proof}
Refer to Appendix \ref{app:proof-lem-g-function}.
\end{proof}

\begin{lem} \label{lem2:g-function}
Using the value of $x_{i}, x_{j}, x_{l}$ and $x_{v}+x_{w}$, we can find the value of $g(x_{i},x_{j},x_{v},x_{w})+2g(x_{i},x_{l},x_{v},x_{w})$.
\end{lem}
\begin{proof}
Refer to Appendix \ref{app:proof-lem-g-function}.
\end{proof}

\begin{prop} \label{prop-thm3-3}
There exists a scalar nonlinear code over the fields with characteristic three, which can achieve the broadcast rate of $\mathcal{I}_{3}$.
\end{prop}
\begin{proof}
First, it can be seen that set $[8]$ is a MAIS set of $\mathcal{I}_{3}$. So, $\beta_{\text{MAIS}}(\mathcal{I}_{3})=8$. Now, we prove that $\beta(\mathcal{C}_{\mathcal{I}_{3}})=8$ for a scalar nonlinear index code $\mathcal{C}_{\mathcal{I}_{3}}=(\phi_{\mathcal{I}_{3}}, \{\psi_{\mathcal{I}_{3}}^i\})$, where the encoder and decoder do as below. 
\\
First, function $\phi_{\mathcal{I}_{3}}$ encodes messages $x_{i}, i\in [58]$ into eight coded messages $z_{k}, k\in [8]$, as follows
\begin{align}
    \{z_{j}, j\in [8]\}=\phi_{\mathcal{I}_{3}}(\{x_{i}, i\in [58]\}),
\end{align}
where
\begin{align}
   z_{1}&= x_{1}+x_{9}+x_{11}+x_{13}+x_{17}+x_{19}+x_{27}+x_{35}+x_{43}+x_{53},
   \nonumber 
    \\
    \nonumber
    \\
   z_{2}&= x_{2}+x_{10}+x_{12}+x_{14}+x_{18}+x_{20}+x_{28}+x_{36}+x_{44}+x_{54}
   \nonumber
    \\
   & \ \ +\ g(x_{9} ,x_{11}, x_{13}, x_{17}), 
   \nonumber
    \\
    \nonumber
    \\
   z_{3}&= x_{3}+x_{9}+x_{11}+x_{15}+x_{17}+x_{21}+x_{29}+x_{37}+x_{45}+x_{51},
   \nonumber
    \\
    \nonumber
    \\
   z_{4}&= x_{4}+x_{10}+x_{12}+x_{16}+x_{18}+x_{22}+x_{30}+x_{38}+x_{46}+x_{52}
   \nonumber
    \\
   &\ \ +\ g(x_{9} ,x_{11}, x_{15}, x_{17}),
   \nonumber
    \\
    \nonumber
    \\
   z_{5}&= x_{5}+x_{9}+x_{13}+x_{15}+x_{17}+x_{23}+x_{31}+x_{39}+x_{47}+x_{57},
   \nonumber
    \\
    \nonumber
    \\
   z_{6}&= x_{6}+x_{10}+x_{14}+x_{16}+x_{18}+x_{24}+x_{32}+x_{40}+x_{48}+x_{58}
   \nonumber
    \\
   & \ \ +\ g(x_{9} ,x_{13}, x_{15}, x_{17}),
   \nonumber
   \\
   \nonumber
    \\
    z_{7}&= x_{7}+x_{11}+x_{13}+x_{15}+x_{17}+x_{25}+x_{33}+x_{41}+x_{49}+x_{55},
   \nonumber
    \\
    \nonumber
    \\
   z_{8}&= x_{8}+x_{12}+x_{14}+x_{16}+x_{18}+x_{26}+x_{34}+x_{42}+x_{50}+x_{56}
   \nonumber
    \\
   & \ \ +\ g(x_{11} ,x_{13}, x_{15}, x_{17}).
   \nonumber
\end{align}
Now, we show how the $i$-th decoder $\psi_{\mathcal{I}_{3}}^i$ recovers the requested message $x_{i}$ using the coded messages $z_{k}, k\in [8]$ along with the messages in its side information.
\begin{itemize}
    \item Each user $u_{i}, i\in [8]$ can directly decode its requested message $x_{i}$, from the coded message $z_{i}$.
    
    \item User $u_{9}$ decodes $(x_{9}+ x_{11}+ x_{13})$, $(x_{9}+ x_{11}+ x_{15})$, $(x_{9}+ x_{13}+ x_{15})$, and $(x_{11}+ x_{13}+ x_{15})$, respectively, from $z_{1}, z_{3}, z_{5}$ and $z_{7}$. It also adds $z_{2}+ z_{4}+ z_{6}+ z_{8}$ to achieve $g(x_{9}, x_{11}, x_{13}, x_{17})+$ $g(x_{9}, x_{11}, x_{15}, x_{17})+$ $g(x_{9}, x_{13}, x_{15}, x_{17})+$ $g(x_{11}, x_{13}, x_{15}, x_{17})$. Now, according to Lemma \ref{lem:g-function}, by having $x_{17}$, it is able to recover its requested message $x_{9}$.
    
    \item User $u_{10}$ first decodes $x_{9}$ and $x_{12}+x_{14}$, respectively, from $z_{1}$ and $z_{8}$. Then, it can decode its requested message $x_{10}$ from $z_{2}$.
    
    \item User $u_{11}$ decodes $(x_{9}+ x_{11}+ x_{13})$, $(x_{9}+ x_{11}+ x_{15})$, $(x_{9}+ x_{13}+ x_{15})$, and $(x_{11}+ x_{13}+ x_{15})$, respectively, from $z_{1}, z_{3}, z_{5}$ and $z_{7}$. It also adds $z_{2}+ z_{4}+ z_{6}+ z_{8}$ to achieve $g(x_{9}, x_{11}, x_{13}, x_{17})+$ $g(x_{9}, x_{11}, x_{15}, x_{17})+$ $g(x_{9}, x_{13}, x_{15}, x_{17})+$ $g(x_{11}, x_{13}, x_{15}, x_{17})$. Now, according to Lemma \ref{lem:g-function}, by having $x_{17}$, it is able to recover its requested message $x_{11}$.
    
    \item User $u_{12}$ first decodes $x_{11}$ and $x_{14}+x_{16}$, respectively, from $z_{3}$ and $z_{6}$. Then, it can decode its requested message $x_{12}$ from $z_{4}$.
    
     \item User $u_{13}$ decodes $(x_{9}+ x_{11}+ x_{13})$, $(x_{9}+ x_{11}+ x_{15})$, $(x_{9}+ x_{13}+ x_{15})$, and $(x_{11}+ x_{13}+ x_{15})$, respectively, from $z_{1}, z_{3}, z_{5}$ and $z_{7}$. It also adds $z_{2}+ z_{4}+ z_{6}+ z_{8}$ to achieve $g(x_{9}, x_{11}, x_{13}, x_{17})+$ $g(x_{9}, x_{11}, x_{15}, x_{17})+$ $g(x_{9}, x_{13}, x_{15}, x_{17})+$ $g(x_{11}, x_{13}, x_{15}, x_{17})$. Now, according to Lemma \ref{lem:g-function}, by having $x_{17}$, it is able to recover its requested message $x_{13}$.
    
    \item User $u_{14}$ first decodes $x_{13}$ and $x_{10}+x_{16}$, respectively, from $z_{5}$ and $z_{4}$. Then, it can decode its requested message $x_{14}$ from $z_{6}$.
    
    \item User $u_{15}$ decodes $(x_{9}+ x_{11}+ x_{13})$, $(x_{9}+ x_{11}+ x_{15})$, $(x_{9}+ x_{13}+ x_{15})$, and $(x_{11}+ x_{13}+ x_{15})$, respectively, from $z_{1}, z_{3}, z_{5}$ and $z_{7}$. It also adds $z_{2}+ z_{4}+ z_{6}+ z_{8}$ to achieve $g(x_{9}, x_{11}, x_{13}, x_{17})+$ $g(x_{9}, x_{11}, x_{15}, x_{17})+$ $g(x_{9}, x_{13}, x_{15}, x_{17})+$ $g(x_{11}, x_{13}, x_{15}, x_{17})$. Now, according to Lemma \ref{lem:g-function}, by having $x_{17}$, it is able to recover its requested message $x_{15}$.
    
    \item User $u_{16}$ first decodes $x_{15}$ and $x_{10}+x_{12}$, respectively, from $z_{7}$ and $z_{2}$. Then, it can decode its requested message $x_{16}$ from $z_{8}$.
    
    \item User $u_{17}$ can decode its desired message $x_{17}$ from $z_{1}$.
    
    \item User $u_{18}$ can decode its desired message $x_{18}$ from $z_{2}$.
    
    \item User $u_{19}$ first decodes $x_{9}$ from $z_{5}$. Then, it can decode $x_{19}$ from $z_{1}$.
    
    \item User $u_{20}$ first decodes $x_{9}$ from $z_{5}$. Then, it decodes $x_{10}$ from $z_{6}$. Finally, it can decode $x_{20}$ from $z_{2}$.
    
    \item User $u_{21}$ first decodes $x_{9}$ from $z_{1}$. Then, it can decode $x_{21}$ from $z_{3}$.
    
    \item User $u_{22}$ first decodes $x_{9}$ from $z_{1}$. Then, it decodes $x_{10}$ from $z_{2}$. Finally, it can decode $x_{22}$ from $z_{4}$.
    
    \item User $u_{23}$ first decodes $x_{9}$ from $z_{3}$. Then, it can decode $x_{23}$ from $z_{5}$.
    
    \item User $u_{24}$ first decodes $x_{9}$ from $z_{3}$. Then, it decodes $x_{10}$ from $z_{4}$. Finally, it can decode $x_{24}$ from $z_{6}$.
    
    \item User $u_{25}$ first decodes $x_{15}+x_{17}$ from $z_{3}$. Then, it can decode $x_{25}$ from $z_{7}$.
    
    \item User $u_{26}$ first decodes $x_{15}+x_{17}$ from $z_{3}$. Then, it adds $z_{6}$ and $2z_{8}$ to achieve $g(x_{9}+x_{13}+x_{15}+x_{17})+2g(x_{11}+x_{13}+x_{15}+x_{17})+2x_{26}$ (note, term $x_{16}+x_{18}$ is canceled out). Now, since $u_{26}$ knows $x_{9},x_{11},x_{13}$ and $x_{15}+x_{17}$, according to Lemma \ref{lem2:g-function}, it can achieve $g(x_{9}+x_{13}+x_{15}+x_{17})+2g(x_{11}+x_{13}+x_{15}+x_{17})$. Thus, it can decode its desired message $x_{26}$.
    
    \item User $u_{27}$ first decodes $x_{11}$ from $z_{7}$. Then, it can decode $x_{27}$ from $z_{1}$.
    
    \item User $u_{28}$ first decodes $x_{11}$ from $z_{7}$. Then, it decodes $x_{12}$ from $z_{8}$. Finally, it can decode $x_{28}$ from $z_{2}$.
    
    \item User $u_{29}$ first decodes $x_{11}$ from $z_{1}$. Then, it can decode $x_{29}$ from $z_{3}$.
    
    \item User $u_{30}$ first decodes $x_{11}$ from $z_{1}$. Then, it decodes $x_{12}$ from $z_{2}$. Finally, it can decode $x_{30}$ from $z_{4}$.
    
    \item User $u_{31}$ first decodes $x_{9}+x_{17}$ from $z_{1}$. Then, it can decode $x_{31}$ from $z_{5}$.
    
    \item User $u_{32}$ first decodes $x_{9}+x_{17}$ from $z_{1}$. Then, it adds $z_{4}$ and $2z_{6}$ to achieve $g(x_{9}+x_{11}+x_{15}+x_{17})+2g(x_{9}+x_{13}+x_{15}+x_{17})+2x_{32}$ (note, term $x_{10}+x_{18}$ is canceled out). Now, since $u_{32}$ knows $x_{11},x_{13},x_{15}$ and $x_{9}+x_{17}$, according to Lemma \ref{lem2:g-function}, it can achieve $g(x_{9}+x_{11}+x_{15}+x_{17})+2g(x_{9}+x_{13}+x_{15}+x_{17})$. Thus, it can decode its desired message $x_{32}$.
    
   \item User $u_{33}$ first decodes $x_{11}$ from $z_{3}$. Then, it can decode $x_{33}$ from $z_{7}$.
    
    \item User $u_{34}$ first decodes $x_{11}$ from $z_{3}$. Then, it decodes $x_{12}$ from $z_{4}$. Finally, it can decode $x_{34}$ from $z_{8}$.
    
    \item User $u_{35}$ first decodes $x_{13}$ from $z_{7}$. Then, it can decode $x_{35}$ from $z_{1}$.
    
    \item User $u_{36}$ first decodes $x_{13}$ from $z_{7}$. Then, it decodes $x_{14}$ from $z_{8}$. Finally, it can decode $x_{36}$ from $z_{2}$.
    
    \item User $u_{37}$ first decodes $x_{11}+x_{17}$ from $z_{7}$. Then, it can decode $x_{37}$ from $z_{3}$.
    
    \item User $u_{38}$ first decodes $x_{11}+x_{17}$ from $z_{7}$. Then, it adds $z_{2}$ and $2z_{4}$ to achieve $g(x_{9}+x_{11}+x_{13}+x_{17})+2g(x_{9}+x_{11}+x_{15}+x_{17})+2x_{38}$ (note, term $x_{12}+x_{18}$ is canceled out). Now, since $u_{38}$ knows $x_{9},x_{13},x_{15}$ and $x_{11}+x_{17}$, according to Lemma \ref{lem2:g-function}, it can achieve $g(x_{9}+x_{11}+x_{13}+x_{17})+2g(x_{9}+x_{11}+x_{15}+x_{17})$. Thus, it can decode its desired message $x_{38}$.
    
    \item User $u_{39}$ first decodes $x_{13}$ from $z_{1}$. Then, it can decode $x_{39}$ from $z_{5}$.
    
    \item User $u_{40}$ first decodes $x_{13}$ from $z_{1}$. Then, it decodes $x_{14}$ from $z_{2}$. Finally, it can decode $x_{40}$ from $z_{6}$.
    
    \item User $u_{41}$ first decodes $x_{13}$ from $z_{5}$. Then, it can decode $x_{41}$ from $z_{7}$.
    
    \item User $u_{42}$ first decodes $x_{13}$ from $z_{5}$. Then, it decodes $x_{14}$ from $z_{6}$. Finally, it can decode $x_{42}$ from $z_{8}$.
    
    \item User $u_{43}$ first decodes $x_{13}+x_{17}$ from $z_{5}$. Then, it can decode $x_{43}$ from $z_{1}$.
    
    \item User $u_{44}$ first decodes $x_{13}+x_{17}$ from $z_{5}$. Then, it adds $z_{6}$ and $2z_{2}$ to achieve $g(x_{9}+x_{13}+x_{15}+x_{17})+2g(x_{9}+x_{11}+x_{13}+x_{17})+2x_{44}$ (note, term $x_{14}+x_{18}$ is canceled out). Now, since $u_{44}$ knows $x_{9},x_{11},x_{15}$ and $x_{13}+x_{17}$, according to Lemma \ref{lem2:g-function}, it can achieve $g(x_{9}+x_{13}+x_{15}+x_{17})+2g(x_{9}+x_{11}+x_{13}+x_{17})$. Thus, it can decode its desired message $x_{44}$.
    
    \item User $u_{45}$ first decodes $x_{15}$ from $z_{7}$. Then, it can decode $x_{45}$ from $z_{3}$.
    
    \item User $u_{46}$ first decodes $x_{15}$ from $z_{7}$. Then, it decodes $x_{16}$ from $z_{8}$. Finally, it can decode $x_{46}$ from $z_{4}$.
    
    \item User $u_{47}$ first decodes $x_{15}$ from $z_{3}$. Then, it can decode $x_{47}$ from $z_{5}$.
    
    \item User $u_{48}$ first decodes $x_{15}$ from $z_{3}$. Then, it decodes $x_{16}$ from $z_{4}$. Finally, it can decode $x_{48}$ from $z_{6}$.
    
    \item User $u_{49}$ first decodes $x_{15}$ from $z_{5}$. Then, it can decode $x_{49}$ from $z_{7}$.
    
    \item User $u_{50}$ first decodes $x_{15}$ from $z_{5}$. Then, it decodes $x_{16}$ from $z_{6}$. Finally, it can decode $x_{50}$ from $z_{8}$.
    
    \item User $u_{51}$ first decodes $x_{9}+x_{17}$ from $z_{1}$. Then, it can decode $x_{51}$ from $z_{3}$.
    
    \item User $u_{52}$ first decodes $x_{9}+x_{17}$ from $z_{1}$. Then, it adds $z_{2}$ and $2z_{4}$ to achieve $g(x_{9}+x_{11}+x_{13}+x_{17})+2g(x_{9}+x_{11}+x_{15}+x_{17})+2x_{52}$ (note, term $x_{10}+x_{18}$ is canceled out). Now, since $u_{52}$ knows $x_{11},x_{13},x_{15}$ and $x_{9}+x_{17}$, according to Lemma \ref{lem2:g-function}, it can find $g(x_{9}+x_{11}+x_{13}+x_{17})+2g(x_{9}+x_{11}+x_{15}+x_{17})$. Thus, it can decode its desired message $x_{52}$.
    
    \item User $u_{53}$ first decodes $x_{11}+x_{17}$ from $z_{7}$. Then, it can decode $x_{53}$ from $z_{1}$.
    
    \item User $u_{54}$ first decodes $x_{11}+x_{17}$ from $z_{7}$. Then, it adds $z_{8}$ and $2z_{2}$ to achieve $g(x_{11}+x_{13}+x_{15}+x_{17})+2g(x_{9}+x_{11}+x_{13}+x_{17})+2x_{54}$ (note, term $x_{12}+x_{18}$ is canceled out). Now, since $u_{54}$ knows $x_{9},x_{13},x_{15}$ and $x_{11}+x_{17}$, according to Lemma \ref{lem2:g-function}, it can find $g(x_{11}+x_{13}+x_{15}+x_{17})+2g(x_{9}+x_{11}+x_{13}+x_{17})$. Thus, it can decode its desired message $x_{54}$.
    
    \item User $u_{55}$ first decodes $x_{13}+x_{17}$ from $z_{5}$. Then, it can decode $x_{55}$ from $z_{7}$.
    
    \item User $u_{56}$ first decodes $x_{13}+x_{17}$ from $z_{5}$. Then, it adds $z_{6}$ and $2z_{8}$ to achieve $g(x_{9}+x_{13}+x_{15}+x_{17})+2g(x_{11}+x_{13}+x_{15}+x_{17})+2x_{56}$ (note, term $x_{14}+x_{18}$ is canceled out). Now, since $u_{56}$ knows $x_{9},x_{11},x_{15}$ and $x_{13}+x_{17}$, according to Lemma \ref{lem2:g-function}, it can find $g(x_{9}+x_{13}+x_{15}+x_{17})+2g(x_{11}+x_{13}+x_{15}+x_{17})$. Thus, it can decode its desired message $x_{56}$.
    
    \item User $u_{57}$ first decodes $x_{15}+x_{17}$ from $z_{3}$. Then, it can decode $x_{57}$ from $z_{5}$.
    
    \item User $u_{58}$ first decodes $x_{15}+x_{17}$ from $z_{3}$. Then, it adds $z_{4}$ and $2z_{6}$ to achieve $g(x_{9}+x_{11}+x_{15}+x_{17})+2g(x_{9}+x_{13}+x_{15}+x_{17})+2x_{58}$ (note, term $x_{16}+x_{18}$ is canceled out). Now, since $u_{58}$ knows $x_{9},x_{11},x_{13}$ and $x_{15}+x_{17}$, according to Lemma \ref{lem2:g-function}, it can find $g(x_{9}+x_{11}+x_{15}+x_{17})+2g(x_{9}+x_{13}+x_{15}+x_{17})$. Thus, it can decode its desired message $x_{58}$.
\end{itemize}
\end{proof}

\section{Conclusion}
The suboptimality of linear coding rate for the general index coding problem is due to its dependency on the field size. This dependency has been illustrated through the two well-known matroid instances, namely the Fano and non-Fano matroids, which, in turn, limits its scope only to fields with characteristic two.
In this paper, this scope of dependency was extended to the fields with characteristic three by designing two index coding instances of size 29 such that for the first instance, linear coding is optimal only over the fields with characteristic three, while for the second instance, linear coding is optimal over fields with any characteristic other than characteristic three. For each instance, it was shown that the key constraints on the column space of its encoding matrix can be captured by a matroid with the ground set of size 9, for which the existence of its linear representation is dependent on the fields with characteristic three. Presenting the
proofs and discussions using these two relatively small matroids is helpful in
pointing out the key constraints causing the linear coding rate to
become dependent on the field size. Finally, we designed the third index coding instance of size 58 such that while linear coding cannot achieve its optimal rate over fields with characteristic three, there exists an optimal nonlinear code over fields with characteristic three. It was shown that connecting the first and third index coding instances in two specific ways, called no-way and two-way connections, will lead to two new index coding instances of size 87 and 91, for which linear coding is outperformed by nonlinear codes.

\appendices

\section{Proof of Lemmas \ref{lem:MAIS1}-\ref{lem:col(9)}} \label{app:proof- Lemma1-5}

\begin{rem} \label{rem:dec-con}
It can be verified that the decoding condition in \eqref{eq:dec-cond} along with the properties of the $\mathrm{rank}$ function gives the following results.
\begin{align}
    &\mathrm{rank} (\boldsymbol{H}^{\{i\}\cup M_{}})= \mathrm{rank} ( \boldsymbol{H}^{M_{}}) + t,  \forall M_{}\subseteq B_{i}, \forall i\in [m],  \label{eq:rem:dec-con-1} 
    \\
    &\mathrm{rank} (\boldsymbol{H}^{\{i\}})= t, \ \ \ \ \ \ \ \ \ \ \ \ \ \ \ \ \ \ \forall i\in [m], \label{eq:rem:dec-con-2}
    \\
    &\mathrm{rank} (\boldsymbol{H}^{M_{1}})\leq \mathrm{rank} ( \boldsymbol{H}^{M_2}), \ \ \ \ \ \forall M_1 \subseteq M_2 \subseteq [m]. \label{eq:rem:dec-con-3}
\end{align}
\end{rem}

\subsection{Proof of Lemma \ref{lem:MAIS1}}
If $M$ is an acyclic set, then we can find a sequence of its elements $i_{1},\dots,i_{|M|}\in M$ such that $M_{j}\subseteq B_{i_{j}}, \forall j\in [|M|]$, where $M_{j}=\{i_{j+1},\dots,i_{|M|}\}, \forall j\in [|M|-1]$ and $M_{|M|}=\emptyset$. Note $M=\{i_{1}\}\cup M_1$ and $M_{j}=\{i_{j+1}\}\cup M_{j+1}, \forall j\in [|M|-1]$. By applying the condition in \eqref{eq:rem:dec-con-1} for each $i=i_{1},\dots,i_{|M|}$, we have
\begin{align}
    \mathrm{rank}(\boldsymbol{H}^{M=\{i_{1}\}\cup M_{1}})&= \mathrm{rank}(\boldsymbol{H}^{M_{1}=\{i_{2}\}\cup M_{2}})+t \nonumber\\
    &= \mathrm{rank}(\boldsymbol{H}^{M_{2}=\{i_{3}\}\cup M_{3}})+2t \nonumber \\
    &=\dots \nonumber \\
    &= |M|t \nonumber,
\end{align}
which means that $M$ is a basis set of $\boldsymbol{H}$.

\subsection{Proof of Lemma \ref{lem:min-cyc}}
First, note that for any $l\in M$, set $M\backslash\{l\}$ is an acyclic set. Then, according to Lemma \ref{lem:MAIS1}, 
\begin{equation} \label{eq:lem3}
    \mathrm{rank}(\boldsymbol{H}^{M\backslash\{l\}})=(|M|-1)t, \ \ \forall l\in M.
\end{equation}
So, having $\mathrm{rank}(\boldsymbol{H}^{M})=(|M|-1)t$ requires $\boldsymbol{H}^{\{l\}}=\sum_{i\in M\backslash\{l\}} \boldsymbol{H}^{\{i\}} \boldsymbol{M}_{l,i}$. Now, if one of the $\boldsymbol{M}_{l,i}, i\in M\backslash \{l\}$ is not invertible, then $\mathrm{rank} (\boldsymbol{H}^{M\backslash\{i\}})<(|M|-1)t$, which contradicts \eqref{eq:lem3}. Thus, each $\boldsymbol{M}_{l,i}$ must be invertible, which means that $M$ is a circuit set of $\boldsymbol{H}$.

\subsection{Proof of Lemma \ref{lem:MAIS2}}
First, because $M$ is an independent set, then $M\backslash \{i\}\subseteq B_{i}, \forall i\in M$. Moreover, since $M$ is an acyclic set of $\mathcal{I}$, then according to Lemma \ref{lem:MAIS1}, $\mathrm{rank}(\boldsymbol{H}^{M})=|M|t$. Now, in order to have $\mathrm{col}(\boldsymbol{H}^{\{j\}})\subseteq \mathrm{col}(\boldsymbol{H}^{M})$, one must have $\boldsymbol{H}^{\{j\}}= \sum_{i\in M} \boldsymbol{H}^{\{i\}}\boldsymbol{M}_{j,i}$. Since $ j\in [m]\backslash M$ and $j\in B_{i}$ for some $i \in  M\backslash \{l\}$, then $\{j\}\cup M\backslash \{l\}\subseteq B_{l}$. Now, assume $\boldsymbol{M}_{j,i}$ is a nonzero matrix (i.e., $\mathrm{rank}(\boldsymbol{M}_{j,i})\geq 1$). Then, 
\begin{align}
    \mathrm{rank}(\boldsymbol{H}^{\{j\}\cup M})&=
    \mathrm{rank}(\boldsymbol{H}^{\{l\}\cup (\{j\}\cup  M\backslash\{l\})})
    \nonumber
    \\
    &=\mathrm{rank}(\boldsymbol{H}^{\{j\}\cup M\backslash \{l\}})+t \label{lem:pr:01} 
    \\
    &=\mathrm{rank}(\left [\begin{array}{c|c}
        \boldsymbol{H}^{\{j\}} & \boldsymbol{H}^{M\backslash \{l\}}
      \end{array}
     \right ])+t
     \nonumber
    \\
    &=\mathrm{rank}(\left [\begin{array}{c|c}
        \sum_{i\in M} \boldsymbol{H}^{\{i\}}\boldsymbol{M}_{j,i} & \boldsymbol{H}^{M\backslash \{l\}}
      \end{array}
     \right ])+t
     \nonumber
    \\
    &\geq \mathrm{rank}(\left [\begin{array}{c|c}
        \boldsymbol{H}^{\{l\}}\boldsymbol{M}_{j,l} & \boldsymbol{H}^{M\backslash \{l\}}
      \end{array}
     \right ])+t 
     \label{lem:pr:02} 
    \\
    &=\mathrm{rank}(\boldsymbol{H}^{\{l\}}\boldsymbol{M}_{j,l})+ (|M|-1)t+t
    \label{lem:pr:03}
    \\
    &> |M|t, 
    \label{lem:pr:04}
\end{align}
where \eqref{lem:pr:01} is due to \eqref{eq:rem:dec-con-1}, \eqref{lem:pr:02} is because of the property of the $\mathrm{rank}$ function by removing the term $\sum_{i\in M\backslash \{l\}} \boldsymbol{H}^{\{i\}}\boldsymbol{M}_{j,i}$ from $\sum_{i\in M} \boldsymbol{H}^{\{i\}}\boldsymbol{M}_{j,l}$ as it is a linear combination of the columns of $\boldsymbol{H}^{M\backslash \{l\}}$. \eqref{lem:pr:03} is based on Lemma \ref{lem:MAIS1} and the fact that $M$ is an acyclic set of $\mathcal{I}$. Thus, the column space of $\boldsymbol{H}^{\{l\}}$ is linearly independent of column space of $\boldsymbol{H}^{M\backslash \{l\}}$. Finally, \eqref{lem:pr:04} is due to the fact that $\boldsymbol{H}^{\{l\}}$ is invertible and $\mathrm{rank}( \boldsymbol{M}_{j,i})\geq 1$. The result in \eqref{lem:pr:04} contradicts the assumption that $\mathrm{rank}( \boldsymbol{H}^{\{j\}\cup M})=|M|t$, and hence, we must have $\boldsymbol{M}_{j,i}=\boldsymbol{0}_{t}$. The same argument for $i\in M\backslash \{l\}$ gives $\boldsymbol{M}_{j,i}=\boldsymbol{0}_{t}, \forall i\in M\backslash \{l\}$. Therefore, $\boldsymbol{H}^{\{j\}}=\boldsymbol{H}^{\{l\}} \boldsymbol{M}_{j,l}$ and $\boldsymbol{M}_{j,l}$ must be invertible to have $\mathrm{rank}\ \boldsymbol{H}^{\{j\}}=t$.

\subsection{Proof of Lemma \ref{lem:independent-cycle}}
Corollaries \ref{cor:j-subspace-M-circuit}-\ref{rem:acyclic-j-acyclic} can be derived from earlier results and will be used in the proof of Lemma \ref{lem:independent-cycle}.

\begin{cor} \label{cor:j-subspace-M-circuit}
Let $M$ be an independent set of $\boldsymbol{H}$. Now, if $\mathrm{col}(\boldsymbol{H}^{\{j\}})\subseteq \mathrm{col}(\boldsymbol{H}^{M})$, then there exists one subset $M^{\prime}\subseteq M$, such that $\{j\}\cup M^{\prime}$ is a circuit set of $\boldsymbol{H}$.
\end{cor}
\begin{proof}
Since $\mathrm{col}(\boldsymbol{H}^{\{j\}})\subseteq \mathrm{col}(\boldsymbol{H}^{M})$, we must have $\boldsymbol{H}^{\{j\}}=\sum_{i\in M} \boldsymbol{H}^{\{i\}}\boldsymbol{M}_{j,i}$ such that only matrices $\boldsymbol{M}_{j,i}, i\in M^{\prime}\subseteq M$ are invertible. Thus, according to Definition \ref{def:Basis-Circuit-Matrix}, set $\{j\}\cup M^{\prime}$ forms a circuit set of $\boldsymbol{H}$. 
\end{proof}

\begin{cor} \label{cor:minimal-cyclic-acylic}
If $M$ is a minimal cyclic set of $\mathcal{I}$, then according to Definitions \ref{def:minimal-cyclic-set} and \ref{def:acyclic-set}, any of its proper subsets $M^{\prime}\subset M$ will be an acyclic set of $\mathcal{I}$. 
\end{cor}

\begin{cor}[\cite{Arbabjolfaei2018}]\label{cor:acyclic-l}
If $M$ is an acyclic set of $\mathcal{I}$, then there exists at least one $l\in M$ such that $M\backslash\{l\}\subseteq B_{l}$. 
\end{cor}

\begin{cor} \label{rem:acyclic-j-acyclic}
Suppose $\{j\}\cup M$ is a circuit set of $\boldsymbol{H}$. Then for any $l\in M$, we have $\mathrm{col}(\boldsymbol{H}^{\{j\}\cup M\backslash \{l\}})=\mathrm{col}(\boldsymbol{H}^{M})$.
\end{cor}
\begin{proof}
Since $\{j\}\cup M$ is a circuit set of $\boldsymbol{H}$, we have $\boldsymbol{H}^{\{j\}}=\sum_{i\in M} \boldsymbol{H}^{\{j\}}\boldsymbol{M}_{j,i}$ such that each $\boldsymbol{M}_{j,i}$ is invertible. Now, assume $l\in M$. Then, we have 
\begin{align}
    \mathrm{col}(\boldsymbol{H}^{\{j\}\cup M\backslash \{l\}})
    &=\mathrm{col}([\boldsymbol{H}^{\{j\}}| \boldsymbol{H}^{M\backslash\{l\}}])
    \nonumber
    \\
    &=\mathrm{col}([\sum_{i\in M} \boldsymbol{H}^{\{i\}}\boldsymbol{M}_{j,i}| \boldsymbol{H}^{M\backslash\{l\}}])
    \nonumber
    \\
    &=\mathrm{col}([\boldsymbol{H}^{\{l\}}\boldsymbol{M}_{j,l}| \boldsymbol{H}^{M\backslash\{l\}}])
    \nonumber
    \\
    &=\mathrm{col}(\boldsymbol{H}^{M}),
    \label{eq:corollary-4}
\end{align}
where \eqref{eq:corollary-4} is due to the invertibility of $\boldsymbol{M}_{j,l}$.
\end{proof}

\subsubsection{Proof of Lemma \ref{lem:independent-cycle}}

Since $M$ is an independent set of $\boldsymbol{H}$, and $\mathrm{col}(\boldsymbol{H}^{\{j\}})\subseteq \mathrm{col}(\boldsymbol{H}^{M})$, then according to Corollary \ref{cor:j-subspace-M-circuit}, there exists a subset $M^{\prime}\subseteq M$ such that $\{j\}\cup M^{\prime}$ is a circuit set. Now, we show that $M^{\prime}=M$, otherwise it leads to a contradiction. Assume $M^{\prime}\subset M$.
First, since $M$ is a minimal cyclic set of $\mathcal{I}$, based on Corollary \ref{cor:minimal-cyclic-acylic}, $M^{\prime}\subset M$ is an acyclic set of $\mathcal{I}$. Second, according to Corollary \ref{cor:acyclic-l}, there exists $l\in M^{\prime}$ such that $M^{\prime}\backslash\{l\}\subseteq B_{l}$. Also, due to $j\in B_{l}$, we get $\{j\}\cup (M^{\prime}\backslash\{l\})\subseteq B_{l}$. Thus,
\begin{align}
  l\in M^{\prime} &\rightarrow \mathrm{col}(\boldsymbol{H}^{\{l\}}) \subseteq \mathrm{col}(\boldsymbol{H}^{M^{\prime}}),
  \nonumber
  \\
  \{j\}\cup (M^{\prime}\backslash\{l\})\subseteq B_{l}&\rightarrow \mathrm{col}(\boldsymbol{H}^{\{j\}\cup (M^{\prime}\backslash\{l\})}) \subseteq \mathrm{col}(\boldsymbol{H}^{B_{l}}).
    \label{eq:rem:acyclic-l}
\end{align}
Third, since $\{j\}\cup M^{\prime}$ is a circuit set, Corollary \ref{rem:acyclic-j-acyclic} leads to
\begin{align}
    \mathrm{col}(\boldsymbol{H}^{\{j\}\cup M^{\prime}\backslash \{l\}})=\mathrm{col}(\boldsymbol{H}^{M^{\prime}}).
    \label{eq:rem:acyclic-j-acyclic}
\end{align}
Now, from \eqref{eq:rem:acyclic-l} and \eqref{eq:rem:acyclic-j-acyclic}, we have
\begin{align}
    \mathrm{col}(\boldsymbol{H}^{\{l\}}) \subseteq \mathrm{col}(\boldsymbol{H}^{M^{\prime}})=\mathrm{col}(\boldsymbol{H}^{\{j\}\cup M^{\prime}\backslash \{l\}})\subseteq \mathrm{col}(\boldsymbol{H}^{B_{l}}),
    \nonumber
\end{align}
which contradicts the decoding condition in \eqref{eq:dec-cond} for user $u_{l}$ as $\mathrm{col}(\boldsymbol{H}^{\{l\}}) \subseteq \mathrm{col}(\boldsymbol{H}^{B_{l}})$. Therefore, we must have $M^{\prime}=M$.

\subsection{Proof of Lemma \ref{lem:col(9)}}
Since $\mathrm{col}(\boldsymbol{H}^{\{9\}})\subseteq \mathrm{col}(\boldsymbol{H}^{\{1,8\}})$, we have 
\begin{align}
    \boldsymbol{H}^{\{9\}}=\boldsymbol{H}^{\{1\}}\boldsymbol{M}_{9,1}+\boldsymbol{H}^{\{8\}}\boldsymbol{M}_{9,8}.
    \label{eq:lem:last-conversion-9-1}
\end{align}
Moreover, since set $\{2,3,4,8\}$ is a circuit set, we must have
\begin{align}
    \boldsymbol{H}^{\{8\}}=\boldsymbol{H}^{\{2\}}\boldsymbol{M}_{8,2}+\boldsymbol{H}^{\{3\}}\boldsymbol{M}_{8,3}+\boldsymbol{H}^{\{4\}}\boldsymbol{M}_{8,4},
    \label{eq:lem:last-conversion-9-2}
\end{align}
where each $\boldsymbol{M}_{8,2},\boldsymbol{M}_{8,3},\boldsymbol{M}_{8,4}$ is invertible. Thus, based on \eqref{eq:lem:last-conversion-9-1} and \eqref{eq:lem:last-conversion-9-2}, $\boldsymbol{H}^{\{9\}}$ is equal to
\begin{align}
    &\boldsymbol{H}^{\{1\}}\boldsymbol{M}_{9,1}+(\boldsymbol{H}^{\{2\}}\boldsymbol{M}_{8,2}+\boldsymbol{H}^{\{3\}}\boldsymbol{M}_{8,3}+\boldsymbol{H}^{\{4\}}\boldsymbol{M}_{8,4})\boldsymbol{M}_{9,8}
    \nonumber
    \\
    &=\boldsymbol{H}^{\{1\}}\boldsymbol{M}_{9,1}+\boldsymbol{H}^{\{2\}}\boldsymbol{M}_{8,2}^{\prime}+\boldsymbol{H}^{\{3\}}\boldsymbol{M}_{8,3}^{\prime}+\boldsymbol{H}^{\{4\}}\boldsymbol{M}_{8,4}^{\prime},
    \label{eq:lem3:1}
\end{align}
where $\boldsymbol{M}_{8,i}^{\prime}=\boldsymbol{M}_{8,i}\boldsymbol{M}_{9,8}, i=2,3,4$. On the other hand, since $\{1,3,4,7\}$ is a circuit set, we get 
\begin{align}
    \boldsymbol{H}^{\{7\}}=\boldsymbol{H}^{\{1\}}\boldsymbol{M}_{7,1}+\boldsymbol{H}^{\{3\}}\boldsymbol{M}_{7,3}+\boldsymbol{H}^{\{4\}}\boldsymbol{M}_{7,4},
    \label{eq:lem3:2}
\end{align}
where each $\boldsymbol{M}_{7,1},\boldsymbol{M}_{7,3},\boldsymbol{M}_{7,4}$ is invertible. Now, for set $\{2,7,9\}$, we have
\begin{align}
    \mathrm{rank}(\boldsymbol{H}^{\{2,7,9\}})
    &=\mathrm{rank}([\boldsymbol{H}^{\{2\}}|\boldsymbol{H}^{\{7\}}|\boldsymbol{H}^{\{9\}}])
    \nonumber
    \\
    &=\mathrm{rank}([\boldsymbol{H}^{\{2\}}|\boldsymbol{H}^{\{7\}}|\boldsymbol{H}^{\{9\}}-\boldsymbol{H}^{\{2\}}\boldsymbol{M}_{8,2}^{\prime}])
    \nonumber
    \\
    &=t +\mathrm{rank}([\boldsymbol{H}^{\{7\}}|\boldsymbol{H}^{\{9\}}-\boldsymbol{H}^{\{2\}}\boldsymbol{M}_{8,2}^{\prime}]).
    \label{eq:lem3:3}
\end{align}
Now,
since $\mathrm{col}(\boldsymbol{H}^{\{9\}})$ must be a subspace of $\mathrm{col}(\boldsymbol{H}^{\{2,7,9\}})$, we must have $\mathrm{rank}(\boldsymbol{H}^{\{2,7,9\}})=2t$. Thus, in \eqref{eq:lem3:3}, $\boldsymbol{H}^{\{9\}}-\boldsymbol{H}^{\{2\}}\boldsymbol{M}_{8,2}^{\prime}$ must be linearly dependent on $\boldsymbol{H}^{\{7\}}$, which based on \eqref{eq:lem3:1} and \eqref{eq:lem3:2} requires each $\boldsymbol{M}_{9,1}, \boldsymbol{M}_{8,3}^{\prime}, \boldsymbol{M}_{8,4}^{\prime}$ to be invertible. 
Besides, each $\boldsymbol{M}_{8,3}^{\prime}, \boldsymbol{M}_{8,4}^{\prime}$ is invertible only if $\boldsymbol{M}_{9,8}$ is invertible. Thus, since both $\boldsymbol{M}_{9,1}$ and $\boldsymbol{M}_{9,8}$ are invertible, set $\{1,8,9\}$ forms a circuit set of $\boldsymbol{H}$. Similarly, it can be shown that all sets $\{2,7,9\}$, $\{3,6,9\}$ and $\{4,5,9\}$ are circuit sets.

\section{Proof of Propositions \ref{prop-thm1-1} and \ref{prop-thm2-1}} \label{app:proof-prop-H-I1-I2}

\subsection{Proof of Propositions \ref{prop-thm1-1}}
We show that matrix $\boldsymbol{H}_{\ast}\in \mathbb{F}_{q}^{4\times 29}$, presented in Figure \ref{fig:1}, will satisfy all users $u_{i}, in\in [29]$ of the index coding instance $\mathcal{I}_{1}$ if the field $\mathbb{F}_{q}$ has characteristic three. Let $\boldsymbol{y}=\boldsymbol{H}_{\ast}\boldsymbol{x}$ and $\mathbb{F}_{q}=GF(3)$. 
Now, we show that encoding matrix $\boldsymbol{H}_{\ast}$ satisfies the decoding condition in \eqref{eq:dec-cond} for all $i\in [29]$. Figures \ref{fig:333} and \ref{fig:334} present $\boldsymbol{H}_{\ast}^{\{i\}\cup B_{i}}$ for all $i\in [29]$.

\begin{itemize}
    \item It can be seen that user $u_{i}, i\in [4]$ can decode its desired message $x_{i}$ from the coded message $y_{i}$ (or the $i$-th row of $\boldsymbol{H}_{\ast}$).
    
    \item User $u_{5}$ first decodes $x_{7}+x_{8}$ from $y_{4}$. Then, it can decode $x_{5}$ from $y_{3}$.
    
    \item User $u_{6}$ first decodes $x_{5}+x_{8}$ from $y_{3}$. Then, it can decode $x_{6}$ from $y_{2}$.

    \item User $u_{7}$ first decodes $x_{5}+x_{6}$ from $y_{2}$. Then, it can decode $x_{7}$ from $y_{1}$.
    
    \item User $u_{8}$ first decodes $x_{6}+x_{7}$ from $y_{1}$. Then, it can decode $x_{8}$ from $y_{4}$.
    
    \item User $u_{9}$ adds $y_{1},y_{2},y_{3},y_{4}$ to achieve $x_{9}$ as follows
    \begin{align}
        y_{1}+y_{2}+y_{3}+y_{4}&=3(x_{5}+x_{6}+x_{7}+x_{8})+4x_{9}
        \nonumber
        \\
        &=x_{9},
        \label{eq:asdfasdf}
    \end{align}
    where \eqref{eq:asdfasdf} follows from the fact that $3=0$ and $4=1$ over the Galois field $GF(3)$.
    
    \item User $u_{10}$ first decodes $x_{5}$ from $y_{3}$. Then, it can decode $x_{10}$ from $y_{1}$.
    
    \item User $u_{11}$ decodes $x_{5}$ from $y_{1}$. Then, it can decode $x_{11}$ from $y_{2}$.
    
    \item User $u_{12}$ decodes $x_{5}$ from $y_{2}$. Then, it can decode $x_{12}$ from $y_{3}$.
    
    \item User $u_{13}$ first decodes $x_{8}+x_{9}$ from $y_{2}$ or $y_{3}$. Then, it can decode $x_{13}$ from $y_{4}$.
    
    \item User $u_{14}$ first decodes $x_{6}$ from $y_{4}$. Then, it can decode $x_{14}$ from $y_{1}$.
    
    \item User $u_{15}$ first decodes $x_{6}$ from $y_{1}$. Then, it can decode $x_{15}$ from $y_{2}$.
    
    \item User $u_{16}$ first decodes $x_{5}+x_{9}$ from $y_{1}$ or $y_{2}$. Then, it can decode $x_{16}$ from $y_{3}$.
    
    \item User $u_{17}$ first decodes $x_{6}$ from $y_{2}$. Then, it can decode $x_{17}$ from $y_{4}$.
    
    \item User $u_{18}$ first decodes $x_{7}$ from $y_{4}$. Then, it can decode $x_{18}$ from $y_{1}$.
    
    \item User $u_{19}$ first decodes $x_{6}+x_{9}$ from $y_{1}$ or $y_{4}$. Then, it can decode $x_{19}$ from $y_{2}$.
    
    \item User $u_{20}$ first decodes $x_{7}$ from $y_{1}$. Then, it can decode $x_{20}$ from $y_{3}$.
    
    \item User $u_{21}$ first decodes $x_{7}$ from $y_{3}$. Then, it can decode $x_{21}$ from $y_{4}$.
    
    \item User $u_{22}$ first decodes $x_{7}+x_{9}$ from $y_{3}$ or $y_{4}$. Then, it can decode $x_{22}$ from $y_{1}$.
    
    \item User $u_{23}$ first decodes $x_{8}$ from $y_{4}$. Then, it can decode $x_{23}$ from $y_{2}$.
    
    \item User $u_{24}$ first decodes $x_{8}$ from $y_{2}$. Then, it can decode $x_{24}$ from $y_{3}$.
    
    \item User $u_{25}$ first decodes $x_{8}$ from $y_{3}$. Then, it can decode $x_{25}$ from $y_{4}$.
    
    \item User $u_{26}$ first decodes $x_{5}+x_{9}$ from $y_{1}$. Then, it can decode $x_{26}$ from $y_{2}$.
    
    \item User $u_{27}$ first decodes $x_{6}+x_{9}$ from $y_{4}$. Then, it can decode $x_{27}$ from $y_{1}$.
     
    \item User $u_{28}$ first decodes $x_{7}+x_{9}$ from $y_{3}$. Then, it can decode $x_{28}$ from $y_{4}$.
    
    \item User $u_{29}$ first decodes $x_{8}+x_{9}$ from $y_{2}$. Then, it can decode $x_{29}$ from $y_{3}$.
    
\end{itemize}
    
\begin{figure*}
\centering

\subfloat[$\boldsymbol{H}_{\ast}^{\{1\}\cup B_{1}}$]
{
\begin{blockarray}{cccccccccccccccccc}
              & 
              \textcolor{blue}{{\tiny 1}}\hspace{0ex} & 
              \textcolor{blue}{{\tiny 2}}\hspace{0ex} & 
              \textcolor{blue}{{\tiny 3}}\hspace{0ex} & 
              \textcolor{blue}{{\tiny 4}}\hspace{0ex} & 
              \textcolor{blue}{{\tiny 8}}\hspace{0ex} & 
              \textcolor{blue}{{\tiny 11}}\hspace{0ex} &
              \textcolor{blue}{{\tiny 12}}\hspace{0ex} &
              \textcolor{blue}{{\tiny 13}}\hspace{0ex} & 
              \textcolor{blue}{{\tiny 15}}\hspace{0ex} &
              \textcolor{blue}{{\tiny 16}}\hspace{0ex} &
              \textcolor{blue}{{\tiny 17}}\hspace{0ex} & 
              \textcolor{blue}{{\tiny 19}}\hspace{0ex} &
              \textcolor{blue}{{\tiny 20}}\hspace{0ex} &
              \textcolor{blue}{{\tiny 21}}\hspace{0ex} & 
              \textcolor{blue}{{\tiny 23}}\hspace{0ex} &
              \textcolor{blue}{{\tiny 24}}\hspace{0ex} &
              \textcolor{blue}{{\tiny 25}}\hspace{0ex}  
            \\
\begin{block}{c[ccccccccccccccccc]}
               \textcolor{blue}{{\tiny 1}} & 
               1\hspace{0ex} & 
               0\hspace{0ex} & 0\hspace{0ex} & 
               0\hspace{0ex} &
               
               0\hspace{0ex} & 
               
               0\hspace{0ex} & 0\hspace{0ex} & 0\hspace{0ex} & 
               
               0\hspace{0ex} & 0\hspace{0ex} & 0\hspace{0ex} &
               
               0\hspace{0ex} & 0\hspace{0ex} & 0\hspace{0ex} &
               
               0\hspace{0ex} & 0\hspace{0ex} & 0\hspace{0ex} 
               \\
               \textcolor{blue}{{\tiny 2}} & 
               0\hspace{0ex} & 1\hspace{0ex} & 0\hspace{0ex} & 0\hspace{0ex} & 
               
               1\hspace{0ex} &
               
               1\hspace{0ex} & 0\hspace{0ex} & 0\hspace{0ex} &
              
               1\hspace{0ex} & 0\hspace{0ex} & 0\hspace{0ex} &
               
               1\hspace{0ex} & 0\hspace{0ex} & 0\hspace{0ex} &
               
               1\hspace{0ex} & 0\hspace{0ex} & 0\hspace{0ex} 
               \\
               \textcolor{blue}{{\tiny 3}} & 
               0\hspace{0ex} & 0\hspace{0ex} & 1\hspace{0ex} & 0\hspace{0ex} & 
               
               1\hspace{0ex} &
               
               0\hspace{0ex} & 1\hspace{0ex} & 0\hspace{0ex}  &
               
               0\hspace{0ex} & 1\hspace{0ex} & 0\hspace{0ex}  &
               
                0\hspace{0ex} & 1\hspace{0ex} & 0\hspace{0ex} &
               
                0\hspace{0ex} & 1\hspace{0ex} & 0\hspace{0ex} 
               \\
               \textcolor{blue}{{\tiny 4}} & 
               0\hspace{0ex} & 0\hspace{0ex} & 0\hspace{0ex} & 1\hspace{0ex} & 
               
                1\hspace{0ex} & 
              
                0\hspace{0ex} & 0\hspace{0ex} & 1\hspace{0ex} &
               
                0\hspace{0ex} & 0\hspace{0ex} & 1\hspace{0ex} &
               
                0\hspace{0ex} & 0\hspace{0ex} & 1\hspace{0ex} &
               
                0\hspace{0ex} & 0\hspace{0ex} & 1\hspace{0ex} 
               \\
             \end{block}
\end{blockarray}
}
\\
\subfloat[$\boldsymbol{H}_{\ast}^{\{2\}\cup B_{2}}$]
{
\begin{blockarray}{cccccccccccccccccc}
              & 
              \textcolor{blue}{{\tiny 1}}\hspace{0ex} & 
              \textcolor{blue}{{\tiny 2}}\hspace{0ex} & 
              \textcolor{blue}{{\tiny 3}}\hspace{0ex} & 
              \textcolor{blue}{{\tiny 4}}\hspace{0ex} & 
              
              \textcolor{blue}{{\tiny 7}}\hspace{0ex} &
              
              \textcolor{blue}{{\tiny 10}}\hspace{0ex} &
              \textcolor{blue}{{\tiny 12}}\hspace{0ex} &
              \textcolor{blue}{{\tiny 13}}\hspace{0ex} & 
              
              \textcolor{blue}{{\tiny 14}}\hspace{0ex} &
              \textcolor{blue}{{\tiny 16}}\hspace{0ex} &
              \textcolor{blue}{{\tiny 17}}\hspace{0ex} & 
              
              \textcolor{blue}{{\tiny 18}}\hspace{0ex} &
              \textcolor{blue}{{\tiny 20}}\hspace{0ex} &
              \textcolor{blue}{{\tiny 21}}\hspace{0ex} & 
              
              \textcolor{blue}{{\tiny 22}}\hspace{0ex} &
              \textcolor{blue}{{\tiny 24}}\hspace{0ex} &
              \textcolor{blue}{{\tiny 25}}\hspace{0ex}  
            \\
\begin{block}{c[ccccccccccccccccc]}
               \textcolor{blue}{{\tiny 1}} & 
               1\hspace{0ex} & 0\hspace{0ex} & 0\hspace{0ex} & 0\hspace{0ex} & 
               
                1\hspace{0ex} &
               
               1\hspace{0ex} & 0\hspace{0ex} & 0\hspace{0ex} & 
               
               1\hspace{0ex} &
               0\hspace{0ex} & 
               0\hspace{0ex} &
               
               1\hspace{0ex} & 0\hspace{0ex} & 0\hspace{0ex} &
               
               1\hspace{0ex} & 0\hspace{0ex} & 0\hspace{0ex} 
               \\
               \textcolor{blue}{{\tiny 2}} & 
               0\hspace{0ex} & 1\hspace{0ex} & 0\hspace{0ex} & 0\hspace{0ex} & 
               
                0\hspace{0ex} & 
                
               0\hspace{0ex} &  0\hspace{0ex} & 0\hspace{0ex} &
               
               0\hspace{0ex} & 0\hspace{0ex} & 0\hspace{0ex} &
               
               0\hspace{0ex} & 0\hspace{0ex} & 0\hspace{0ex} &
               
               0\hspace{0ex} & 0\hspace{0ex} & 0\hspace{0ex} 
               \\
               \textcolor{blue}{{\tiny 3}} & 
               0\hspace{0ex} & 0\hspace{0ex} & 1\hspace{0ex} & 0\hspace{0ex} &
               
                1\hspace{0ex} & 
               
               0\hspace{0ex} & 1\hspace{0ex} & 0\hspace{0ex}  &
               
               0\hspace{0ex} & 1\hspace{0ex} & 0\hspace{0ex}  &
               
               0\hspace{0ex} & 1\hspace{0ex} & 0\hspace{0ex} &
               
               0\hspace{0ex} & 1\hspace{0ex} & 0\hspace{0ex} 
               \\
               \textcolor{blue}{{\tiny 4}} & 
               0\hspace{0ex} & 0\hspace{0ex} & 0\hspace{0ex} & 1\hspace{0ex} &
               
                1\hspace{0ex} & 
                
               0\hspace{0ex} & 0\hspace{0ex} & 1\hspace{0ex} &
               
               0\hspace{0ex} & 0\hspace{0ex} & 1\hspace{0ex} &
               
               0\hspace{0ex} & 0\hspace{0ex} & 1\hspace{0ex} &
               
               0\hspace{0ex} & 0\hspace{0ex} & 1\hspace{0ex} 
               \\
             \end{block}
\end{blockarray}
}
\\
\subfloat[$\boldsymbol{H}_{\ast}^{\{3\}\cup B_{3}}$]
{
\begin{blockarray}{cccccccccccccccccc}
              & 
              \textcolor{blue}{{\tiny 1}}\hspace{0ex} & 
              \textcolor{blue}{{\tiny 2}}\hspace{0ex} & 
              \textcolor{blue}{{\tiny 3}}\hspace{0ex} & 
              \textcolor{blue}{{\tiny 4}}\hspace{0ex} & 
              
              \textcolor{blue}{{\tiny 6}}\hspace{0ex} &
              
              \textcolor{blue}{{\tiny 10}}\hspace{0ex} &
              \textcolor{blue}{{\tiny 11}}\hspace{0ex} &
              \textcolor{blue}{{\tiny  13}}\hspace{0ex} & 
              
              \textcolor{blue}{{\tiny 14}}\hspace{0ex} &
              \textcolor{blue}{{\tiny 15}}\hspace{0ex} &
              \textcolor{blue}{{\tiny  17}}\hspace{0ex} & 
              
              \textcolor{blue}{{\tiny 18}}\hspace{0ex} &
              \textcolor{blue}{{\tiny 19}}\hspace{0ex} &
              \textcolor{blue}{{\tiny  21}}\hspace{0ex} & 
              
              \textcolor{blue}{{\tiny 22}}\hspace{0ex} &
              \textcolor{blue}{{\tiny 23}}\hspace{0ex} &
              \textcolor{blue}{{\tiny 25}}\hspace{0ex}  
            \\
\begin{block}{c[ccccccccccccccccc]}
               \textcolor{blue}{{\tiny 1}} & 
               1\hspace{0ex} & 0\hspace{0ex} & 0\hspace{0ex} & 
               0\hspace{0ex} & 
               
                1\hspace{0ex} & 
                
               1\hspace{0ex} & 0\hspace{0ex} & 0\hspace{0ex} &
               
               1\hspace{0ex} & 0\hspace{0ex} & 0\hspace{0ex} &
               
               1\hspace{0ex} & 0\hspace{0ex} & 0\hspace{0ex} &
               
               1\hspace{0ex} & 0\hspace{0ex} & 0\hspace{0ex} 
               \\
               \textcolor{blue}{{\tiny 2}} & 
               0\hspace{0ex} & 1\hspace{0ex} & 0\hspace{0ex} & 
               0\hspace{0ex} & 
               
               1\hspace{0ex} & 
               
               0\hspace{0ex} & 1\hspace{0ex} & 0\hspace{0ex} &
               
               0\hspace{0ex} & 1\hspace{0ex} & 0\hspace{0ex} &
               
               0\hspace{0ex} & 1\hspace{0ex} & 0\hspace{0ex} &
               
               0\hspace{0ex} & 1\hspace{0ex} & 0\hspace{0ex} 
               \\
               \textcolor{blue}{{\tiny 3}} & 
               0\hspace{0ex} & 0\hspace{0ex} & 1\hspace{0ex} & 
               0\hspace{0ex} & 
               
               0\hspace{0ex} & 
               
               0\hspace{0ex} & 0\hspace{0ex} & 0\hspace{0ex}  &
               
               0\hspace{0ex} & 0\hspace{0ex} & 0\hspace{0ex}  &
               
               0\hspace{0ex} & 0\hspace{0ex} & 0\hspace{0ex} &
               
               0\hspace{0ex} & 0\hspace{0ex} & 0\hspace{0ex} 
               \\
               \textcolor{blue}{{\tiny 4}} & 
               0\hspace{0ex} & 0\hspace{0ex} & 0\hspace{0ex} & 
               1\hspace{0ex} & 
               
               1\hspace{0ex} & 
               
               0\hspace{0ex} & 0\hspace{0ex} & 1\hspace{0ex} &
               
               0\hspace{0ex} & 0\hspace{0ex} & 1\hspace{0ex} &
               
               0\hspace{0ex} & 0\hspace{0ex} & 1\hspace{0ex} &
               
               0\hspace{0ex} & 0\hspace{0ex} & 1\hspace{0ex} 
               \\
             \end{block}
\end{blockarray}
}
\\
\subfloat[$\boldsymbol{H}_{\ast}^{\{4\}\cup B_{4}}$]
{
\begin{blockarray}{cccccccccccccccccc}
              & 
              \textcolor{blue}{{\tiny 1}}\hspace{0ex} & 
              \textcolor{blue}{{\tiny 2}}\hspace{0ex} & 
              \textcolor{blue}{{\tiny 3}}\hspace{0ex} & 
              \textcolor{blue}{{\tiny 4}}\hspace{0ex} & 
              \textcolor{blue}{{\tiny 5}}\hspace{0ex} & 
              
              \textcolor{blue}{{\tiny 10}}\hspace{0ex} &
              \textcolor{blue}{{\tiny 11}}\hspace{0ex} &
              \textcolor{blue}{{\tiny 12}}\hspace{0ex} &

              \textcolor{blue}{{\tiny 14}}\hspace{0ex} &
              \textcolor{blue}{{\tiny 15}}\hspace{0ex} &
              \textcolor{blue}{{\tiny 16}}\hspace{0ex} &

              \textcolor{blue}{{\tiny 18}}\hspace{0ex} &
              \textcolor{blue}{{\tiny 19}}\hspace{0ex} &
              \textcolor{blue}{{\tiny 20}}\hspace{0ex} &

              \textcolor{blue}{{\tiny 22}}\hspace{0ex} &
              \textcolor{blue}{{\tiny 23}}\hspace{0ex} &
              \textcolor{blue}{{\tiny 24}}\hspace{0ex} 
            \\
\begin{block}{c[ccccccccccccccccc]}
               \textcolor{blue}{{\tiny 1}} & 
               1\hspace{0ex} & 0\hspace{0ex} & 0\hspace{0ex} & 
               0\hspace{0ex} & 
               
               1\hspace{0ex} &

               1\hspace{0ex} & 0\hspace{0ex} & 0\hspace{0ex} &

               1\hspace{0ex} & 0\hspace{0ex} & 0\hspace{0ex} &

               1\hspace{0ex} & 0\hspace{0ex} & 0\hspace{0ex} &

               1\hspace{0ex} & 0\hspace{0ex} & 0\hspace{0ex}  
               
               \\
               \textcolor{blue}{{\tiny 2}} & 
               0\hspace{0ex} & 1\hspace{0ex} & 0\hspace{0ex} & 
               0\hspace{0ex} & 
               
               1\hspace{0ex} & 
               
               0\hspace{0ex} & 1\hspace{0ex} & 0\hspace{0ex} &

               0\hspace{0ex} & 1\hspace{0ex} & 0\hspace{0ex} &

               0\hspace{0ex} & 1\hspace{0ex} & 0\hspace{0ex} &

               0\hspace{0ex} & 1\hspace{0ex} & 0\hspace{0ex}  
               
               \\
               \textcolor{blue}{{\tiny 3}} & 
               0\hspace{0ex} & 0\hspace{0ex} & 1\hspace{0ex} & 
               0\hspace{0ex} & 
               
               1\hspace{0ex} & 
               
               0\hspace{0ex} & 0\hspace{0ex} & 1\hspace{0ex} &

               0\hspace{0ex} & 0\hspace{0ex} & 1\hspace{0ex} &

               0\hspace{0ex} & 0\hspace{0ex} & 1\hspace{0ex} &

               0\hspace{0ex} & 0\hspace{0ex} & 1\hspace{0ex}  
               
               \\
               \textcolor{blue}{{\tiny 4}} & 
               0\hspace{0ex} & 0\hspace{0ex} & 0\hspace{0ex} & 
               1\hspace{0ex} & 
               
               0\hspace{0ex} & 
               
               0\hspace{0ex} & 0\hspace{0ex} & 0\hspace{0ex} &

               0\hspace{0ex} & 0\hspace{0ex} & 0\hspace{0ex} &

               0\hspace{0ex} & 0\hspace{0ex} & 0\hspace{0ex} &

               0\hspace{0ex} & 0\hspace{0ex} & 0\hspace{0ex} 
               
               \\
             \end{block}
\end{blockarray}
}
\\
\subfloat[$\boldsymbol{H}_{\ast}^{\{5\}\cup B_{5}}$]
{
\begin{blockarray}{cccc}
              & 
              \textcolor{blue}{{\tiny 5}}\hspace{0ex} & 
              \textcolor{blue}{{\tiny 7}}\hspace{0ex} &
              \textcolor{blue}{{\tiny 8}}\hspace{0ex} 
            \\
\begin{block}{c[ccc]}
               \textcolor{blue}{{\tiny 1}} & 
               1\hspace{0ex} & 1\hspace{0ex} &
               0\hspace{0ex} 
               \\
               \textcolor{blue}{{\tiny 2}} & 
               
               1\hspace{0ex} & 0\hspace{0ex} & 1\hspace{0ex}  
               
               \\
               \textcolor{blue}{{\tiny 3}} & 
               
               1\hspace{0ex} & 1\hspace{0ex} & 1\hspace{0ex}  
               
               \\
               \textcolor{blue}{{\tiny 4}} & 
               
               0\hspace{0ex} & 1\hspace{0ex} & 1\hspace{0ex} 
               
               \\
             \end{block}
\end{blockarray}
}
\ \ \ \ \ \
\subfloat[$\boldsymbol{H}_{\ast}^{\{6\}\cup B_{6}}$]
{
\begin{blockarray}{cccc}
              & 
              \textcolor{blue}{{\tiny 5}}\hspace{0ex} & 
              \textcolor{blue}{{\tiny 6}}\hspace{0ex} &
              \textcolor{blue}{{\tiny 8}}\hspace{0ex} 
            \\
\begin{block}{c[ccc]}
               \textcolor{blue}{{\tiny 1}} & 
               1\hspace{0ex} & 1\hspace{0ex} &
               0\hspace{0ex} 
               \\
               \textcolor{blue}{{\tiny 2}} & 
               
               1\hspace{0ex} & 1\hspace{0ex} & 1\hspace{0ex}  
               
               \\
               \textcolor{blue}{{\tiny 3}} & 
               
               1\hspace{0ex} & 0\hspace{0ex} & 1\hspace{0ex}  
               
               \\
               \textcolor{blue}{{\tiny 4}} & 
               
               0\hspace{0ex} & 1\hspace{0ex} & 1\hspace{0ex} 
               
               \\
             \end{block}
\end{blockarray}
}
\ \ \ \ \ \
\subfloat[$\boldsymbol{H}_{\ast}^{\{7\}\cup B_{8}}$]
{
\begin{blockarray}{cccc}
              & 
              \textcolor{blue}{{\tiny 5}}\hspace{0ex} & 
              \textcolor{blue}{{\tiny 6}}\hspace{0ex} &
              \textcolor{blue}{{\tiny 7}}\hspace{0ex} 
            \\
\begin{block}{c[ccc]}
               \textcolor{blue}{{\tiny 1}} & 
               1\hspace{0ex} & 1\hspace{0ex} &
               1\hspace{0ex} 
               \\
               \textcolor{blue}{{\tiny 2}} & 
               
               1\hspace{0ex} & 1\hspace{0ex} & 0\hspace{0ex}  
               
               \\
               \textcolor{blue}{{\tiny 3}} & 
               
               1\hspace{0ex} & 0\hspace{0ex} & 1\hspace{0ex}  
               
               \\
               \textcolor{blue}{{\tiny 4}} & 
               
               0\hspace{0ex} & 1\hspace{0ex} & 1\hspace{0ex} 
               
               \\
             \end{block}
\end{blockarray}
}
\ \ \ \ \ \
\subfloat[$\boldsymbol{H}_{\ast}^{\{8\}\cup B_{8}}$]
{
\begin{blockarray}{cccc}
              & 
              \textcolor{blue}{{\tiny 6}}\hspace{0ex} & 
              \textcolor{blue}{{\tiny 7}}\hspace{0ex} &
              \textcolor{blue}{{\tiny 8}}\hspace{0ex} 
            \\
\begin{block}{c[ccc]}
               \textcolor{blue}{{\tiny 1}} & 
               1\hspace{0ex} & 1\hspace{0ex} &
               0\hspace{0ex} 
               \\
               \textcolor{blue}{{\tiny 2}} & 
               
               1\hspace{0ex} & 0\hspace{0ex} & 1\hspace{0ex}  
               
               \\
               \textcolor{blue}{{\tiny 3}} & 
               
               0\hspace{0ex} & 1\hspace{0ex} & 1\hspace{0ex}  
               
               \\
               \textcolor{blue}{{\tiny 4}} & 
               
               1\hspace{0ex} & 1\hspace{0ex} & 1\hspace{0ex} 
               
               \\
             \end{block}
\end{blockarray}
}
\ \ \ \ \ \
\subfloat[$\boldsymbol{H}_{\ast}^{\{9\}\cup B_{9}}$]
{
\begin{blockarray}{cccccc}
              & 
              \textcolor{blue}{{\tiny 5}}\hspace{0ex} & 
              \textcolor{blue}{{\tiny 6}}\hspace{0ex} &
              \textcolor{blue}{{\tiny 7}}\hspace{0ex} &
              \textcolor{blue}{{\tiny 8}}\hspace{0ex} &
              \textcolor{blue}{{\tiny 9}}\hspace{0ex} 
            \\
\begin{block}{c[ccccc]}
               \textcolor{blue}{{\tiny 1}} & 
               1\hspace{0ex} & 1\hspace{0ex} &
               1\hspace{0ex} & 0\hspace{0ex} &
               1\hspace{0ex} 
               \\
               \textcolor{blue}{{\tiny 2}} & 
               
               1\hspace{0ex} & 1\hspace{0ex} &
               0\hspace{0ex} & 1\hspace{0ex} &
               1\hspace{0ex}  
               
               \\
               \textcolor{blue}{{\tiny 3}} & 
               
               1\hspace{0ex} & 0\hspace{0ex} &
               1\hspace{0ex} & 1\hspace{0ex} &
               1\hspace{0ex}  
               
               \\
               \textcolor{blue}{{\tiny 4}} & 
               
               0\hspace{0ex} & 1\hspace{0ex} &
               1\hspace{0ex} & 1\hspace{0ex} &
               1\hspace{0ex} 
               
               \\
             \end{block}
\end{blockarray}
}
\\
\subfloat[$\boldsymbol{H}_{\ast}^{\{10\}\cup B_{10}}$]
{
\begin{blockarray}{cccc}
              & 
              \textcolor{blue}{{\tiny 5}}\hspace{0ex} & 
              \textcolor{blue}{{\tiny 10}}\hspace{0ex} &
              \textcolor{blue}{{\tiny 11}}\hspace{0ex} 
            \\
\begin{block}{c[ccc]}
               \textcolor{blue}{{\tiny 1}} & 
               1\hspace{0ex} & 1\hspace{0ex} &
               0\hspace{0ex} 
               \\
               \textcolor{blue}{{\tiny 2}} & 
               
               1\hspace{0ex} & 0\hspace{0ex} & 1\hspace{0ex}  
               
               \\
               \textcolor{blue}{{\tiny 3}} & 
               
               1\hspace{0ex} & 0\hspace{0ex} & 0\hspace{0ex}  
               
               \\
               \textcolor{blue}{{\tiny 4}} & 
               
               0\hspace{0ex} & 0\hspace{0ex} & 0\hspace{0ex} 
               
               \\
             \end{block}
\end{blockarray}
}
\ \ \ \ \ \
\subfloat[$\boldsymbol{H}_{\ast}^{\{11\}\cup B_{11}}$]
{
\begin{blockarray}{cccc}
              & 
              \textcolor{blue}{{\tiny 5}}\hspace{0ex} & 
              \textcolor{blue}{{\tiny 11}}\hspace{0ex} &
              \textcolor{blue}{{\tiny 12}}\hspace{0ex} 
            \\
\begin{block}{c[ccc]}
               \textcolor{blue}{{\tiny 1}} & 
               1\hspace{0ex} & 0\hspace{0ex} &
               0\hspace{0ex} 
               \\
               \textcolor{blue}{{\tiny 2}} & 
               
               1\hspace{0ex} & 1\hspace{0ex} & 0\hspace{0ex}  
               
               \\
               \textcolor{blue}{{\tiny 3}} & 
               
               1\hspace{0ex} & 0\hspace{0ex} & 1\hspace{0ex}  
               
               \\
               \textcolor{blue}{{\tiny 4}} & 
               
               0\hspace{0ex} & 0\hspace{0ex} & 0\hspace{0ex} 
               
               \\
             \end{block}
\end{blockarray}
}
\ \ \ \ \ \
\subfloat[$\boldsymbol{H}_{\ast}^{\{12\}\cup B_{12}}$]
{
\begin{blockarray}{cccc}
              & 
              \textcolor{blue}{{\tiny 5}}\hspace{0ex} & 
              \textcolor{blue}{{\tiny 10}}\hspace{0ex} &
              \textcolor{blue}{{\tiny 12}}\hspace{0ex} 
            \\
\begin{block}{c[ccc]}
               \textcolor{blue}{{\tiny 1}} & 
               1\hspace{0ex} & 1\hspace{0ex} &
               0\hspace{0ex} 
               \\
               \textcolor{blue}{{\tiny 2}} & 
               
               1\hspace{0ex} & 0\hspace{0ex} & 0\hspace{0ex}  
               
               \\
               \textcolor{blue}{{\tiny 3}} & 
               
               1\hspace{0ex} & 0\hspace{0ex} & 1\hspace{0ex}  
               
               \\
               \textcolor{blue}{{\tiny 4}} & 
               
               0\hspace{0ex} & 0\hspace{0ex} & 0\hspace{0ex} 
               
               \\
             \end{block}
\end{blockarray}
}
\ \ \ \ \ \
\subfloat[$\boldsymbol{H}_{\ast}^{\{13\}\cup B_{13}}$]
{
\begin{blockarray}{ccccc}
              & 
              \textcolor{blue}{{\tiny 1}}\hspace{0ex} & 
              \textcolor{blue}{{\tiny 8}}\hspace{0ex} &
              \textcolor{blue}{{\tiny 9}}\hspace{0ex} &
              \textcolor{blue}{{\tiny 13}}\hspace{0ex} 
            \\
\begin{block}{c[cccc]}
               \textcolor{blue}{{\tiny 1}} & 
               1\hspace{0ex} & 0\hspace{0ex} &
               1\hspace{0ex} &
               0\hspace{0ex} 
               \\
               \textcolor{blue}{{\tiny 2}} & 
               
               0\hspace{0ex} & 1\hspace{0ex} &
               1\hspace{0ex} &
               0\hspace{0ex}  
               
               \\
               \textcolor{blue}{{\tiny 3}} & 
               
               0\hspace{0ex} & 1\hspace{0ex} &
               1\hspace{0ex} &
               0\hspace{0ex}  
               
               \\
               \textcolor{blue}{{\tiny 4}} & 
               
               0\hspace{0ex} & 1\hspace{0ex} &
               1\hspace{0ex} &
               1\hspace{0ex} 
               
               \\
             \end{block}
\end{blockarray}
}
\\
\subfloat[$\boldsymbol{H}_{\ast}^{\{14\}\cup B_{14}}$]
{
\begin{blockarray}{cccc}
              & 
              \textcolor{blue}{{\tiny 6}}\hspace{0ex} & 
              \textcolor{blue}{{\tiny 14}}\hspace{0ex} &
              \textcolor{blue}{{\tiny 15}}\hspace{0ex} 
            \\
\begin{block}{c[ccc]}
               \textcolor{blue}{{\tiny 1}} & 
               1\hspace{0ex} & 1\hspace{0ex} &
               0\hspace{0ex} 
               \\
               \textcolor{blue}{{\tiny 2}} & 
               
               1\hspace{0ex} & 0\hspace{0ex} & 1\hspace{0ex}  
               
               \\
               \textcolor{blue}{{\tiny 3}} & 
               
               0\hspace{0ex} & 0\hspace{0ex} & 0\hspace{0ex}  
               
               \\
               \textcolor{blue}{{\tiny 4}} & 
               
               1\hspace{0ex} & 0\hspace{0ex} & 0\hspace{0ex} 
               
               \\
             \end{block}
\end{blockarray}
}
\ \ \ \ \ \
\subfloat[$\boldsymbol{H}_{\ast}^{\{15\}\cup B_{15}}$]
{
\begin{blockarray}{cccc}
              & 
              \textcolor{blue}{{\tiny 6}}\hspace{0ex} & 
              \textcolor{blue}{{\tiny 15}}\hspace{0ex} &
              \textcolor{blue}{{\tiny 17}}\hspace{0ex} 
            \\
\begin{block}{c[ccc]}
               \textcolor{blue}{{\tiny 1}} & 
               1\hspace{0ex} & 0\hspace{0ex} &
               0\hspace{0ex} 
               \\
               \textcolor{blue}{{\tiny 2}} & 
               
               1\hspace{0ex} & 1\hspace{0ex} & 0\hspace{0ex}  
               
               \\
               \textcolor{blue}{{\tiny 3}} & 
               
               0\hspace{0ex} & 0\hspace{0ex} & 0\hspace{0ex}  
               
               \\
               \textcolor{blue}{{\tiny 4}} & 
               
               1\hspace{0ex} & 0\hspace{0ex} & 1\hspace{0ex} 
               
               \\
             \end{block}
\end{blockarray}
}
\ \ \ \ \ \
\subfloat[$\boldsymbol{H}_{\ast}^{\{16\}\cup B_{16}}$]
{
\begin{blockarray}{ccccc}
              & 
              \textcolor{blue}{{\tiny 4}}\hspace{0ex} & 
              \textcolor{blue}{{\tiny 5}}\hspace{0ex} &
              \textcolor{blue}{{\tiny 9}}\hspace{0ex} &
              \textcolor{blue}{{\tiny 16}}\hspace{0ex} 
            \\
\begin{block}{c[cccc]}
               \textcolor{blue}{{\tiny 1}} & 
               0\hspace{0ex} & 1\hspace{0ex} &
               1\hspace{0ex} &
               0\hspace{0ex} 
               \\
               \textcolor{blue}{{\tiny 2}} & 
               
               0\hspace{0ex} & 1\hspace{0ex} &
               1\hspace{0ex} &
               0\hspace{0ex}  
               
               \\
               \textcolor{blue}{{\tiny 3}} & 
               
               0\hspace{0ex} & 1\hspace{0ex} &
               1\hspace{0ex} &
               1\hspace{0ex}  
               
               \\
               \textcolor{blue}{{\tiny 4}} & 
               
               1\hspace{0ex} & 0\hspace{0ex} &
               1\hspace{0ex} &
               0\hspace{0ex} 
               
               \\
             \end{block}
\end{blockarray}
}
\ \ \ \ \ \
\subfloat[$\boldsymbol{H}_{\ast}^{\{17\}\cup B_{17}}$]
{
\begin{blockarray}{cccc}
              & 
              \textcolor{blue}{{\tiny 6}}\hspace{0ex} & 
              \textcolor{blue}{{\tiny 14}}\hspace{0ex} &
              \textcolor{blue}{{\tiny 17}}\hspace{0ex} 
            \\
\begin{block}{c[ccc]}
               \textcolor{blue}{{\tiny 1}} & 
               1\hspace{0ex} & 1\hspace{0ex} &
               0\hspace{0ex} 
               \\
               \textcolor{blue}{{\tiny 2}} & 
               
               1\hspace{0ex} & 0\hspace{0ex} & 0\hspace{0ex}  
               
               \\
               \textcolor{blue}{{\tiny 3}} & 
               
               0\hspace{0ex} & 0\hspace{0ex} & 0\hspace{0ex}  
               
               \\
               \textcolor{blue}{{\tiny 4}} & 
               
               1\hspace{0ex} & 0\hspace{0ex} & 1\hspace{0ex} 
               
               \\
             \end{block}
\end{blockarray}
}
\caption{
$\boldsymbol{H}_{\ast}^{\{i\}\cup B_{i}}, i\in [17]$.}
\label{fig:333}
\end{figure*}

\begin{figure*}
\centering
\subfloat[$\boldsymbol{H}_{\ast}^{\{18\}\cup B_{18}}$]
{
\begin{blockarray}{cccc}
              & 
              \textcolor{blue}{{\tiny 7}}\hspace{0ex} & 
              \textcolor{blue}{{\tiny 18}}\hspace{0ex} &
              \textcolor{blue}{{\tiny 20}}\hspace{0ex} 
            \\
\begin{block}{c[ccc]}
               \textcolor{blue}{{\tiny 1}} & 
               1\hspace{0ex} & 1\hspace{0ex} &
               0\hspace{0ex} 
               \\
               \textcolor{blue}{{\tiny 2}} & 
               
               0\hspace{0ex} & 0\hspace{0ex} & 0\hspace{0ex}  
               
               \\
               \textcolor{blue}{{\tiny 3}} & 
               
               1\hspace{0ex} & 0\hspace{0ex} & 1\hspace{0ex}  
               
               \\
               \textcolor{blue}{{\tiny 4}} & 
               
               1\hspace{0ex} & 0\hspace{0ex} & 0\hspace{0ex} 
               
               \\
             \end{block}
\end{blockarray}
}
\ \ \ \ \ \
\subfloat[$\boldsymbol{H}_{\ast}^{\{19\}\cup B_{19}}$]
{
\begin{blockarray}{ccccc}
              & 
              \textcolor{blue}{{\tiny 3}}\hspace{0ex} & 
              \textcolor{blue}{{\tiny 6}}\hspace{0ex} &
              \textcolor{blue}{{\tiny 9}}\hspace{0ex} &
              \textcolor{blue}{{\tiny 19}}\hspace{0ex} 
            \\
\begin{block}{c[cccc]}
               \textcolor{blue}{{\tiny 1}} & 
               0\hspace{0ex} & 1\hspace{0ex} &
               1\hspace{0ex} &
               0\hspace{0ex} 
               \\
               \textcolor{blue}{{\tiny 2}} & 
               
               0\hspace{0ex} & 1\hspace{0ex} &
               1\hspace{0ex} &
               1\hspace{0ex}  
               
               \\
               \textcolor{blue}{{\tiny 3}} & 
               
               1\hspace{0ex} & 0\hspace{0ex} &
               1\hspace{0ex} &
               0\hspace{0ex}  
               
               \\
               \textcolor{blue}{{\tiny 4}} & 
               
               0\hspace{0ex} & 1\hspace{0ex} &
               1\hspace{0ex} &
               0\hspace{0ex} 
               
               \\
             \end{block}
\end{blockarray}
}
\ \ \ \ \ \
\subfloat[$\boldsymbol{H}_{\ast}^{\{20\}\cup B_{20}}$]
{
\begin{blockarray}{cccc}
              & 
              \textcolor{blue}{{\tiny 7}}\hspace{0ex} & 
              \textcolor{blue}{{\tiny 20}}\hspace{0ex} &
              \textcolor{blue}{{\tiny 21}}\hspace{0ex} 
            \\
\begin{block}{c[ccc]}
               \textcolor{blue}{{\tiny 1}} & 
               1\hspace{0ex} & 0\hspace{0ex} &
               0\hspace{0ex} 
               \\
               \textcolor{blue}{{\tiny 2}} & 
               
               0\hspace{0ex} & 0\hspace{0ex} & 0\hspace{0ex}  
               
               \\
               \textcolor{blue}{{\tiny 3}} & 
               
               1\hspace{0ex} & 1\hspace{0ex} & 0\hspace{0ex}  
               
               \\
               \textcolor{blue}{{\tiny 4}} & 
               
               1\hspace{0ex} & 0\hspace{0ex} & 1\hspace{0ex} 
               
               \\
             \end{block}
\end{blockarray}
}
\ \ \ \ \ \
\subfloat[$\boldsymbol{H}_{\ast}^{\{21\}\cup B_{21}}$]
{
\begin{blockarray}{cccc}
              & 
              \textcolor{blue}{{\tiny 7}}\hspace{0ex} & 
              \textcolor{blue}{{\tiny 18}}\hspace{0ex} &
              \textcolor{blue}{{\tiny 21}}\hspace{0ex} 
            \\
\begin{block}{c[ccc]}
               \textcolor{blue}{{\tiny 1}} & 
               1\hspace{0ex} & 1\hspace{0ex} &
               0\hspace{0ex} 
               \\
               \textcolor{blue}{{\tiny 2}} & 
               
               0\hspace{0ex} & 0\hspace{0ex} & 0\hspace{0ex}  
               
               \\
               \textcolor{blue}{{\tiny 3}} & 
               
               1\hspace{0ex} & 0\hspace{0ex} & 0\hspace{0ex}  
               
               \\
               \textcolor{blue}{{\tiny 4}} & 
               
               1\hspace{0ex} & 0\hspace{0ex} & 1\hspace{0ex} 
               
               \\
             \end{block}
\end{blockarray}
}
\\
\subfloat[$\boldsymbol{H}_{\ast}^{\{22\}\cup B_{22}}$]
{
\begin{blockarray}{ccccc}
              & 
              \textcolor{blue}{{\tiny 2}}\hspace{0ex} & 
              \textcolor{blue}{{\tiny 7}}\hspace{0ex} &
              \textcolor{blue}{{\tiny 9}}\hspace{0ex} &
              \textcolor{blue}{{\tiny 22}}\hspace{0ex} 
            \\
\begin{block}{c[cccc]}
               \textcolor{blue}{{\tiny 1}} & 
               0\hspace{0ex} & 1\hspace{0ex} &
               1\hspace{0ex} &
               1\hspace{0ex} 
               \\
               \textcolor{blue}{{\tiny 2}} & 
               
               1\hspace{0ex} & 0\hspace{0ex} &
               1\hspace{0ex} &
               0\hspace{0ex}  
               
               \\
               \textcolor{blue}{{\tiny 3}} & 
               
               0\hspace{0ex} & 1\hspace{0ex} &
               1\hspace{0ex} &
               0\hspace{0ex}  
               
               \\
               \textcolor{blue}{{\tiny 4}} & 
               
               0\hspace{0ex} & 1\hspace{0ex} &
               1\hspace{0ex} &
               0\hspace{0ex} 
               
               \\
             \end{block}
\end{blockarray}
}
\ \ \ \ \ \
\subfloat[$\boldsymbol{H}_{\ast}^{\{23\}\cup B_{23}}$]
{
\begin{blockarray}{cccc}
              & 
              \textcolor{blue}{{\tiny 8}}\hspace{0ex} & 
              \textcolor{blue}{{\tiny 23}}\hspace{0ex} &
              \textcolor{blue}{{\tiny 24}}\hspace{0ex} 
            \\
\begin{block}{c[ccc]}
               \textcolor{blue}{{\tiny 1}} & 
               0\hspace{0ex} & 0\hspace{0ex} &
               0\hspace{0ex} 
               \\
               \textcolor{blue}{{\tiny 2}} & 
               
               1\hspace{0ex} & 1\hspace{0ex} & 0\hspace{0ex}  
               
               \\
               \textcolor{blue}{{\tiny 3}} & 
               
               1\hspace{0ex} & 0\hspace{0ex} & 1\hspace{0ex}  
               
               \\
               \textcolor{blue}{{\tiny 4}} & 
               
               1\hspace{0ex} & 0\hspace{0ex} & 0\hspace{0ex} 
               
               \\
             \end{block}
\end{blockarray}
}
\ \ \ \ \ \
\subfloat[$\boldsymbol{H}_{\ast}^{\{24\}\cup B_{24}}$]
{
\begin{blockarray}{cccc}
              & 
              \textcolor{blue}{{\tiny 8}}\hspace{0ex} & 
              \textcolor{blue}{{\tiny 24}}\hspace{0ex} &
              \textcolor{blue}{{\tiny 25}}\hspace{0ex} 
            \\
\begin{block}{c[ccc]}
               \textcolor{blue}{{\tiny 1}} & 
               0\hspace{0ex} & 0\hspace{0ex} &
               0\hspace{0ex} 
               \\
               \textcolor{blue}{{\tiny 2}} & 
               
               1\hspace{0ex} & 0\hspace{0ex} & 0\hspace{0ex}  
               
               \\
               \textcolor{blue}{{\tiny 3}} & 
               
               1\hspace{0ex} & 1\hspace{0ex} & 0\hspace{0ex}  
               
               \\
               \textcolor{blue}{{\tiny 4}} & 
               
               1\hspace{0ex} & 0\hspace{0ex} & 1\hspace{0ex} 
               
               \\
             \end{block}
\end{blockarray}
}
\ \ \ \ \ \
\subfloat[$\boldsymbol{H}_{\ast}^{\{25\}\cup B_{25}}$]
{
\begin{blockarray}{cccc}
              & 
              \textcolor{blue}{{\tiny 8}}\hspace{0ex} & 
              \textcolor{blue}{{\tiny 23}}\hspace{0ex} &
              \textcolor{blue}{{\tiny 25}}\hspace{0ex} 
            \\
\begin{block}{c[ccc]}
               \textcolor{blue}{{\tiny 1}} & 
               0\hspace{0ex} & 0\hspace{0ex} &
               0\hspace{0ex} 
               \\
               \textcolor{blue}{{\tiny 2}} & 
               
               1\hspace{0ex} & 1\hspace{0ex} & 0\hspace{0ex}  
               
               \\
               \textcolor{blue}{{\tiny 3}} & 
               
               1\hspace{0ex} & 0\hspace{0ex} & 0\hspace{0ex}  
               
               \\
               \textcolor{blue}{{\tiny 4}} & 
               
               1\hspace{0ex} & 0\hspace{0ex} & 1\hspace{0ex} 
               
               \\
             \end{block}
\end{blockarray}
}
\\
\subfloat[$\boldsymbol{H}_{\ast}^{\{26\}\cup B_{26}}$]
{
\begin{blockarray}{cccccc}
              & 
              \textcolor{blue}{{\tiny 4}}\hspace{0ex} & 
              \textcolor{blue}{{\tiny 5}}\hspace{0ex} &
              \textcolor{blue}{{\tiny 9}}\hspace{0ex} &
              \textcolor{blue}{{\tiny 16}}\hspace{0ex} &
              \textcolor{blue}{{\tiny 26}}\hspace{0ex} 
            \\
\begin{block}{c[ccccc]}
               \textcolor{blue}{{\tiny 1}} & 
               0\hspace{0ex} & 1\hspace{0ex} &
               1\hspace{0ex} & 0\hspace{0ex} &
               0\hspace{0ex} 
               \\
               \textcolor{blue}{{\tiny 2}} & 
               
               0\hspace{0ex} & 1\hspace{0ex} &
               1\hspace{0ex} & 0\hspace{0ex} &
               1\hspace{0ex}  
               
               \\
               \textcolor{blue}{{\tiny 3}} & 
               
               0\hspace{0ex} & 1\hspace{0ex} &
               1\hspace{0ex} & 1\hspace{0ex} &
               0\hspace{0ex}  
               
               \\
               \textcolor{blue}{{\tiny 4}} & 
               
               1\hspace{0ex} & 0\hspace{0ex} &
               1\hspace{0ex} & 0\hspace{0ex} &
               0\hspace{0ex} 
               
               \\
             \end{block}
\end{blockarray}
}
\ \ \ \ \ \
\subfloat[$\boldsymbol{H}_{\ast}^{\{27\}\cup B_{27}}$]
{
\begin{blockarray}{cccccc}
              & 
              \textcolor{blue}{{\tiny 3}}\hspace{0ex} & 
              \textcolor{blue}{{\tiny 6}}\hspace{0ex} &
              \textcolor{blue}{{\tiny 9}}\hspace{0ex} &
              \textcolor{blue}{{\tiny 19}}\hspace{0ex} &
              \textcolor{blue}{{\tiny 27}}\hspace{0ex} 
            \\
\begin{block}{c[ccccc]}
               \textcolor{blue}{{\tiny 1}} & 
               0\hspace{0ex} & 1\hspace{0ex} &
               1\hspace{0ex} & 0\hspace{0ex} &
               1\hspace{0ex} 
               \\
               \textcolor{blue}{{\tiny 2}} & 
               
               0\hspace{0ex} & 1\hspace{0ex} &
               1\hspace{0ex} & 1\hspace{0ex} &
               0\hspace{0ex}  
               
               \\
               \textcolor{blue}{{\tiny 3}} & 
               
               1\hspace{0ex} & 0\hspace{0ex} &
               1\hspace{0ex} & 0\hspace{0ex} &
               0\hspace{0ex}  
               
               \\
               \textcolor{blue}{{\tiny 4}} & 
               
               0\hspace{0ex} & 1\hspace{0ex} &
               1\hspace{0ex} & 0\hspace{0ex} &
               0\hspace{0ex} 
               
               \\
             \end{block}
\end{blockarray}
}
\ \ \ \ \ \
\subfloat[$\boldsymbol{H}_{\ast}^{\{28\}\cup B_{28}}$]
{
\begin{blockarray}{cccccc}
              & 
              \textcolor{blue}{{\tiny 2}}\hspace{0ex} & 
              \textcolor{blue}{{\tiny 7}}\hspace{0ex} &
              \textcolor{blue}{{\tiny 9}}\hspace{0ex} &
              \textcolor{blue}{{\tiny 22}}\hspace{0ex} &
              \textcolor{blue}{{\tiny 28}}\hspace{0ex} 
            \\
\begin{block}{c[ccccc]}
               \textcolor{blue}{{\tiny 1}} & 
               0\hspace{0ex} & 1\hspace{0ex} &
               1\hspace{0ex} & 1\hspace{0ex} &
               0\hspace{0ex} 
               \\
               \textcolor{blue}{{\tiny 2}} & 
               
               1\hspace{0ex} & 0\hspace{0ex} &
               1\hspace{0ex} & 0\hspace{0ex} &
               0\hspace{0ex}  
               
               \\
               \textcolor{blue}{{\tiny 3}} & 
               
               0\hspace{0ex} & 1\hspace{0ex} &
               1\hspace{0ex} & 0\hspace{0ex} &
               0\hspace{0ex}  
               
               \\
               \textcolor{blue}{{\tiny 4}} & 
               
               0\hspace{0ex} & 1\hspace{0ex} &
               1\hspace{0ex} & 0\hspace{0ex} &
               1\hspace{0ex} 
               
               \\
             \end{block}
\end{blockarray}
}
\ \ \ \ \ \
\subfloat[$\boldsymbol{H}_{\ast}^{\{29\}\cup B_{29}}$]
{
\begin{blockarray}{cccccc}
              & 
              \textcolor{blue}{{\tiny 1}}\hspace{0ex} & 
              \textcolor{blue}{{\tiny 8}}\hspace{0ex} &
              \textcolor{blue}{{\tiny 9}}\hspace{0ex} &
              \textcolor{blue}{{\tiny 13}}\hspace{0ex} &
              \textcolor{blue}{{\tiny 29}}\hspace{0ex} 
            \\
\begin{block}{c[ccccc]}
               \textcolor{blue}{{\tiny 1}} & 
               1\hspace{0ex} & 0\hspace{0ex} &
               1\hspace{0ex} & 0\hspace{0ex} &
               0\hspace{0ex} 
               \\
               \textcolor{blue}{{\tiny 2}} & 
               
               0\hspace{0ex} & 1\hspace{0ex} &
               1\hspace{0ex} & 0\hspace{0ex} &
               0\hspace{0ex}  
               
               \\
               \textcolor{blue}{{\tiny 3}} & 
               
               0\hspace{0ex} & 1\hspace{0ex} &
               1\hspace{0ex} & 0\hspace{0ex} &
               1\hspace{0ex}  
               
               \\
               \textcolor{blue}{{\tiny 4}} & 
               
               0\hspace{0ex} & 1\hspace{0ex} &
               1\hspace{0ex} & 1\hspace{0ex} &
               0\hspace{0ex} 
               
               \\
             \end{block}
\end{blockarray}
}
\caption{
$\boldsymbol{H}_{\ast}^{\{i\}\cup B_{i}}, i\in [29]\backslash [17]$.}
\label{fig:334}
\end{figure*}

\subsection{Proof of Propositions \ref{prop-thm2-1}}
We show that matrix $\boldsymbol{H}_{\ast}\in \mathbb{F}_{q}^{4\times 29}$, shown in Figure \ref{fig:1}, will satisfy all users $u_{i}, \in [29]$ of the index coding instance $\mathcal{I}_{2}$ if the field $\mathbb{F}_{q}$ does have any characteristic other than characteristic three. Let $\boldsymbol{y}=\boldsymbol{H}_{\ast}\boldsymbol{x}$. 
Since all the users except $u_{i}, i\in [29]\backslash \{5,6,7,8,9\}$ have the same interfering message set as the users in the index coding instance $\mathcal{I}_{1}$, we can use the same argument in the previous subsection to show that these users will be satisfied by $\boldsymbol{H}_{\ast}$, shown in Figure \ref{fig:1}. We note that the characteristic of the field does not affect the results for users $u_{i}, i\in [29]\backslash \{5,6,7,8,9\}$. Thus, we just need to prove that users $u_{i}, i\in \{5,6,7,8,9\}$ will be satisfied.
\\
Due to $B_{9}=\emptyset$, user $u_{9}$ can easily decode its desired message $x_{9}$ from any of the coded messages $y_{i}, i\in [4]$. 
\\
For the users $u_{i}, i\in \{5,6,7,8\}$, we have
\begin{align}
\boldsymbol{H}_{\ast}^{\{i\}\cup B_{i}}=\boldsymbol{H}_{\ast}^{\{5,6,7,8\}}=
    \begin{bmatrix}
    1 & 1 & 1 & 0 \\
    1 & 1 & 0 & 1 \\
    1 & 0 & 1 & 1 \\
    0 & 1 & 1 & 1 
    \end{bmatrix}, \ i\in \{5,6,7,8\}
    \nonumber
\end{align}
\begin{align}
    u_{5}: \ y_{1}+y_{2}+y_{3}-y_{4}-y_{4}=3x_{5},
    \nonumber
    \\
    u_{6}: \ y_{1}+y_{2}+y_{4}-y_{3}-y_{3}=3x_{6},
    \nonumber
    \\
    u_{7}: \ y_{1}+y_{3}+y_{4}-y_{2}-y_{2}=3x_{7},
    \nonumber
    \\
    u_{8}: \ y_{2}+y_{3}+y_{4}-y_{1}-y_{1}=3x_{8}.
\end{align}
Since number 3 is invertible in the fields with any characteristic other than characteristic three, all users $u_{i}, i\in \{5,6,7,8\}$ can decode their requested message. This completes the proof.

\section{Proof of Lemmas \ref{lem:quasi-independent-cycle}-\ref{lem:two-interference-dimention}}
\label{app:proof- Lemma 7-9}

\subsection{Proof of Lemma \ref{lem:quasi-independent-cycle}}
Corollaries \ref{cor:j-subspace-M-circuit-quasi}-\ref{rem:acyclic-j-acyclic-quasi} can be derived from earlier results and will be used in the proof of Lemma \ref{lem:quasi-independent-cycle}.

\begin{cor} \label{cor:j-subspace-M-circuit-quasi}
Let $M$ be an independent set of $\boldsymbol{H}$. Now, if $\mathrm{col}(\boldsymbol{H}^{\{2j-1, 2j\}})\subseteq \mathrm{col}(\boldsymbol{H}^{M})$, then there exists one subset $M^{\prime}\subseteq M$, such that $\{2j-1, 2j\}\cup M^{\prime}$ is a quasi-circuit set of $\boldsymbol{H}$.
\end{cor}
\begin{proof}
Since $\mathrm{col}(\boldsymbol{H}^{\{2j-1, 2j\}})\subseteq \mathrm{col}(\boldsymbol{H}^{M})$, we must have $\boldsymbol{H}^{\{2j-1, 2j\}}=\sum_{i\in L} \boldsymbol{H}^{\{i\}}\boldsymbol{N}_{j,i}$ such that only matrices $\boldsymbol{N}_{j,i}, i\in L^{\prime}\subseteq L$ are invertible. Thus, according to Lemma \ref{lem:quasi-circuit-invertible}, set $\{2j-1, 2j\}\cup M^{\prime}$ forms a circuit set of $\boldsymbol{H}$ where $M^{\prime}=\{2i-1, 2i, i\in L^{\prime}\}\subseteq M$.
\end{proof}

\begin{cor} \label{cor:minimal-cyclic-acylic-quasi}
If $M$ is a quasi-minimal cyclic set of $\mathcal{I}$, then according to Definitions \ref{def:quasi-minimal-cyclic-set} and \ref{def:acyclic-set}, its proper subsets $M^{\prime}\subset M$ will be an acyclic set of $\mathcal{I}$. 
\end{cor}

\begin{cor}[\cite{Arbabjolfaei2018}]\label{cor:acyclic-l-quasi}
If $M$ is an acyclic set of $\mathcal{I}$, then there exists at least one $2i-1\in M$ such that $M\backslash\{2i-1, 2i\}\subseteq B_{2i-1}$. 
\end{cor}

\begin{cor} \label{rem:acyclic-j-acyclic-quasi}
Suppose $\{2j-1, 2j\}\cup M$ is a quasi-circuit set of $\boldsymbol{H}$. Then for any $\{2l-1, 2l\}\subseteq M$, we have $\mathrm{col}(\boldsymbol{H}^{\{2j-1, 2j\}\cup M\backslash \{2i-1, 2i\}})=\mathrm{col}(\boldsymbol{H}^{M})$.
\end{cor}
\begin{proof}
Since $\{2j-1, 2j\}\cup M$ is a quasi-circuit set of $\boldsymbol{H}$, we have $\boldsymbol{H}^{\{2j-1, 2j\}}=\sum_{i\in L} \boldsymbol{H}^{\{2i-1, 2i\}}\boldsymbol{N}_{j,i}$ such that each $\boldsymbol{N}_{j,i}$ is invertible. Now, assume $\{2l-1, 2l\}\subseteq M$. Then, we have 
\begin{align}
    &\mathrm{col}(\boldsymbol{H}^{\{2j-1,2j\}\cup M\backslash \{2l-1,2l\}})=
    \nonumber
    \\
    &\mathrm{col}([\boldsymbol{H}^{\{2j-1,2j\}}| \boldsymbol{H}^{M\backslash\{2l-1,2l\}}])=
    \nonumber
    \\
    &\mathrm{col}([\sum_{i\in L} \boldsymbol{H}^{\{2i-1, 2i\}}\boldsymbol{N}_{j,i}| \boldsymbol{H}^{M\backslash\{2l-1,2l\}}])=
    \nonumber
    \\
    &\mathrm{col}([\boldsymbol{H}^{\{2l-1, 2l\}}\boldsymbol{N}_{j,l}| \boldsymbol{H}^{M\backslash\{2l-1, 2l\}}])=
    \nonumber
    \\
    &\mathrm{col}(\boldsymbol{H}^{M}),
    \label{eq:corollary-4-quasi}
\end{align}
where \eqref{eq:corollary-4-quasi} is due to the invertibility of $\boldsymbol{N}_{j,l}$.
\end{proof}

\subsubsection{Proof of Lemma \ref{lem:quasi-independent-cycle}}
Since $M$ is an independent set of $\boldsymbol{H}$, and $\mathrm{col}(\boldsymbol{H}^{\{2j-1, 2j\}})\subseteq \mathrm{col}(\boldsymbol{H}^{M})$, then according to Corollary \ref{cor:j-subspace-M-circuit-quasi}, there exists a subset $M^{\prime}\subseteq M$ such that $\{2j-1, 2j\}\cup M^{\prime}$ is a quasi-circuit set. Now, we show that $M^{\prime}=M$, otherwise it leads to a contradiction. Assume $M^{\prime}\subset M$.
First, since $M$ is a quasi-minimal cyclic set of $\mathcal{I}$, based on Corollary \ref{cor:minimal-cyclic-acylic-quasi}, $M^{\prime}\subset M$ is an acyclic set of $\mathcal{I}$. Second, according to Corollary \ref{cor:acyclic-l}, there exists $\{2l-1, 2l\}\subseteq M^{\prime}$ such that $M^{\prime}\backslash\{2l-1, 2l\}\subseteq B_{2l-1}, B_{2l}$. Moreover, due to $\{2j-1, 2j\}\subseteq B_{2l-1}, B_{2l}$, we get $\{2j-1, 2j\}\cup (M^{\prime}\backslash\{2l-1, 2l\})\subseteq B_{2l-1}, B_{2l}$. Thus,
\begin{align}
  &\{2l-1, 2l\}\subseteq M^{\prime} \rightarrow \mathrm{col}(\boldsymbol{H}^{\{2l-1, 2l\}}) \subseteq \mathrm{col}(\boldsymbol{H}^{M^{\prime}}),
  \nonumber
  \\
  &\{2j-1, 2j\}\cup (M^{\prime}\backslash\{2l-1, 2l\})\subseteq B_{2l-1}\rightarrow
  \label{eq:rem:acyclic-l-quasi-0}
  \\
  &\mathrm{col}(\boldsymbol{H}^{\{2j-1, 2j\}\cup (M^{\prime}\backslash\{2l-1, 2l\})}) \subseteq \mathrm{col}(\boldsymbol{H}^{B_{2l-1}}).
    \label{eq:rem:acyclic-l-quasi}
\end{align}
Third, since $\{2j-1, 2j\}\cup M^{\prime}$ is a quasi-circuit set, Corollary \ref{rem:acyclic-j-acyclic-quasi} leads to
\begin{align}
    \mathrm{col}(\boldsymbol{H}^{\{2j-1, 2j\}\cup M^{\prime}\backslash \{2l-1, 2l\}})=\mathrm{col}(\boldsymbol{H}^{M^{\prime}}).
    \label{eq:rem:acyclic-j-acyclic-quasi}
\end{align}
Now, from \eqref{eq:rem:acyclic-l-quasi-0}, \eqref{eq:rem:acyclic-l-quasi} and \eqref{eq:rem:acyclic-j-acyclic}, we have
\begin{align}
    \mathrm{col}(\boldsymbol{H}^{\{2l-1, 2l\}}) 
    \subseteq \mathrm{col}(\boldsymbol{H}^{M^{\prime}})&=\mathrm{col}(\boldsymbol{H}^{\{2j-1, 2j\}\cup M^{\prime}\backslash \{2l-1, 2l\}})
    \nonumber
    \\
    &\subseteq \mathrm{col}(\boldsymbol{H}^{B_{2l-1}}),
    \nonumber
\end{align}
which contradicts the decoding condition in \eqref{eq:dec-cond} for user $u_{2l-1}$ as $\mathrm{col}(\boldsymbol{H}^{\{2l-1\}}) \subseteq \mathrm{col}(\boldsymbol{H}^{B_{2l-1}})$. Therefore, we must have $M^{\prime}=M$.

\subsection{Proof of Lemma \ref{lem:quasi-col(9)}}
Since $\mathrm{col}(\boldsymbol{H}^{\{17,18\}})\subseteq \mathrm{col}(\boldsymbol{H}^{\{1,2,15,16\}})$, we have 
\begin{align}
    \boldsymbol{H}^{\{17,18\}}=\boldsymbol{H}^{\{1,2\}}\boldsymbol{N}_{9,1}+\boldsymbol{H}^{\{15,16\}}\boldsymbol{N}_{9,8}.
    \label{eq:lem:last-conversion-9-1-quasi}
\end{align}
Moreover, since set $\{3,4,5,6,7,8,15,16\}$ is a quasi-circuit set, we must have
\begin{align}
    \boldsymbol{H}^{\{15,16\}}=\boldsymbol{H}^{\{3,4\}}\boldsymbol{N}_{8,2}+\boldsymbol{H}^{\{5,6\}}\boldsymbol{N}_{8,3}+\boldsymbol{H}^{\{7,8\}}\boldsymbol{N}_{8,4},
    \label{eq:lem:last-conversion-9-2-quasi}
\end{align}
where each $\boldsymbol{N}_{8,2},\boldsymbol{N}_{8,3},\boldsymbol{N}_{8,4}$ is invertible. Thus, based on \eqref{eq:lem:last-conversion-9-1-quasi} and \eqref{eq:lem:last-conversion-9-2-quasi}, $\boldsymbol{H}^{\{17,18\}}$ is equal to
\begin{align}
    &\boldsymbol{H}^{\{1,2\}}\boldsymbol{N}_{9,1}+(\boldsymbol{H}^{\{3,4\}}\boldsymbol{N}_{8,2}+\boldsymbol{H}^{\{5,6\}}\boldsymbol{N}_{8,3}+\boldsymbol{H}^{\{7,8\}}\boldsymbol{N}_{8,4})\boldsymbol{N}_{9,8}
    \nonumber
    \\
    &=\boldsymbol{H}^{\{1,2\}}\boldsymbol{N}_{9,1}+\boldsymbol{H}^{\{3,4\}}\boldsymbol{N}_{8,2}^{\prime}+\boldsymbol{H}^{\{5,6\}}\boldsymbol{N}_{8,3}^{\prime}+\boldsymbol{H}^{\{7,8\}}\boldsymbol{N}_{8,4}^{\prime},
    \label{eq:lem3:1-quasi}
\end{align}
where $\boldsymbol{N}_{8,i}^{\prime}=\boldsymbol{N}_{8,i}\boldsymbol{N}_{9,8}, i=2,3,4$. On the other hand, since $\{1,2,5,6,7,8,13,14\}$ is a quasi-circuit set, we get 
\begin{align}
    \boldsymbol{H}^{\{13,14\}}=\boldsymbol{H}^{\{1,2\}}\boldsymbol{N}_{7,1}+\boldsymbol{H}^{\{5,6\}}\boldsymbol{N}_{7,3}+\boldsymbol{H}^{\{7,8\}}\boldsymbol{N}_{7,4},
    \label{eq:lem3:2-quasi}
\end{align}
where each $\boldsymbol{N}_{7,1},\boldsymbol{N}_{7,3},\boldsymbol{N}_{7,4}$ is invertible. Now, for set $\{3,4,13,14,17,18\}$, we have
\begin{align}
    \mathrm{rank}(\boldsymbol{H}^{\{3,4,13,14,17,18\}})
    &=\mathrm{rank}([\boldsymbol{H}^{\{3,4\}}|\boldsymbol{H}^{\{13,14\}}|\boldsymbol{H}^{\{17,18\}}])
    \nonumber
    \\
    &=\mathrm{rank}([\boldsymbol{H}^{\{3,4\}}|\boldsymbol{H}^{\{13,14\}}|\boldsymbol{H}^{\{17,18\}}
    \nonumber
    \\
    & \ \ \ \ -\boldsymbol{H}^{\{3,4\}}\boldsymbol{N}_{8,2}^{\prime}])
    \nonumber
    \\
    &=2t +\mathrm{rank}([\boldsymbol{H}^{\{13,14\}}|\boldsymbol{H}^{\{17,18\}}
    \nonumber
    \\
    & \ \ \ \
    -\boldsymbol{H}^{\{3,4\}}\boldsymbol{N}_{8,2}^{\prime}]).
    \label{eq:lem3:3-quasi}
\end{align}
Now,
since $\mathrm{col}(\boldsymbol{H}^{\{17,18\}})$ must be a subspace of $\mathrm{col}(\boldsymbol{H}^{\{3,4,13,14,17,18\}})$, we must have $\mathrm{rank}(\boldsymbol{H}^{\{3,4,13,14,17,18\}})=4t$. Thus, in \eqref{eq:lem3:3-quasi}, $\boldsymbol{H}^{\{17,18\}}-\boldsymbol{H}^{\{3,4\}}\boldsymbol{N}_{8,2}^{\prime}$ must be linearly dependent on $\boldsymbol{H}^{\{13,14\}}$, which based on \eqref{eq:lem3:1-quasi} and \eqref{eq:lem3:2-quasi} requires each $\boldsymbol{N}_{9,1}, \boldsymbol{N}_{8,3}^{\prime}, \boldsymbol{N}_{8,4}^{\prime}$ to be invertible. 
Besides, each $\boldsymbol{N}_{8,3}^{\prime}, \boldsymbol{N}_{8,4}^{\prime}$ is invertible only if $\boldsymbol{N}_{9,8}$ is invertible. Thus, since both $\boldsymbol{N}_{9,1}$ and $\boldsymbol{N}_{9,8}$ are invertible, set $\{1,2,15,16,17,18\}$ forms a quasi-circuit set of $\boldsymbol{H}$. Similarly, it can be shown that all sets $\{3,4,13,14,15,16\}$, $\{5,6,11,12,17,18\}$ and $\{7,8,9,10,17,18\}$ are quasi-circuit sets.

\subsection{Proof of Lemma \ref{lem:two-interference-dimention}}
Assume $M^{\prime}=\{i_{1}, \dots, i_{|M^{\prime}|}\}$. Then, applying the decoding condition in \ref{rem:dec-con} for $i_{1}, \dots, i_{|M^{\prime}|}$, will result in
\begin{align}
    \mathrm{rank}(\boldsymbol{H}^{M})
    &=t+\mathrm{rank}(\boldsymbol{H}^{M\backslash\{i_{1}\}}),
    \nonumber
    \\
    &=\dots,
    \nonumber
    \\
    &=|M^{\prime}|t+\mathrm{rank}(\boldsymbol{H}^{M\backslash\{i_{1},\dots, i_{|M^{\prime}|}\}}),
    \nonumber
    \\
    &=|M^{\prime}|t+\mathrm{rank}(\boldsymbol{H}^{M\backslash M^{\prime}}),
    \nonumber
\end{align}
which completes the proof.

\section{Proof of Proposition \ref{prop-thm3-2}} \label{app:proof-Prop-thm3-2}

We prove that $N_{0}=[8]$ is a basis set of $\boldsymbol{H}$, each set $N_{i}, i\in [8]$ in \eqref{eq:-N3-matroid} will be a quasi-circuit set of $\boldsymbol{H}$, and $\mathrm{rank}(\boldsymbol{H}^{[9:18]})\geq 7$. The proof is described as follows.
\begin{itemize}
    \item First, since $\beta_{\text{MAIS}}(\mathcal{I}_{3})=8$, we must have $\mathrm{rank}(\boldsymbol{H})=8t$. Now, from $B_{i}, i\in [8]$ in \eqref{I_3-B_i}, it can be seen that set $[8]$ is an independent set of $\mathcal{I}_{3}$, so according to Lemma \ref{lem:MAIS1}, set $[8]$ is an independent set of $\boldsymbol{H}$. Moreover, since $\mathrm{rank}(\boldsymbol{H})=8t$, set $N_{0}=[8]$ will be a basis set of $\boldsymbol{H}$. Now, in order to have $\mathrm{rank} (\boldsymbol{H})=8t$, for all $j\in [29]\backslash [8]$, we must have $\mathrm{col}(\boldsymbol{H}^{\{2j-1,2j\}})\subseteq \mathrm{col}(\boldsymbol{H}^{[8]})$.
    \item According to Lemma \ref{lem:MAIS2}, from $B_{i}, i\in [8]$, it can be seen that:
    \begin{itemize}
        \item for each $j\in \{19,27,35,43\}$,
         \begin{align}
         j\in B_{i}, i\in [8]\backslash \{1\} \rightarrow \mathrm{col}(\boldsymbol{H}^{\{j\}})=\mathrm{col}(\boldsymbol{H}^{\{1\}}),
         \label{eq:prop4-third-1-repeat}
        \end{align}
       \item for each $j\in \{20,28,36,44\}$,
        \begin{align}
        j\in B_{i}, i\in [8]\backslash \{2\}\rightarrow \mathrm{col}(\boldsymbol{H}^{\{j\}})=\mathrm{col}(\boldsymbol{H}^{\{2\}}),
        \label{eq:prop4-third-2-repeat}
       \end{align}
      \item for each $j\in \{21,29,37,45\}$,
       \begin{align}
       j\in B_{i}, i\in [8]\backslash \{3\}\rightarrow \mathrm{col}(\boldsymbol{H}^{\{j\}})=\mathrm{col}(\boldsymbol{H}^{\{3\}}),
       \label{eq:prop4-third-3-repeat}
       \end{align}
      \item for each $j\in \{22,30,38,46\}$,
       \begin{align}
       j\in B_{i}, i\in [8]\backslash \{4\}\rightarrow \mathrm{col}(\boldsymbol{H}^{\{j\}})=\mathrm{col}(\boldsymbol{H}^{\{4\}}).
       \label{eq:prop4-third-4-repeat}
       \end{align}
        \item for each $j\in \{23,31,39,47\}$,
         \begin{align}
         j\in B_{i}, i\in [8]\backslash \{5\} \rightarrow \mathrm{col}(\boldsymbol{H}^{\{j\}})=\mathrm{col}(\boldsymbol{H}^{\{5\}}),
         \label{eq:prop4-third-5-repeat}
        \end{align}
       \item for each $j\in \{24,32,40,48\}$,
        \begin{align}
        j\in B_{i}, i\in [8]\backslash \{6\}\rightarrow \mathrm{col}(\boldsymbol{H}^{\{j\}})=\mathrm{col}(\boldsymbol{H}^{\{6\}}),
        \label{eq:prop4-third-6-repeat}
       \end{align}
      \item for each $j\in \{25,33,41,49\}$,
       \begin{align}
       j\in B_{i}, i\in [8]\backslash \{7\}\rightarrow \mathrm{col}(\boldsymbol{H}^{\{j\}})=\mathrm{col}(\boldsymbol{H}^{\{7\}}),
       \label{eq:prop4-third-7-repeat}
       \end{align}
      \item for each $j\in \{26,34,42,50\}$,
       \begin{align}
       j\in B_{i}, i\in [8]\backslash \{8\}\rightarrow \mathrm{col}(\boldsymbol{H}^{\{j\}})=\mathrm{col}(\boldsymbol{H}^{\{8\}}).
       \label{eq:prop4-third-8-repeat}
       \end{align}
    \end{itemize}
    Let 
    \begin{align}
        M_{1}=\{19,20,21,22,23,24\},
        \nonumber
        \\
        M_{2}=\{27,28,29,30,33,34\},
        \nonumber
        \\
        M_{3}=\{35,36,39,40,41,42\},
        \nonumber
        \\
        M_{4}=\{45,46,47,48,49,50\}.
        \nonumber
    \end{align}
    From \eqref{eq:prop4-third-1-repeat}-\eqref{eq:prop4-third-4-repeat}, it can be seen that 
    \begin{align}
        \mathrm{col}(\boldsymbol{H}^{M_{1}})= \mathrm{col}(\boldsymbol{H}^{[8]\backslash \{7,8\}}), 
        \label{eq:prop4-third-1-new-repeat}
        \\
        \mathrm{col}(\boldsymbol{H}^{M_{2}})= \mathrm{col}(\boldsymbol{H}^{[8]\backslash \{5,6\}}), 
        \label{eq:prop4-third-2-new-repeat}
        \\
        \mathrm{col}(\boldsymbol{H}^{M_{3}})= \mathrm{col}(\boldsymbol{H}^{[8]\backslash \{3,4\}}), 
        \label{eq:prop4-third-3-new-repeat}
        \\
        \mathrm{col}(\boldsymbol{H}^{M_{4}})= \mathrm{col}(\boldsymbol{H}^{[8]\backslash \{1,2\}}).
        \label{eq:prop4-third-4-new-repeat}
    \end{align}
    Thus, each set $M_{1}, M_{2}, M_{3}$ and $M_{4}$ is an independent set of $\boldsymbol{H}$.
    \item To have $\mathrm{rank}(\boldsymbol{H})=8t$, one must have $\mathrm{rank}(\boldsymbol{H}^{B_{i}})=7t, i\in [58]$. Now, since set $[8]$ is a basis set, then from $B_{i}, i\in [8]$ we must have
    \begin{align}
    B_{7},B_{8}&\rightarrow \mathrm{col}(\boldsymbol{H}^{\{9,10\}})\ \subseteq \mathrm{col}(\boldsymbol{H}^{[8]\backslash \{7,8\}}) \stackrel{\eqref{eq:prop4-third-1-new-repeat}}{=}\mathrm{col}(\boldsymbol{H}^{M_{1}}),
    \label{eq:prop4-second-1-repeat}
    \\
    B_{5},B_{6}&\rightarrow \mathrm{col}(\boldsymbol{H}^{\{11,12\}})\subseteq \mathrm{col}(\boldsymbol{H}^{[8]\backslash \{5,6\}})\stackrel{\eqref{eq:prop4-third-2-new-repeat}}{=}\mathrm{col}(\boldsymbol{H}^{M_{2}}),
    \label{eq:prop4-second-2-repeat}
    \\
    B_{3},B_{4}&\rightarrow \mathrm{col}(\boldsymbol{H}^{\{13,14\}})\subseteq \mathrm{col}(\boldsymbol{H}^{[8]\backslash \{3,4\}})\stackrel{\eqref{eq:prop4-third-3-new-repeat}}{=}\mathrm{col}(\boldsymbol{H}^{M_{3}}),
    \label{eq:prop4-second-3-repeat}
    \\
    B_{1},B_{2}&\rightarrow \mathrm{col}(\boldsymbol{H}^{\{15,16\}})\subseteq \mathrm{col}(\boldsymbol{H}^{[8]\backslash \{1,2\}})\stackrel{\eqref{eq:prop4-third-4-new-repeat}}{=}\mathrm{col}(\boldsymbol{H}^{M_{4}}).
    \label{eq:prop4-second-4-repeat}
  \end{align}
    \item From $B_{i}, i\in M_{1}, M_{2}, M_{3}$ and $M_{4}$, it can be verified that
    \begin{align}
        &M_{1}\ \text{is a quasi-minimal cyclic set}\ \& \ \{9,10\}\in B_{i}, i\in M_{1},
        \label{eq:prop4-forth-1-repeat}
        \\
        &M_{2}\ \text{is a quasi-minimal cyclic set}\ \& \ \{11,12\}\in B_{i}, i\in M_{2},
        \label{eq:prop4-forth-2-repeat}
        \\
        &M_{3}\ \text{is a quasi-minimal cyclic set}\ \& \ \{13,14\}\in B_{i}, i\in M_{3},
        \label{eq:prop4-forth-3-repeat}
        \\
        &M_{4}\ \text{is a quasi-minimal cyclic set}\ \& \ \{15,16\}\in B_{i}, i\in M_{4}.
        \label{eq:prop4-forth-4-repeat}
    \end{align}
     \item 
     Now, all the four conditions in Lemma \ref{lem:quasi-independent-cycle} are satisfied for set $M_{1}$ with $j=5$, set $M_{2}$ with $j=6$, set $M_{3}$ with $j=7$, and set $M_{4}$ with $j=8$. Thus, according to Lemma \ref{lem:independent-cycle}, each set $\{9,10\}\cup M_{1}, \{11,12\}\cup M_{2}, \{13,14\}\cup M_{3}$ and $\{15,16\}\cup M_{4}$ will be a quasi-circuit set of $\boldsymbol{H}$. Now, according to \eqref{eq:prop4-third-1-new-repeat}-\eqref{eq:prop4-third-4-new-repeat}, each set 
     \begin{align}
         N_{1} &=\{1,2,3,4,5,6,9,10\},
         \nonumber
         \\
         N_{2}&=\{1,2,3,4,7,8,11,12\},
         \nonumber
         \\
         N_{3}&=\{1,2,5,6,7,8,13,14\},
         \nonumber
         \\
         N_{4}&=\{3,4,5,6,7,8,15,16\},
         \nonumber
     \end{align}
      will also form a quasi-circuit set. 
     \item Since
     \begin{align}
         &\{7,8,9,10,17,18,31,32,51,52\}\ \backslash \{i\}\subseteq B_{i}, i\in \{51,52\},
         \label{eq:prop4-forth-1-new-0-1}
         \\
         &\{5,6,11,12,17,18,37,38,53,54\}\backslash \{i\}\subseteq B_{i}, i\in \{53,54\},
         \label{eq:prop4-forth-1-new-0-2}
         \\
         &\{3,4,13,14,17,18,43,44,55,56\}\backslash \{i\}\subseteq B_{i}, i\in \{55,56\},
         \label{eq:prop4-forth-1-new-0-3}
         \\
         &\{1,2,15,16,17,18,25,26,57,58\}\backslash \{i\}\subseteq B_{i}, i\in \{57,58\},
         \label{eq:prop4-forth-1-new-0-4}
     \end{align}
     based on Lemma \ref{lem:two-interference-dimention}, we must have 
     \begin{align}
         \eqref{eq:prop4-forth-1-new-0-1} \rightarrow \mathrm{rank} ( \boldsymbol{H}^{\{7,8,9,10,17,18,31,32\}})\ &\leq6t,
         \label{eq:prop4-forth-1-new-1-1}
         \\
         \eqref{eq:prop4-forth-1-new-0-2} \rightarrow \mathrm{rank} ( \boldsymbol{H}^{\{5,6,11,12,17,18,37,38\}})&\leq6t,
         \label{eq:prop4-forth-1-new-1-2}
         \\
         \eqref{eq:prop4-forth-1-new-0-3} \rightarrow \mathrm{rank} ( \boldsymbol{H}^{\{3,4,13,14,17,18,43,44\}})&\leq6t,
         \label{eq:prop4-forth-1-new-1-3}
         \\
         \eqref{eq:prop4-forth-1-new-0-4} \rightarrow \mathrm{rank} ( \boldsymbol{H}^{\{1,2,15,16,17,18,25,26\}})&\leq6t.
         \label{eq:prop4-forth-1-new-1-4}
     \end{align}
     Now, since 
     \begin{align}
         &\{7,8,9,10,17,18,31,32\}\ \backslash \{i\}\subseteq B_{i}, i\in \{31,32\},
         \label{eq:prop4-forth-1-new-1-5}
         \\
         &\{5,6,11,12,17,18,37,38\}\backslash \{i\}\subseteq B_{i}, i\in \{37,38\},
         \label{eq:prop4-forth-1-new-1-6}
         \\
         &\{3,4,13,14,17,18,43,44\}\backslash \{i\}\subseteq B_{i}, i\in \{43,44\},
         \label{eq:prop4-forth-1-new-1-7}
         \\
         &\{1,2,15,16,17,18,25,26\}\backslash \{i\}\subseteq B_{i}, i\in \{25,26\},
         \label{eq:prop4-forth-1-new-1-8}
     \end{align}
     based on Lemma \ref{lem:two-interference-dimention}, we must have
     \begin{align}
         \eqref{eq:prop4-forth-1-new-1-1},\eqref{eq:prop4-forth-1-new-1-5}\rightarrow \mathrm{rank} (\boldsymbol{H}^{\{7,8,9,10,17,18\}})\ &\leq4t,
         \label{eq:prop4-forth-1-new-1-9}
         \\
         \eqref{eq:prop4-forth-1-new-1-2},\eqref{eq:prop4-forth-1-new-1-6}\rightarrow\mathrm{rank} ( \boldsymbol{H}^{\{5,6,11,12,17,18\}})&\leq4t,
         \label{eq:prop4-forth-1-new-1-10}
         \\
         \eqref{eq:prop4-forth-1-new-1-3},\eqref{eq:prop4-forth-1-new-1-7}\rightarrow \mathrm{rank} ( \boldsymbol{H}^{\{3,4,13,14,17,18\}})&\leq4t,
         \label{eq:prop4-forth-1-new-1-11}
         \\
         \eqref{eq:prop4-forth-1-new-1-4},\eqref{eq:prop4-forth-1-new-1-8}\rightarrow \mathrm{rank} ( \boldsymbol{H}^{\{1,2,15,16,17,18\}})&\leq4t.
         \label{eq:prop4-forth-1-new-1-12}
     \end{align}
     Thus,
     \begin{align}
         \eqref{eq:prop4-forth-1-new-1-9}\rightarrow \mathrm{col}(\boldsymbol{H}^{\{17,18\}})&\subseteq \mathrm{col}(\boldsymbol{H}^{\{7,8,9,10\}}).
         \nonumber
         \\
         \eqref{eq:prop4-forth-1-new-1-10}\rightarrow \mathrm{col}(\boldsymbol{H}^{\{17,18\}})&\subseteq \mathrm{col}(\boldsymbol{H}^{\{5,6,11,12\}}),
         \nonumber
         \\
         \eqref{eq:prop4-forth-1-new-1-11}\rightarrow \mathrm{col}(\boldsymbol{H}^{\{17,18\}})&\subseteq \mathrm{col}(\boldsymbol{H}^{\{3,4,13,14\}}),
         \nonumber
         \\
         \eqref{eq:prop4-forth-1-new-1-12}\rightarrow \mathrm{col}(\boldsymbol{H}^{\{17,18\}})&\subseteq \mathrm{col}(\boldsymbol{H}^{\{1,2,15,16\}}),
         \nonumber
     \end{align}
     Hence, based on Lemma \ref{lem:quasi-col(9)}, each set
     \begin{align}
         N_{5}&=\{1,2,15,16,17,18\},
         \nonumber
         \\
         N_{6}&=\{3,4,13,14,17,18\},
         \nonumber
         \\
         N_{7}&=\{5,6,11,12,17,18\},
         \nonumber
         \\
         N_{8}&=\{7,8,9,10,17,18\},
         \nonumber
     \end{align}
     is a quasi-circuit set.
     \item Finally, from $B_{i}, i\in [9:16]$, it can be seen that 
     \begin{align}
         [9:16]\backslash \{i\}\subseteq B_{i}, i\in \{9,11,13,15\}.
     \end{align}
     Thus, according to Lemma \ref{lem:two-interference-dimention}, we get
     \begin{align}
         \mathrm{rank}(\boldsymbol{H}^{\{10,12,14,16\}})=\mathrm{rank}(\boldsymbol{H}^{[9:16]})-4t,
         \label{eq:final-quasi-I3-1}
     \end{align}
     On the other hand, it can be observed that set $\{10,12,14,16\}$ is a minimal cyclic set of $\mathcal{I}_{3}$. Thus, according to Proposition \ref{prop:rate-acyclic-minimal}, we have
     \begin{align}
         \mathrm{rank}(\boldsymbol{H}^{\{10,12,14,16\}})\geq 3t,
         \label{eq:final-quasi-I3-2}
     \end{align}
     Now, \eqref{eq:final-quasi-I3-1} and \eqref{eq:final-quasi-I3-2} will result in
     \begin{align}
          \mathrm{rank}(\boldsymbol{H}^{N_{9}=[9:16]})\geq 7t,
     \end{align}
     which completes the proof.
\end{itemize}

\section{Proof of Lemmas \ref{lem:g-function} and \ref{lem2:g-function}} \label{app:proof-lem-g-function}

\subsection{Proof of Lemma \ref{lem:g-function}} 
Let 
\begin{align}
    a_{1} =\ & g(x_{i},x_{j},x_{l},x_{w}) + g(x_{i},x_{j},x_{v},x_{w})\ +
    \nonumber
    \\
    & g(x_{i},x_{l},x_{v},x_{w}) + g(x_{j},x_{l},x_{v},x_{w}),
    \nonumber
    \\
     \nonumber
    \\
    a_{2} =\ & (x_{i}+x_{j}+x_{l})(x_{i}+x_{j}+x_{v})(x_{i}+x_{l}+x_{v}),
    \nonumber
    \\
     \nonumber
    \\
    a_{3} =\ & 2(x_{i}+x_{j}+x_{l})(x_{i}+x_{j}+x_{l})(2x_{i}+x_{j}+x_{l}+2x_{v}) +
    \nonumber
    \\
     & 2(x_{i}+x_{j}+x_{v})(x_{i}+x_{j}+x_{v})(2x_{i}+x_{j}+x_{v}+2x_{l}) +
    \nonumber
    \\
     & 2(x_{i}+x_{l}+x_{v})(x_{i}+x_{l}+x_{v})(2x_{i}+x_{l}+x_{v}+2x_{j}),
     \nonumber
     \\
      \nonumber
    \\
     a_{4} =\ & 2(x_{i}+x_{j}+x_{l})(x_{i}+x_{j}+x_{l}) +
    \nonumber
    \\
    & 2(x_{i}+x_{j}+x_{v})(x_{i}+x_{j}+x_{v}) +
    \nonumber
    \\
    & 2(x_{i}+x_{l}+x_{v})(x_{i}+x_{l}+x_{v}) +
    \nonumber
    \\
    & 2(x_{j}+x_{l}+x_{v})(x_{j}+x_{l}+x_{v}).
\end{align}
It can be verified that
\begin{align}
    a_{1} =\ &
    x_{i}x_{i}(x_{j}+x_{l}+x_{v}) + x_{j}x_{j}(x_{i}+x_{l}+x_{v}) +
    \nonumber
    \\
    & x_{l}x_{l}(x_{i}+x_{j}+x_{v}) + x_{v}x_{v}(x_{i}+x_{j}+x_{l}) +
    \nonumber
    \\
    & (x_{i}x_{j}+x_{i}x_{l}+x_{i}x_{v}+x_{j}x_{l}+x_{j}x_{v}+x_{l}x_{v})(1+2x_{w}) +
    \nonumber
    \\
    & x_{i}x_{j}x_{l} + x_{i}x_{j}x_{v} + x_{i}x_{l}x_{v} + x_{j}x_{l}x_{v},
    \label{lem-g-function-1}
\end{align}
\begin{align}
    a_{2} =\ &  x_{i}x_{i}x_{i} + 2x_{i}x_{i}(x_{j}+x_{l}+x_{v}) + x_{j}x_{j}(x_{i}+x_{l}+x_{v}) +
    \nonumber
    \\
    & x_{l}x_{l}(x_{i}+x_{j}+x_{v}) + x_{v}x_{v}(x_{i}+x_{j}+x_{l}) + 2x_{j}x_{l}x_{v},
    \label{lem-g-function-2}
\end{align}
\begin{align}
    a_{3} =\ & x_{j}x_{j}x_{j} + x_{l}x_{l}x_{l} + x_{v}x_{v}x_{v} + x_{j}x_{j}(x_{i}+x_{l}+x_{v}) +
    \nonumber
    \\
    & x_{l}x_{l}(x_{i}+x_{j}+x_{v}) + x_{v}x_{v}(x_{i}+x_{j}+x_{l}) +
    \nonumber
    \\
    & 2(x_{i}x_{j}x_{l} + x_{i}x_{j}x_{v} + x_{i}x_{l}x_{v}),
    \label{lem-g-function-3}
\end{align}
\begin{align}
    a_{4} =\ & 2(x_{i}x_{j}+x_{i}x_{l}+x_{i}x_{v}+x_{j}x_{l}+x_{j}x_{v}+x_{l}x_{v}).
    \label{lem-g-function-4}
\end{align}
Now, it can be seen that
\begin{align}
    &a_{1} + a_{2} + a_{3} + a_{4}(1+2x_{w}) =
    \nonumber
    \\
    & x_{i}x_{i}x_{i} + x_{j}x_{j}x_{j} + x_{l}x_{l}x_{l} + x_{v}x_{v}x_{v} =
    \nonumber
    \\
    & x_{i} + x_{j} + x_{l} + x_{v},
    \label{lem-g-function-5}
\end{align}
where \eqref{lem-g-function-5} follows from the fact that for any $x_{i}\in GF(3)$, we have $x_{i}x_{i}x_{i}=x_{i}$.
Thus, by having $x_{i}+x_{j}+x_{l}$, $x_{i}+x_{j}+x_{v}$, $x_{i}+x_{l}+x_{v}$, and $x_{j}+x_{l}+x_{v}$, each $x_{i}, x_{j}, x_{l}$ and $x_{v}$ is decodable from \eqref{lem-g-function-5}. This completes the proof.

\subsection{Proof of Lemma \ref{lem2:g-function}} 

It can be verified that function $g(x_{i},x_{j},x_{v},x_{w})$ can be rewritten as follows
\begin{align}
    g(x_{i},x_{j},x_{v},x_{w})=& 2x_{i}x_{i}(x_{j}+x_{v}+x_{w}) +
    \nonumber
    \\
    & 2x_{j}x_{j}(x_{i}+x_{v}+x_{w}) +
    \nonumber
    \\
    & 2(x_{v}x_{v}+x_{w}x_{w})(x_{i}+x_{j})+
    \nonumber
    \\
    & 2(x_{v}x_{v}x_{w}+x_{w}x_{w}x_{v}) +
    \nonumber
    \\
    & 2x_{i}x_{j} + 2(x_{i}+x_{j})(x_{v}+x_{w}) + 2x_{v}x_{w} +
    \nonumber
    \\
    & x_{i}x_{j}(x_{v}+x_{w}) + x_{v}x_{w}(x_{i}+x_{j}).
    \label{lem2-g-function-1}
\end{align}
It can be seen that
\begin{align}
    &2(x_{v}x_{v}+x_{w}x_{w})(x_{i}+x_{j}) + x_{v}x_{w}(x_{i}+x_{j}) =
    \nonumber
    \\
    &2(x_{v}x_{v}+x_{w}x_{w}+2x_{v}x_{w})(x_{i}+x_{j})=
    \nonumber
    \\
    &2(x_{v}+x_{w})(x_{v}+x_{w})(x_{i}+x_{j}).
    \label{lem2-g-function-2}
\end{align}
Since terms $x_{i},x_{j}$ and $x_{v}+x_{w}$ are known, from \eqref{lem2-g-function-1} and \eqref{lem2-g-function-2}, it can be observed that only term $2(x_{v}x_{v}x_{w}+x_{w}x_{w}x_{v}+x_{v}x_{w})$ is unknown.
\\
Similarly, in $g(x_{i},x_{l},x_{v},x_{w})$, since the terms $x_{i},x_{l}$ and $x_{v}+x_{w}$ are known, only term $2(x_{v}x_{v}x_{w}+x_{w}x_{w}x_{v}+x_{v}x_{w})$ is unknown. 
\\
Thus, in $g(x_{i},x_{j},x_{v},x_{w})+2g(x_{i},x_{l},x_{v},x_{w})$, the unknown terms $2(x_{v}x_{v}x_{w}+x_{w}x_{w}x_{v}+x_{v}x_{w})$ is canceled out. Hence, using the value of $x_{i},x_{j},x_{l}$ and $x_{v}+x_{w}$, the value of $g(x_{i},x_{j},x_{v},x_{w})+2g(x_{i},x_{l},x_{v},x_{w})$ will be found.

\IEEEpeerreviewmaketitle


\bibliographystyle{IEEEtran}
	\bibliography{References}

\end{document}